\documentclass[aps,pra,onecolumn,superscriptaddress]{revtex4-1}
\usepackage{amsmath,amssymb,graphicx,bm}
\usepackage{adjustbox}
\usepackage[table]{xcolor}
\usepackage{subcaption}
\usepackage{booktabs}
\usepackage{multirow}
\usepackage{float}
\usepackage{placeins}
\usepackage{sansmath}
\usepackage{hyperref}
\usepackage{braket}
\usepackage{tikz}
\usetikzlibrary{trees,positioning,calc,arrows.meta,shapes.geometric}
\usepackage{tikz-cd}
\usepackage[qm]{qcircuit}

\DeclareMathOperator{\atantwo}{atan2}

\newcommand{\Active}{\textsf{A}}
\newcommand{\Inactive}{\textsf{I}}

\newtheorem{theorem}{Theorem}

\newtheorem{proposition}{Proposition}
\newtheorem{lemma}{Lemma}
\newtheorem{definition}{Definition}
\newtheorem{corollary}{Corollary}
\newcounter{algorithm}

\newcommand{\sref}[1]{Sec.~\ref{#1}}

\providecommand{\Achip}[1]{\fcolorbox{yellow!70!black}{yellow!55}{\rule{0pt}{1.0ex}\,$\scriptstyle #1$\,}}
\providecommand{\Fchip}[1]{\fcolorbox{red!65!black}{red!40}{\rule{0pt}{1.0ex}\,$\scriptstyle #1$\,}}
\providecommand{\Ichip}[1]{\fcolorbox{gray!55}{gray!15}{\rule{0pt}{1.0ex}\,$\scriptstyle #1$\,}}
\providecommand{\TargetColFill}{orange!20}
\providecommand{\PivotRow}{cyan!22}
\providecommand{\AffectedRow}{yellow!22}
\providecommand{\SurvivorRow}{green!25}
\providecommand{\TargetPivot}{orange!42}

\begin{document}

\title{A hardware-efficient variational ansatz with an exact diagonal metric for real- and imaginary-time evolution and Haar sampling}

\author{Dario Picozzi}
\email{picozzi.dario@gmail.com}
\affiliation{Department of Physics and Astronomy, University College London (UCL), Gower Street, London, WC1E 6BT, United Kingdom}
\affiliation{London Centre for Nanotechnology, 19 Gordon St, London, WC1H 0AH, United Kingdom}

\begin{abstract}
Variational quantum algorithms depend on the geometry of their parametrised circuits: metric-aware optimisation and time evolution require the Fubini--Study metric, which has hitherto demanded costly auxiliary measurements and ill-conditioned inversions. This work introduces a hardware-efficient $n$-qubit ansatz, which parametrises states by a binary tree and whose Hilbert-space pullback metric is diagonal in closed form. Quantum natural gradient on the tree parameters, variational imaginary- and real-time evolution, and exact unitary-invariant (Haar) sampling on a symmetry sector run with no auxiliary metric circuits or matrix inversion. When the target state is supported on a subspace of $k$ computational-basis states, the redundant tree parameters carry a gauge freedom a pruning compiler converts into circuits whose two-qubit count provably grows linearly in $k$; a variant reaches near-optimal $O(nk/\log n)$ scaling with the closed-form metric intact. On electronic-structure calculations for small molecules and half-filled Hubbard quench dynamics, the method reaches reference-level accuracy with one to three orders of magnitude fewer two-qubit gates than leading alternatives. Interchangeable constructions (a Schur-transform dressing or internal reparameterisations) make the ansatz exactly spin-adapted, with fixed total spin at every parameter and no penalty terms. The bare ansatz is an exactly controllable, well-conditioned and barren-plateau-free primitive for preparing and sampling sector states: on its own, it is classically simulable in $k$ (a boundary proved for a general class of sector-sparse ans\"atze); composed with a classically hard dressing, it yields molecular ground states, sector-Haar benchmarking, thermal correlators, and exact effective Hamiltonians trained from energy measurements alone, with the composed circuit carrying the potential for quantum advantage.
\end{abstract}

\maketitle

Variational quantum algorithms are widely regarded as the most plausible route to extracting useful computation from intermediate-scale quantum hardware~\cite{Tilly2022, Cerezo2021, McClean2016, McArdle2020}. In their canonical form, a parametric quantum circuit prepares a trial state $\ket{\psi(\bm\lambda)}$ whose parameters $\bm\lambda$ are tuned by a classical optimiser against a problem-defined cost~\cite{Peruzzo2014}, and the same machinery, with imaginary or real time playing the role of an optimiser step, also drives the variational simulation of static and dynamical quantum many-body problems. The geometry of the parametric manifold, captured by the Fubini--Study metric and its compatible K\"ahler structure, controls the efficiency of every step of this loop: metric-aware optimisation and imaginary-time evolution require the metric gradient, while real-time evolution rotates that same metric gradient by the complex structure. For all but a handful of fine-tuned ansätze, however, the metric is dense, must be estimated by $\mathcal{O}(P^2)$ auxiliary Hadamard tests, and is then regularised by a pseudoinverse on the classical side, an overhead that scales with the parameter count $P$ and ultimately limits which problems are reachable in practice~\cite{Stokes2020, McArdle2019, Yuan2019, Gacon2021}.

Here we cut this overhead away at its source by designing an ansatz whose Hilbert-space pullback metric on the tree coordinates is diagonal in closed form, so that the linear system reduces to elementwise division and the auxiliary metric circuits disappear entirely. The construction is a binary tree of uniformly controlled $R_y$ rotations interleaved with diagonal phase gates (geometrically, a system of polyspherical coordinates on the projective Hilbert space $\mathbb{C}\mathbf{P}^{2^n-1}$), which is simultaneously complete (every state can be reached) and hardware-efficient (the unparameterised scaffold is the standard CNOT staircase). Like the standard hardware-efficient ansatz (HEA)~\cite{Kandala2017} it uses only CNOTs and single-qubit rotations on a device-friendly layout, but it is not that ansatz: the HEA is a fixed, problem-agnostic layered template, whereas the tree is a structured cascade dictated by the target subspace, and that difference is the decisive one. The HEA's template determines neither which subspace the circuit reaches nor how expressive it is, so that making it expressive enough to be certain of reaching a chosen target drives it towards a unitary $2$-design and the exponentially vanishing gradients of a barren plateau~\cite{McClean2018, Holmes2022, Sim2019, Leone2023}. The tree instead makes expressibility a controllable resource: fixing a set of $k$ active leaves renders it maximally expressive within, and strictly confined to, exactly the chosen $k$-dimensional subspace, at the minimal parameter count that subspace admits (\sref{sec:pruning}), so that a target symmetry sector is reached exactly and without penalty terms. Because that subspace is only polynomially large, the bare ansatz is moreover free of barren plateaus: under Fubini--Study initialisation its natural-gradient signal equals the quantum variance of the sector-projected Hamiltonian, which is bounded below by an inverse polynomial in $n$ rather than being exponentially small (Appendix~\ref{app:no-bp})~\cite{McClean2018, Holmes2022, Larocca2025}; the exact $2$-design of \sref{sec:fs-apps} makes this a closed-form statement. We show that the metric admits the diagonal closed form via a direct calculation that uses no Hadamard tests; that an exact four-term parameter-shift rule yields analytic Euclidean gradients on the same circuit at four evaluations per parameter; and that, when the target state is sparsely supported on $k\ll 2^n$ computational-basis states, the redundant parameters of the full tree carry a gauge freedom made explicit through an active/fixed/inactive node classification and exploit to compile pruned circuits whose CNOT count provably grows linearly with $k$ (at most $2nA_s=O(n^2k)$ for any ordering, with the tight constant conjectured). In independent coordinates $\bm\lambda$, sampling from the Fubini--Study volume measure
\begin{equation}\label{eq:fs-measure-intro}
d\mu_{\mathrm{FS}} \;\propto\; \sqrt{\det g^{\mathrm{FS}}}\;\, d\bm\lambda
\end{equation}
reproduces the unitary-invariant (Haar) ensemble on a chosen symmetry sector exactly and rejection-free, the diagonal metric reducing the draw to one independent sample per tree node. Three interchangeable constructions (a Schur-transform dressing and two gate-free variants) further make the ansatz exactly spin-adapted, with fixed total spin at every parameter and no penalty terms.

On the resulting compiler the bare ansatz is first benchmarked across static and dynamical problems: symmetry-adapted variational quantum eigensolver (VQE) calculations on H$_3^+$, LiH, BeH$_2$, H$_2$O and NH$_3$, and half-filled $1\!\times\!4$ and $2\!\times\!2$ Hubbard quenches simulated with the exact metric/K\"ahler integrator. The pruned ansatz reaches FCI accuracy with one to two orders of magnitude fewer CNOTs than UCCSD and reproduces exact real-time evolution at a fraction of the depth of pVQD and up to three orders of magnitude fewer two-qubit gates than Trotter. These bare computations are, moreover, classically simulable in the active-leaf count: the prepared state, its energy, its metric and its Fubini--Study sample ensembles all admit a closed-form, polynomial-cost classical description. This dequantization is proved not only for the tree but for any sector-sparse variational family and its classically contractible dressings (Appendix~\ref{app:simulability}), so that the construction becomes genuinely quantum only when the prepared state is composed with a non-classically-simulable dressing $U(\bm\phi)$. This bare/dressed boundary organizes the applications that follow: in the composed role the same diagonal-metric ansatz and the same shift-rule primitive deliver dressed molecular ground states, Fubini--Study process benchmarking and infinite-temperature transport, and exact effective Hamiltonians trained from energy measurements alone. What follows sets out the construction, the compiler and the supporting numerics; Figure~\ref{fig:overview} summarises the end-to-end pipeline.

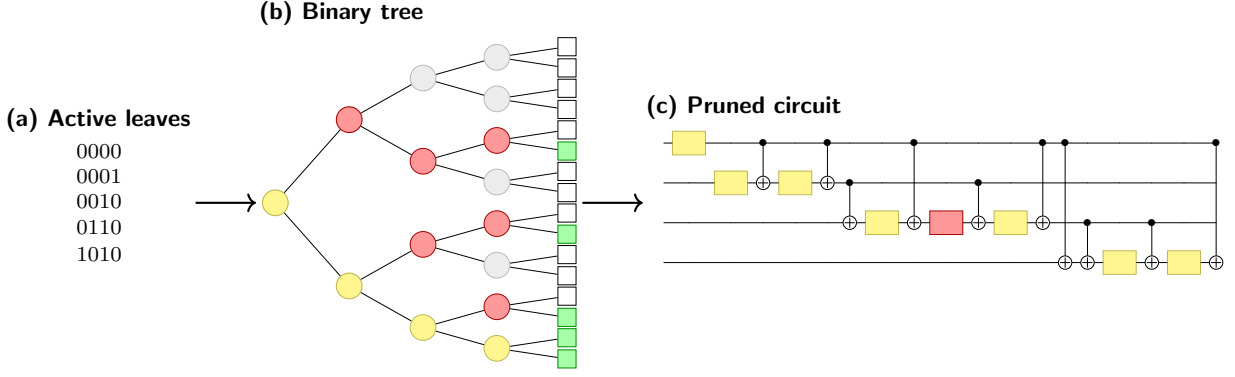
\begin{figure*}[!htbp]
\centering

\providecommand{\PWactC}{yellow!70!black}
\providecommand{\PWactF}{yellow!55}
\providecommand{\PWfixC}{red!65!black}
\providecommand{\PWfixF}{red!40}
\providecommand{\PWinaC}{gray!55}
\providecommand{\PWinaF}{gray!15}
\providecommand{\PWleafC}{green!55!black}
\providecommand{\PWleafF}{green!35}

\providecommand{\PWactB}{\push{\fcolorbox{\PWactC}{\PWactF}{\rule{0pt}{1.4ex}\rule{2.6ex}{0pt}}}\qw}
\providecommand{\PWfixB}{\push{\fcolorbox{\PWfixC}{\PWfixF}{\rule{0pt}{1.4ex}\rule{2.6ex}{0pt}}}\qw}
\providecommand{\PWTARG}{\push{\raisebox{-0.12ex}{\tikz[baseline=-0.6ex]{%
  \draw[fill=white,line width=0.4pt] (0,0) circle (0.85ex);%
  \draw[line width=0.5pt] (-0.65ex,0) -- (0.65ex,0);%
  \draw[line width=0.5pt] (0,-0.65ex) -- (0,0.65ex);}}}\qw}

\begin{tikzpicture}[
  every node/.style={font=\small\sffamily},
  panel/.style={inner sep=2pt, align=center},
  plab/.style={font=\small\sffamily\bfseries, anchor=south west, inner sep=1pt, yshift=1pt},
]

\node[panel, anchor=north west, text width=24mm] (leaves) at (0,0) {%
  {\footnotesize\sffamily $0000$\\$0001$\\$0010$\\$0110$\\$1010$}%
};
\node[plab] at (leaves.north west) {(a)~Active leaves};

\node[panel, right=8mm of leaves, anchor=west] (tree) {%
  \begin{tikzpicture}[
    grow=right,
    level distance=8mm,
    level 1/.style={sibling distance=22mm},
    level 2/.style={sibling distance=11mm},
    level 3/.style={sibling distance=5.5mm},
    level 4/.style={sibling distance=2.8mm},
    int/.style={circle, draw, inner sep=0pt, font=\tiny, minimum size=3.4mm},
    lf/.style={draw, rectangle, minimum size=2.4mm, inner sep=0pt},
    edge from parent/.style={draw, line width=0.35pt}
  ]
  \node[int,fill=\PWactF,draw=\PWactC] {}
    child { node[int,fill=\PWactF,draw=\PWactC] {}
      child { node[int,fill=\PWactF,draw=\PWactC] {}
        child { node[int,fill=\PWactF,draw=\PWactC] {}
          child { node[lf,fill=\PWleafF,draw=\PWleafC] {} }
          child { node[lf,fill=\PWleafF,draw=\PWleafC] {} }
        }
        child { node[int,fill=\PWfixF,draw=\PWfixC] {}
          child { node[lf,fill=\PWleafF,draw=\PWleafC] {} }
          child { node[lf] {} }
        }
      }
      child { node[int,fill=\PWfixF,draw=\PWfixC] {}
        child { node[int,fill=\PWinaF,draw=\PWinaC] {}
          child { node[lf] {} } child { node[lf] {} }
        }
        child { node[int,fill=\PWfixF,draw=\PWfixC] {}
          child { node[lf,fill=\PWleafF,draw=\PWleafC] {} }
          child { node[lf] {} }
        }
      }
    }
    child { node[int,fill=\PWfixF,draw=\PWfixC] {}
      child { node[int,fill=\PWfixF,draw=\PWfixC] {}
        child { node[int,fill=\PWinaF,draw=\PWinaC] {}
          child { node[lf] {} } child { node[lf] {} }
        }
        child { node[int,fill=\PWfixF,draw=\PWfixC] {}
          child { node[lf,fill=\PWleafF,draw=\PWleafC] {} }
          child { node[lf] {} }
        }
      }
      child { node[int,fill=\PWinaF,draw=\PWinaC] {}
        child { node[int,fill=\PWinaF,draw=\PWinaC] {}
          child { node[lf] {} } child { node[lf] {} }
        }
        child { node[int,fill=\PWinaF,draw=\PWinaC] {}
          child { node[lf] {} } child { node[lf] {} }
        }
      }
    };
  \end{tikzpicture}%
};
\node[plab] at (tree.north west) {(b)~Binary tree};

\node[panel, right=8mm of tree, anchor=west, text width=78mm] (pruned) {%
  \resizebox{74mm}{!}{\Qcircuit @C=0.45em @R=0.85em {
    & \PWactB & \qw     & \ctrl{1} & \qw     & \ctrl{1} & \qw     & \qw     & \ctrl{2} & \qw     & \qw     & \qw     & \ctrl{2} & \ctrl{3} & \qw     & \qw     & \qw     & \qw     & \ctrl{3} \\
    & \qw     & \PWactB & \PWTARG  & \PWactB & \PWTARG  & \ctrl{1}& \qw     & \qw      & \qw     & \ctrl{1}& \qw     & \qw      & \qw      & \qw     & \qw     & \qw     & \qw     & \qw \\
    & \qw     & \qw     & \qw      & \qw     & \qw      & \PWTARG & \PWactB & \PWTARG  & \PWfixB & \PWTARG & \PWactB & \PWTARG  & \qw      & \ctrl{1}& \qw     & \ctrl{1}& \qw     & \qw \\
    & \qw     & \qw     & \qw      & \qw     & \qw      & \qw     & \qw     & \qw      & \qw     & \qw     & \qw     & \qw      & \PWTARG  & \PWTARG & \PWactB & \PWTARG & \PWactB & \PWTARG
  }}%
};
\node[plab] at (pruned.north west) {(c)~Pruned circuit};

\draw[->, line width=0.8pt] (leaves.east) -- (tree.west);
\draw[->, line width=0.8pt] (tree.east)   -- (pruned.west);

\end{tikzpicture}
\caption{End-to-end pipeline, illustrated on $n=4$ qubits with $k=|S|=5$ active basis states, $S=\{0000,0001,0010,0110,1010\}$. From the user-supplied active leaves \textbf{(a)} (green nodes in the tree), the compiler builds the binary tree and classifies each internal node as \textsc{Active} (yellow, free $R_y$ angle), \textsc{Fixed} (red, $R_y$ angle pinned to $0$ or $1$ in units of $\pi/2$) or \textsc{Inactive} (grey, gauge direction that does not affect the represented state) \textbf{(b)}. Gauge-fixing the inactive parameters eliminates redundant elementary $R_y$ gates in the cascade and exposes pairwise CNOT cancellations, yielding the pruned circuit \textbf{(c)} (shown stylised; coloured boxes mark \textsc{Active}/\textsc{Fixed} elementary $R_y$ rotations in the same colour convention) with $|\bm\theta_A|=4$ free parameters and $10$ CNOTs, against $15$ free $R_y$ angles and $14$ CNOTs for the unrestricted ansatz on the same $n$. The Hilbert-space pullback metric on the surviving parameters remains diagonal in closed form (see \sref{sec:metric}). The same example is used throughout the paper to detail each step (Fig.~\ref{fig:pruned-example}).}
\label{fig:overview}
\end{figure*}

\section{Methods}\label{sec:methods}

\subsection{Binary-tree ansatz}\label{sec:circuit}\label{sec:polyspherical}

The full ansatz is a cascade of uniformly controlled $R_y$ rotations followed by uniformly controlled $R_z$ rotations: the $R_y$ block prepares all real amplitudes, the $R_z$ block dresses them with the necessary complex phases. Throughout, $\theta_i$ and $\omega_j$ are reserved for the tree parameters carried by the binary tree below (one per internal node and one per leaf, respectively), and use $k=|S|$ for the size of a target set of active computational-basis states (defined below in \sref{sec:pruning}). Each uniformly controlled rotation is realised as a triangular cascade of elementary $R_y$/$R_z$ rotations on the target qubit separated by CNOTs whose controls run in Gray-code order, with the elementary angles obtained from the user-specified angles by a Walsh--Hadamard transform re-indexed by the Gray code~\cite{Mottonen2004, Mottonen2005, Savage1997}; the chart-to-circuit map relating $(\bm\theta,\bm\omega)$ to the elementary rotation angles, and the ranges of $\theta_i$, are recorded in Appendix~\ref{app:c2p-detail}. Throughout, hardware-efficient refers to this fixed CNOT-plus-single-qubit-rotation scaffold~\cite{Kandala2017}; the cascade's CNOTs are not all nearest-neighbour, but long-range controls are geometrically rare, so the construction maps onto a linear qubit chain at constant-factor overhead in the dense limit, and on the symmetry sectors of \sref{sec:vqe} the pruned circuit likewise maps onto the chain at constant-factor two-qubit overhead, realising each uniformly controlled rotation by the nearest-neighbour synthesis of Ref.~\cite{Bergholm2005} (Appendix~\ref{app:cnot-structure}).

The resulting parametrization can be read off a binary tree with $n+1$ layers (Fig.~\ref{fig:classified-tree}). Each leaf corresponds to one of the $2^n$ computational-basis states. Each internal node in the first $n$ layers carries one $R_y$ parameter $\theta_i$, with the two outgoing branches from that node weighted by $\cos\theta_i$ (downward) and $\sin\theta_i$ (upward); each node on the final layer carries one $R_z$ parameter $\omega_j$ that contributes a phase $e^{i\omega_j}$ on its single outgoing edge. Global-phase invariance lets us fix one phase to zero (e.g.\ for $\ket{0}$), giving $2^n - 1$ independent phases.

Given an assignment of parameters, the amplitude of any computational-basis state is the product of the factors along the unique root-to-leaf path. For the example in Fig.~\ref{fig:classified-tree}, $\braket{0110\,|\,\psi} = \cos\theta_0\,\sin\theta_1\,\sin\theta_4\,\cos\theta_{10}$. With the parametrization in hand, we turn to its information geometry, and then exploit the redundancy that arises when the target state is supported only on a small subset of the leaves.
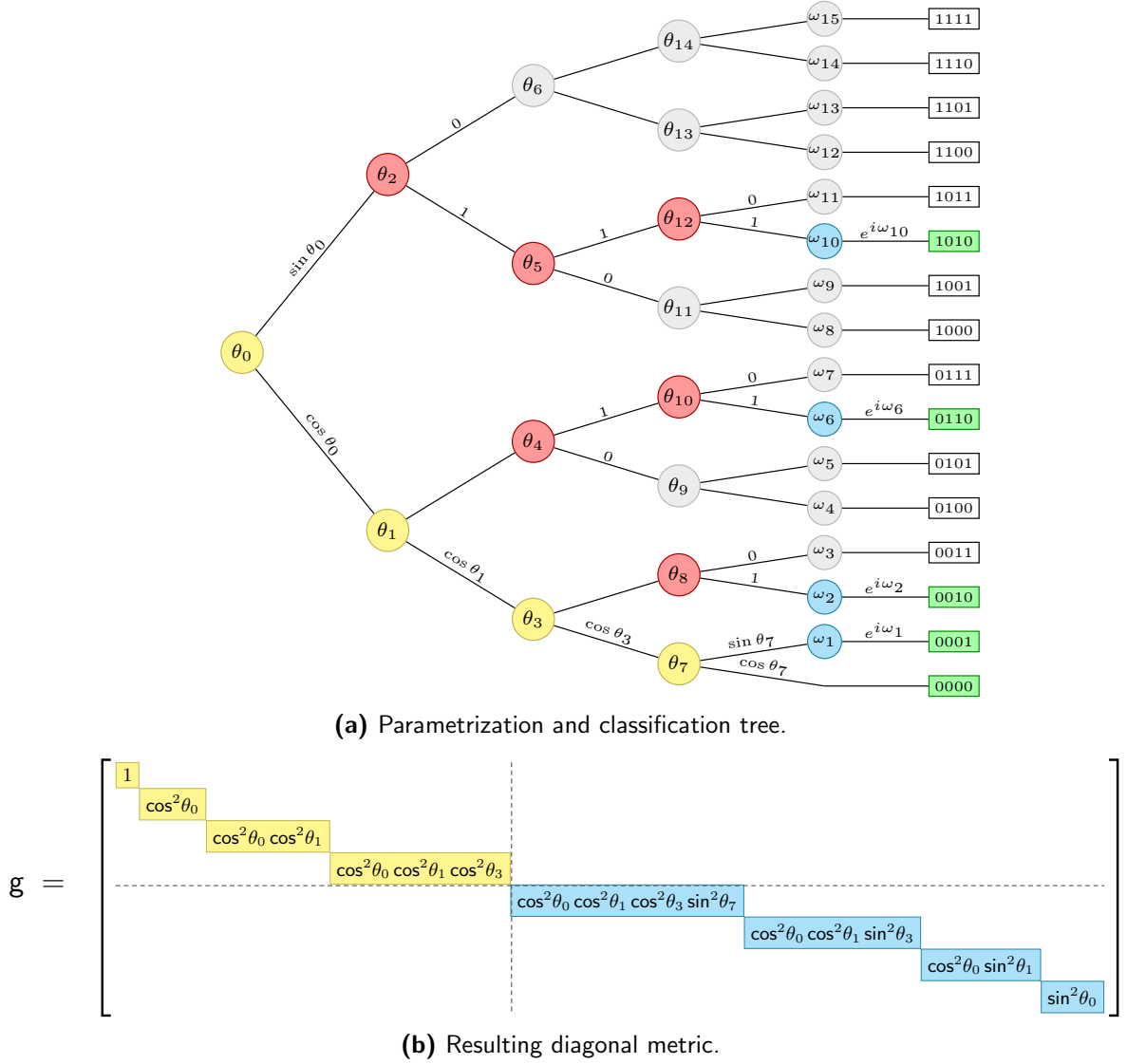
\begin{figure*}[!htbp]
    \centering
    \captionsetup[subfigure]{font=sf, labelfont={sf,bf}}
    \begin{subfigure}[t]{0.95\textwidth}
        \centering
        \resizebox{0.70\linewidth}{!}{

\providecommand{\PWactC}{yellow!70!black}
\providecommand{\PWactF}{yellow!55}
\providecommand{\PWfixC}{red!65!black}
\providecommand{\PWfixF}{red!40}
\providecommand{\PWinaC}{gray!55}
\providecommand{\PWinaF}{gray!15}
\providecommand{\PWleafC}{green!55!black}
\providecommand{\PWleafF}{green!35}
\providecommand{\PWomegaC}{cyan!60!black}
\providecommand{\PWomegaF}{cyan!30}

\begin{tikzpicture}[
  grow=right,
  level distance=18mm,
  level 1/.style={sibling distance=44mm},
  level 2/.style={sibling distance=22mm},
  level 3/.style={sibling distance=11mm},
  level 4/.style={sibling distance=5.5mm},
  level 5/.style={sibling distance=5.5mm, level distance=16mm},
  int/.style={circle, draw, inner sep=0pt, font=\scriptsize, minimum size=5.2mm, align=center},
  ph/.style ={circle, draw=\PWleafC, fill=green!8, inner sep=0pt, font=\tiny, minimum size=4.2mm},
  phA/.style={circle, draw=\PWomegaC, fill=\PWomegaF, inner sep=0pt, font=\tiny, minimum size=4.2mm},
  phI/.style={circle, draw=\PWinaC, fill=\PWinaF, inner sep=0pt, font=\tiny, minimum size=4.2mm},
  lf/.style ={rectangle, draw, inner sep=2pt, font=\tiny},
  lfS/.style={rectangle, draw=\PWleafC, fill=\PWleafF, inner sep=2pt, font=\tiny},
  edge from parent/.style={draw, line width=0.4pt},
  el/.style={midway, font=\tiny, sloped, above, inner sep=1pt},
]

\node[int, fill=\PWactF, draw=\PWactC] {$\theta_0$}
  child { node[int, fill=\PWactF, draw=\PWactC] {$\theta_1$}
    child { node[int, fill=\PWactF, draw=\PWactC] {$\theta_3$}
      child { node[int, fill=\PWactF, draw=\PWactC] {$\theta_7$}
        child { coordinate
          child { node[lfS] {$0000$} }
          edge from parent node[el]{$\cos\theta_7$}
        }
        child { node[phA] {$\omega_1$}
          child { node[lfS] {$0001$} edge from parent node[el]{$e^{i\omega_1}$} }
          edge from parent node[el]{$\sin\theta_7$}
        }
        edge from parent node[el]{$\cos\theta_3$}
      }
      child { node[int, fill=\PWfixF, draw=\PWfixC] {$\theta_8$}
        child { node[phA] {$\omega_2$} child { node[lfS] {$0010$} edge from parent node[el]{$e^{i\omega_2}$} } edge from parent node[el]{$1$} }
        child { node[phI] {$\omega_3$} child { node[lf] {$0011$} } edge from parent node[el]{$0$} }
      }
      edge from parent node[el]{$\cos\theta_1$}
    }
    child { node[int, fill=\PWfixF, draw=\PWfixC] {$\theta_4$}
      child { node[int, fill=\PWinaF, draw=\PWinaC] {$\theta_9$}
        child { node[phI] {$\omega_4$} child { node[lf] {$0100$} } }
        child { node[phI] {$\omega_5$} child { node[lf] {$0101$} } }
        edge from parent node[el]{$0$}
      }
      child { node[int, fill=\PWfixF, draw=\PWfixC] {$\theta_{10}$}
        child { node[phA] {$\omega_6$} child { node[lfS] {$0110$} edge from parent node[el]{$e^{i\omega_6}$} } edge from parent node[el]{$1$} }
        child { node[phI] {$\omega_7$} child { node[lf] {$0111$} } edge from parent node[el]{$0$} }
        edge from parent node[el]{$1$}
      }
    }
    edge from parent node[el]{$\cos\theta_0$}
  }
  child { node[int, fill=\PWfixF, draw=\PWfixC] {$\theta_2$}
    child { node[int, fill=\PWfixF, draw=\PWfixC] {$\theta_5$}
      child { node[int, fill=\PWinaF, draw=\PWinaC] {$\theta_{11}$}
        child { node[phI] {$\omega_8$} child { node[lf] {$1000$} } }
        child { node[phI] {$\omega_9$} child { node[lf] {$1001$} } }
        edge from parent node[el]{$0$}
      }
      child { node[int, fill=\PWfixF, draw=\PWfixC] {$\theta_{12}$}
        child { node[phA] {$\omega_{10}$} child { node[lfS] {$1010$} edge from parent node[el]{$e^{i\omega_{10}}$} } edge from parent node[el]{$1$} }
        child { node[phI] {$\omega_{11}$} child { node[lf] {$1011$} } edge from parent node[el]{$0$} }
        edge from parent node[el]{$1$}
      }
      edge from parent node[el]{$1$}
    }
    child { node[int, fill=\PWinaF, draw=\PWinaC] {$\theta_6$}
      child { node[int, fill=\PWinaF, draw=\PWinaC] {$\theta_{13}$}
        child { node[phI] {$\omega_{12}$} child { node[lf] {$1100$} } }
        child { node[phI] {$\omega_{13}$} child { node[lf] {$1101$} } }
      }
      child { node[int, fill=\PWinaF, draw=\PWinaC] {$\theta_{14}$}
        child { node[phI] {$\omega_{14}$} child { node[lf] {$1110$} } }
        child { node[phI] {$\omega_{15}$} child { node[lf] {$1111$} } }
      }
      edge from parent node[el]{$0$}
    }
    edge from parent node[el]{$\sin\theta_0$}
  };
\end{tikzpicture}}
        \subcaption{Parametrization and classification tree.}\label{fig:classified-tree-a}
    \end{subfigure}\\[4pt]
    \begin{subfigure}[t]{0.95\textwidth}
        \centering

\providecommand{\PWactC}{yellow!70!black}
\providecommand{\PWactF}{yellow!55}
\providecommand{\PWomegaC}{cyan!60!black}
\providecommand{\PWomegaF}{cyan!30}

\scalebox{0.85}{\begin{tikzpicture}[font=\small\sffamily,
    cell/.style={draw, rectangle, inner xsep=3pt, inner ysep=3pt, anchor=north west, line width=0.3pt},
    actcell/.style={cell, fill=\PWactF, draw=\PWactC},
    leafcell/.style={cell, fill=\PWomegaF, draw=\PWomegaC}]
  \renewcommand{\cos}{\operatorname{\mathsf{cos}}}%
  \renewcommand{\sin}{\operatorname{\mathsf{sin}}}%
  \node[actcell] (n0) at (0,0) {$1$};
  \node[actcell] (n1) at (n0.south east) {$\cos^{2}\!\theta_{0}$};
  \node[actcell] (n2) at (n1.south east) {$\cos^{2}\!\theta_{0}\cos^{2}\!\theta_{1}$};
  \node[actcell] (n3) at (n2.south east) {$\cos^{2}\!\theta_{0}\cos^{2}\!\theta_{1}\cos^{2}\!\theta_{3}$};
  \node[leafcell] (n4) at (n3.south east) {$\cos^{2}\!\theta_{0}\cos^{2}\!\theta_{1}\cos^{2}\!\theta_{3}\sin^{2}\!\theta_{7}$};
  \node[leafcell] (n5) at (n4.south east) {$\cos^{2}\!\theta_{0}\cos^{2}\!\theta_{1}\sin^{2}\!\theta_{3}$};
  \node[leafcell] (n6) at (n5.south east) {$\cos^{2}\!\theta_{0}\sin^{2}\!\theta_{1}$};
  \node[leafcell] (n7) at (n6.south east) {$\sin^{2}\!\theta_{0}$};
  \coordinate (TL) at (n0.north west);
  \coordinate (BR) at (n7.south east);
  \coordinate (TR) at (n7.east |- TL);
  \coordinate (BL) at (n0.west |- BR);
  \draw[dashed, line width=0.5pt, black!65, dash pattern=on 2.2pt off 1.6pt]
        ($(n3.north east |- TL)+(0.5pt,0)$) -- ($(n3.north east |- BR)+(0.5pt,0)$);
  \draw[dashed, line width=0.5pt, black!65, dash pattern=on 2.2pt off 1.6pt]
        ($(TL |- n3.south)+(0,-0.5pt)$) -- ($(TR |- n3.south)+(0,-0.5pt)$);
  \draw[line width=1.1pt] ($(TL)+(-0.08,0.04)$) -- ($(TL)+(-0.22,0.04)$)
        -- ($(BL)+(-0.22,-0.04)$) -- ($(BL)+(-0.08,-0.04)$);
  \draw[line width=1.1pt] ($(TR)+(0.08,0.04)$)  -- ($(TR)+(0.22,0.04)$)
        -- ($(BR)+(0.22,-0.04)$)  -- ($(BR)+(0.08,-0.04)$);
  \node[left=10pt, font=\Large\sffamily] at ($(TL)!0.5!(BL)+(-0.22,0)$) {$\mathsf{g} \;=\;$};
\end{tikzpicture}}
        \subcaption{Resulting diagonal metric.}\label{fig:classified-tree-b}
    \end{subfigure}
    \caption{Parametrization, classification, and the resulting diagonal metric for $n=4$ and $S=\{0000,0001,0010,0110,1010\}$. \textbf{(a)} Each internal node carries one $R_y$ parameter $\theta_i$; each final-layer node carries one $R_z$ parameter $\omega_j$, contributing the phase $e^{i\omega_j}$ on its outgoing edge. Branch factors follow the convention of \sref{sec:circuit} ($\cos\theta_i$ on downward, $\sin\theta_i$ on upward branches), labelled here along the leftmost path; the redundant global-phase coordinate $\omega_0$ is gauge-fixed to zero and drawn as a direct edge from $\theta_7$ to $\ket{0000}$ (no phase node). The metric weights $w_i$ and leaf probabilities $p_j$ are the squared path products of \sref{sec:metric}. Node colours follow the classification of \sref{sec:classify}: \textsc{Active} (yellow, $\theta_i$ free), \textsc{Fixed} (red, $\theta_i$ pinned to $0$ or $\pi/2$), \textsc{Inactive} (grey, gauge direction); green leaves are the elements of $S$. \textbf{(b)} The Hilbert-space chart pullback metric $g$ on the surviving parameters is diagonal, block-split into the active $\bm\theta = (\theta_0,\theta_1,\theta_3,\theta_7)$ coordinates (yellow block, entries $w_i$) and the surviving $\bm\omega = (\omega_1,\omega_2,\omega_6,\omega_{10})$ coordinates (cyan block, entries $p_j$), read off the tree in \textbf{(a)}. Its inverse acts elementwise; projective updates with the reference phase fixed use the restoration and re-gauging rule in \sref{sec:metric}.}
    \label{fig:classified-tree}
\end{figure*}

\subsection{Closed-form diagonal chart metric}\label{sec:metric}

The information-geometric content of a variational ansatz is captured by the parameter-space pullback of the Hilbert inner product, which controls quantum natural gradient~\cite{Amari1998}, variational imaginary-time evolution, and variational real-time evolution. For the binary-tree ansatz, this metric admits an exact closed form, is diagonal on every chart parameter, and can be read off the tree by inspection, with no auxiliary circuits, no matrix inversion, and no regularization.

Concretely, writing $g_{\mu\nu} \;:=\; \mathrm{Re}\,\braket{\partial_\mu\psi\,|\,\partial_\nu\psi}$ for that pullback, the metric is block-diagonal between the amplitude parameters $\bm\theta$ and the leaf phases $\bm\omega$, with each block itself diagonal,
\begin{align}
g_{\theta_i\theta_j} &\;=\;
\begin{cases}
\displaystyle w_i \;=\; \prod_{m\in\mathrm{anc}(i)} f_m(\theta_m)^2, & i=j,\\[6pt]
0, & i\neq j,
\end{cases} \label{eq:g-theta}\\[6pt]
g_{\omega_j\omega_{j'}} &\;=\;
\begin{cases}
\displaystyle p_j \;=\; |c_j|^2 \;=\; \prod_{m\in\mathrm{anc}(j)} f_m(\theta_m)^2, & j=j',\\[6pt]
0, & j\neq j',
\end{cases} \label{eq:g-omega}\\[6pt]
g_{\theta_i\omega_j} &\;=\; 0,\label{eq:g-cross}
\end{align}
where each factor $f_m\in\{\sin,\cos\}$ is selected by the path from ancestor $m$ (cosine at every downward step, sine at every upward step), the amplitude weight $w_i$ stopping at node $i$ and the leaf probability $p_j$ continuing all the way to the leaf.

These diagonal entries are simply measurement probabilities read off the tree. The leaf probability $p_j=|c_j|^2$ is the Born probability of measuring basis state $j$. The amplitude weight $w_i$ is the probability of measuring any basis state in the subtree rooted at node $i$, i.e.\ the chance that an outcome lies among the descendant leaves of $\theta_i$, so $w_i=\sum_{j\,\preceq\, i}p_j$ and $w_{\text{root}}=1$. (Equivalently, $w_i$ is the running path product of $f_m(\theta_m)^2$ from the root down to node $i$, each factor $f_m(\theta_m)^2$ being the conditional probability of branching into the chosen child at node $m$; $p_j$ is the same product carried all the way to the leaf.) The metric diagonal is thus just the list of subtree and leaf probabilities of the tree, and the natural-gradient preconditioner $g^{-1}$ divides each parameter's Euclidean gradient by the probability mass it controls. The formal statement (Theorem~\ref{thm:diag-metric}), a short proof, and a worked example of reading these weights off the tree are collected in Appendix~\ref{app:diag-metric}. The projective Fubini--Study metric $g^{\mathrm{FS}}$ is one quarter of the quantum Fisher information~\cite{BraunsteinCaves1994, Meyer2021} (Appendix~\ref{app:geometry-primer}).

Because $g$ is positive on every open coordinate patch (all $w_i,p_j>0$), its inverse is the Hadamard reciprocal,
\begin{equation}\label{eq:Minv}
g^{-1} \;=\; \mathrm{diag}(1/w_i) \,\oplus\, \mathrm{diag}(1/p_j),
\end{equation}
computable without auxiliary circuits and without a regularization parameter. The projective Fubini--Study metric of Provost and Vall\'ee~\cite{Provost1980} is obtained from $g$ by subtracting $pp^\top$ from its redundant phase block, removing the global-phase direction; the two preconditioners yield the same physical update on the prepared state, and a self-contained derivation of this equivalence, together with the projective K\"ahler formulation of real-time evolution, is given in Appendix~\ref{app:ng-omega-equiv}. The corresponding projective determinant also factorizes along the tree (Appendix~\ref{app:fs-design}), the basis for the exact Haar sampling of \sref{sec:fs-apps}.

The diagonal phase inverse is applied in the redundant one-phase-per-active-leaf representation. If the circuit fixes $\omega_r=0$, restore $u_r=-\sum_{j\ne r}u_j$, where $u_j=\partial_{\omega_j}E$. The metric-raised relative-phase components are then $u_j/p_j-u_r/p_r$ for $j\ne r$, requiring no additional energy evaluations.

\medskip\noindent\textbf{Invariance under fixed dressing.}\ Because $g_{\mu\nu}=\mathrm{Re}\,\braket{\partial_\mu\psi|\partial_\nu\psi}$ is built from the chart pullback alone, it is unchanged when the prepared state is post-composed with any fixed (parameter-independent) unitary $U$: the metric of $U\ket{\psi}$ equals that of $\ket{\psi}$, since $U$ unitary leaves $\braket{\partial_\mu\psi|\partial_\nu\psi}$ invariant. The closed-form diagonal metric is therefore a property of the bare tree core that survives verbatim inside any fixed dressing $\ket{\psi}\!\to\!U\ket{\psi}$, the composition $\ket{\Psi}=U(\bm\phi)\ket{\psi(\bm\theta,\bm\omega)}$ that organizes the dressed applications of \sref{sec:two-block}--\sref{sec:decouple}, so the geometry framework extends from the tree itself to the whole class $\{U\ket{\psi(\bm\theta,\bm\omega)}\}$ of fixed-dressed cores (Appendix~\ref{app:invariance}). One instance, a sparse-preparation circuit reaching the near-optimal $O(nk/\log n)$ scaling while retaining the diagonal metric, is developed in Appendix~\ref{app:sparse-routing}.

\subsection{Imaginary and real time from a single K\"ahler structure}\label{sec:kahler}

The same diagonal $g^{-1}$ drives both imaginary-time and real-time evolution, and the reason is geometry intrinsic to quantum mechanics. At every state $\ket\psi$, the Hilbert tangent space carries a real symmetric form $g(\xi,\eta) = \mathrm{Re}\braket{\xi|\eta}$, the Riemannian metric controlling distances. It also carries a compatible antisymmetric form $\Omega(\xi,\eta)=g(\xi,J\eta)$, with $J=i$ the complex structure inherited from multiplication by $i$ in Hilbert space; equivalently, $\Omega$ is the imaginary part of the Hilbert inner product up to the conventional sign. These structures make the Hilbert space a K\"ahler manifold~\cite{Provost1980, Ashtekar1999}. A short introductory primer, including the one-qubit Bloch-sphere picture of the two perpendicular flows, is given in Appendix~\ref{app:geometry-primer}. Both flows act on two ingredients built from the energy $E(\bm\lambda)=\braket{\psi(\bm\lambda)|H|\psi(\bm\lambda)}$, with $\bm\lambda=(\bm\theta,\bm\omega)$ the joint vector of tree angles $\theta_i$ and leaf phases $\omega_j$: the Euclidean energy covector $dE=\bigl(\partial_{\theta_i}E,\ \partial_{\omega_j}E\bigr)$, assembled from the parameter-shift derivatives of $E$ with respect to each tree parameter, and the diagonal pullback metric $g=\operatorname{diag}\bigl(w_i,\ p_j\bigr)$ of Theorem~\ref{thm:diag-metric}, whose entries are the amplitude weights $w_i$ and leaf probabilities $p_j$. On the projective sector parametrised by this chart, imaginary-time evolution is the downhill Riemannian gradient flow of the energy with respect to $g$~\cite{McArdle2019, Yuan2019, Stokes2020, McMahon2026}, while real-time Schr\"odinger evolution is the Hamiltonian flow generated by the same energy; the factor $1/2$ that accompanies the latter follows from the convention $g(\xi,\eta)=\mathrm{Re}\braket{\xi|\eta}$ for the tangent vectors $\xi,\eta$ introduced above. The physical imaginary-time equation carries the same overall factor $1/(2\hbar)$; in the natural-gradient optimizer this constant is absorbed into the imaginary-time step size or learning rate. The algorithms below set $\hbar=1$. The complex structure supplies the $90^\circ$ rotation between the two tangent directions, which is exactly the Wick rotation $\tau\mapsto-it$.

The following formulas use the redundant-phase convention of \sref{sec:metric}. The raised gradient is horizontal, and the projective complex structure $J$ acts on its Hilbert-space representative by multiplication by $i$.

For this ansatz the rotation has an unusually concrete form. The same shift-rule oracle returns the Euclidean derivative $dE$ for both flows. The diagonal metric first raises this covector by the Hadamard reciprocal $g^{-1}$, giving the metric gradient $\nabla_g E=g^{-1}dE$, which is already the entire imaginary-time / natural-gradient update,
\begin{equation}\label{eq:vite-flow}
\dot{\bm\lambda}_{\mathrm{imag}} \;=\; -\,g^{-1}\,dE \;=\; -\,\nabla_g E .
\end{equation}
Real-time evolution uses the same raised gradient, but rotated by the K\"ahler complex structure: the binary tree applies $J$ through one bottom-up and one top-down pass over the active split tree, supplying the $90^\circ$ turn that selects Hamiltonian rather than gradient motion. This gives
\begin{equation}\label{eq:tdvp-flow}
\dot{\bm\lambda}_{\mathrm{real}} \;=\; -\,\frac{1}{2\hbar}\,J\!\left(g^{-1}\,dE\right).
\end{equation}
The subscripts label the imaginary-time / natural-gradient and physical real-time conventions respectively; the tree form of $J$ is recorded explicitly in Appendix~\ref{app:tdvp-rhs}. This reduction to one diagonal metric raise plus linear-time tree passes is the central computational consequence of Theorem~\ref{thm:diag-metric}: no auxiliary metric circuits, dense matrix assembly or classical inversion are required. In particular, one explicit-Euler step of the real-time flow~\eqref{eq:tdvp-flow} is the exact real-time twin of one quantum-natural-gradient / imaginary-time step of~\eqref{eq:vite-flow}: same Euclidean derivative oracle and same diagonal preconditioner, followed by the K\"ahler rotation and interpreted with timestep $dt$ rather than learning rate.

\subsection{Phase gauge invariance and chart singularities}\label{sec:gauge-singular}

For any physical objective $\mathcal F(\ket\psi)$ (energy expectation, fidelity, an action functional), the leaf-phase gradient is automatically orthogonal to the global-phase direction,
\begin{equation}\label{eq:gauge-orth}
\sum_j\partial_{\omega_j}\mathcal F \;=\; 0,
\end{equation}
because $\sum_j\partial_{\omega_j}\ket\psi = i\ket\psi$ generates the global phase and $\mathcal F$ is invariant under that generator. The imaginary-time flow~\eqref{eq:vite-flow} therefore never excites the gauge mode; the real-time K\"ahler flow~\eqref{eq:tdvp-flow} is defined only up to an arbitrary uniform phase velocity, which is projected out at each step (Appendix~\ref{app:tdvp-rhs}). Theorem~\ref{thm:diag-metric} carries over verbatim to the pruned ansatz of \sref{sec:pruning}, with the indices restricted to the active internal nodes and the active-leaves set $S$; the active-leaf probabilities then satisfy $\sum_{j\in S}p_j = 1$. The weights $w_i$ and $p_j$ vanish only on the measure-zero singular configurations of the polyspherical chart, where an entire subtree of parameters becomes underdetermined by the metric. Such configurations arise generically only at initialisation, a Hartree--Fock start places every active angle at $0$ or $\pi/2$, and are resolved by an analytic singular-coordinate initialisation step (Appendix~\ref{app:singular-init}) that matches the chart move to the leading-order Euclidean update from a single sweep of $(\bm\theta,\bm\omega)$ shift-rule queries, with cost $8(|S|-1)$. Once the chart is open, $g^{-1}$ is used unregularised.

\subsection{Analytic gradients via parameter-shift rules}\label{sec:shift}

The Euclidean energy gradient that enters Eq.~\eqref{eq:vite-flow} is itself available in closed form on this ansatz. Each tree parameter $\theta_i$ drives a single uniformly controlled $R_y$ rotation whose generator has spectrum $\{-1,0,+1\}$, so $E(\theta_i)$ (with all other parameters fixed) is a trigonometric polynomial of degree at most $2$ and its derivative obeys the exact four-term shift rule
\begin{equation}\label{eq:shift-rule}
\partial_{\theta_i}E \;=\; \frac{2+\sqrt 2}{4}\bigl[E(\theta_i+\tfrac{\pi}{4})-E(\theta_i-\tfrac{\pi}{4})\bigr] \;-\; \frac{2-\sqrt 2}{4}\bigl[E(\theta_i+\tfrac{3\pi}{4})-E(\theta_i-\tfrac{3\pi}{4})\bigr].
\end{equation}
This is the $R=2$ instance of the equispaced parameter-shift family of Wierichs, Izaac and Killoran~\cite{Wierichs2022}, which generalises the original parameter-shift rule~\cite{Mitarai2018, SchuldGradient2019}; the same rule has appeared previously for fermionic single excitations~\cite{Anselmetti2021, Kottmann2021}. For the leaf phases, $\partial_{\omega_j}\ket\psi = i c_j \ket{j}$ has spectrum $\{-1,0\}$ and the standard two-term rule applies,
\begin{equation}\label{eq:shift-rule-omega}
\partial_{\omega_j}E \;=\; \tfrac{1}{2}\bigl[E(\omega_j+\tfrac{\pi}{2}) - E(\omega_j-\tfrac{\pi}{2})\bigr].
\end{equation}
Both rules are exact for any Hermitian observable, incur no finite-difference bias and no $\varepsilon$ to tune; a self-contained derivation is given in Appendix~\ref{app:shift-proof}. One full Euclidean energy gradient therefore costs $4|\bm\theta|+2|\bm\omega|$ circuit evaluations exactly, so one full quantum-natural-gradient step is computable from at most $4P+1$ energy evaluations, where $P$ is the number of variational parameters.

When the prepared state and the observable are real, the amplitude-only chart $\bm\omega=\bm 0$ that underlies ground-state search and imaginary-time evolution, each tree derivative collapses further to a two-evaluation rule~\cite{Mari2021, Kottmann2021},
\begin{equation}\label{eq:shift-rule-real}
\partial_{\theta_i}E \;=\; \tilde E(\theta_i+\tfrac{\pi}{4}) - \tilde E(\theta_i-\tfrac{\pi}{4}),
\end{equation}
where $\tilde E(\theta_i\pm\tfrac{\pi}{4})$ is the expectation of $H$ in the state the circuit prepares at the shifted angle, followed by a fixed gate $U_0^{(\pm)}=\exp(\pm i\tfrac{\pi}{4} P_i)$, with $P_i = \ket{b_i}\!\bra{b_i}_{\mathrm{ctrl}}\otimes I$ the control projector of Eq.~\eqref{eq:shift-generator}. This $U_0^{(\pm)}$ is diagonal in the computational basis and acts only on the ancestor qubits, and is realised by the same $R_z$ phase block the complex ansatz already compiles, truncated above the target qubit $q_i$; it therefore adds at most the linear-in-$k$ CNOT overhead of the pruned ansatz, with one such block per derivative. The full real-target gradient then costs $2|\bm\theta|$ evaluations and one natural-gradient step $2P+1$, half the generic figure; the reduction is specific to the real chart and does not apply to the complex-amplitude real-time runs, which retain the four-term rule. The two rules trade off against each other: the four-term rule keeps the circuit at its most compact but doubles the evaluation count, whereas the two-term rule halves the evaluations at the price of the slightly longer correction block. A derivation, together with the precise reuse of the phase-block compiler, is given in Appendix~\ref{app:shift-real}.

The same shift-rule gradient, followed by a single $\mathcal O(P)$ bottom-up/top-down sweep through the chart tree, also supplies the right-hand side of the real-time flow~\eqref{eq:tdvp-flow}, integrated with standard explicit time-steppers, first-order Euler and fourth-order Runge--Kutta (\sref{sec:hubbard}), around the singular-coordinate event handler of Appendix~\ref{app:singular-init}; the closed-form RHS expression and the global-phase gauge fix applied at every step are given in Appendix~\ref{app:tdvp-rhs}. This stands in sharp contrast to projection-based schemes such as pVQD~\cite{Barison2021, Gacon2024}, in which each timestep is itself a nonlinear least-squares optimisation over the ansatz parameters, and to the McArdle--Endo--Benjamin scheme~\cite{McArdle2019}, which assembles the metric from $\mathcal O(P^2)$ Hadamard tests on generic ans\"atze.

\subsection{Active, fixed and inactive tree parameters}\label{sec:pruning}\label{sec:classify}

In many applications the target state is known a priori to be supported on a small subset of computational-basis states, a particle-number, spin-projection or symmetry-irrep sector in chemistry; a feasible-solutions subspace in constrained optimization~\cite{Farhi2014, Hadfield2019, Lucas2014}; a fixed-Hamming-weight code in error correction~\cite{Gottesman1997}. These are called the active leaves~$S\subset\{0,\dots,2^n-1\}$, with $k=|S|$. The full tree of \sref{sec:circuit} can prepare any such state, but most of its $2^n - 1$ tree parameters are then redundant. This redundancy can be turned into an explicit gauge freedom and used to compile a much smaller circuit.

Mark every leaf in $S$ as occupied and propagate ``occupied'' bottom-up: an internal node is occupied if at least one of its descendant leaves is active. Each internal node then falls into exactly one of three classes (illustrated for the running example in Fig.~\ref{fig:classified-tree}). An internal node is active if both children are occupied, in which case the corresponding $\theta_i$ remains a free variational parameter. It is fixed if exactly one child is occupied: the corresponding $\theta_i$ is then forced to the value that sends all amplitude into the occupied subtree, $\theta_i = 0$ when the downward (cosine) branch is occupied and $\theta_i = \pi/2$ otherwise. It is inactive if neither child is occupied, in which case $\ket\psi$ does not depend on $\theta_i$.

The total number of active and fixed internal nodes equals the number of internal nodes of the binary subtree spanned by $S$ (the union of the root-to-leaf paths of the active leaves); a sharper structure-dependent form of the CNOT bound based on this count is given in Appendix~\ref{app:cnot-evidence}.

\subsection{Inactive-parameter gauge fixing and the pruning compiler}\label{sec:gauge}

The inactive parameters do not affect the prepared state: by Eq.~\eqref{eq:c2p-ry}, the amplitude on every leaf in $S$ is a product of factors along a path that visits only active and fixed nodes. The pruning compiler turns this gauge freedom into circuit savings in two stages (illustrated end-to-end for the running example in Fig.~\ref{fig:pruned-example}): first, the inactive parameters are chosen so as to zero as many circuit parameters (elementary $R_y$ angles) in the Walsh--Hadamard expansion of the cascade as possible (exact integer Gaussian elimination on the inactive columns of an affine system, performed level by level); second, the resulting empty sub-cascades cause pairs of bracketing CNOTs in the triangular cascade to cancel pairwise. A constant-bit removal pass is run beforehand to strip qubits whose value is fixed across all active leaves. The full affine system, the level-block matrices, the constant-bit pass, and the end-to-end algorithm (Algorithm~\ref{alg:compile}) are recorded in Appendix~\ref{app:compiler}; that the elimination stays in exact integer arithmetic (every pivot a power of two) is proved in Appendix~\ref{app:integrality}.

The resulting circuit is referred to as the pruned ansatz; the diagonal metric of Thm.~\ref{thm:diag-metric} restricts to a diagonal metric on the active parameters, with the same path-product weights $w_i$ and $p_j$ of Eqs.~\eqref{eq:g-theta}--\eqref{eq:g-omega}, where the angles at fixed nodes take their pinned values and the inactive ancestors are absent. The surviving free parameters number exactly the dimension of the subspace they span---$k-1$ amplitude angles, together with up to $k-1$ leaf phases for a complex target (one global phase fixed), matching $\dim_{\mathbb R}\mathbb{C}\mathbf P^{\,k-1}$. The pruned ansatz is therefore a minimal complete chart on $\mathrm{span}(S)$: it reaches every state of the target subspace and never leaves it, carrying no redundant variational directions. Reachability, strict confinement to the sector, and this parameter economy are exactly the guarantees a fixed-template hardware-efficient ansatz~\cite{Kandala2017} does not provide~\cite{Leone2023}, and here they follow from the active-leaf set alone rather than from a hand-tuned circuit template.

\begin{figure*}[!htbp]
\centering
\makebox[\textwidth][c]{

\providecommand{\PWactC}{yellow!70!black}
\providecommand{\PWactF}{yellow!55}
\providecommand{\PWfixC}{red!65!black}
\providecommand{\PWfixF}{red!40}
\providecommand{\PWinaC}{gray!55}
\providecommand{\PWinaF}{gray!15}
\providecommand{\PWdelC}{gray!50}
\providecommand{\PWdelF}{gray!8}
\providecommand{\PWcanC}{orange!75!red}
\providecommand{\PWcanF}{orange!22}

\providecommand{\PWactBox}{\push{\fcolorbox{\PWactC}{\PWactF}{\rule{0pt}{4.0pt}\hspace{6.0pt}}}\qw}
\providecommand{\PWfixBox}{\push{\fcolorbox{\PWfixC}{\PWfixF}{\rule{0pt}{4.0pt}\hspace{6.0pt}}}\qw}
\providecommand{\PWdelBox}{\push{\fcolorbox{\PWdelC}{\PWdelF}{\rule{0pt}{4.0pt}\hspace{6.0pt}}}\qw}
\providecommand{\PWactL}[1]{\push{\fcolorbox{\PWactC}{\PWactF}{$\rule{0pt}{1.9ex}\mathsf{R_y}(#1)$}}\qw}
\providecommand{\PWfixL}[1]{\push{\fcolorbox{\PWfixC}{\PWfixF}{$\rule{0pt}{1.9ex}\mathsf{R_y}(#1)$}}\qw}
\providecommand{\PWTARG}{\push{\raisebox{-0.12ex}{\tikz[baseline=-0.6ex]{%
  \draw[fill=white,line width=0.4pt] (0,0) circle (0.85ex);%
  \draw[line width=0.5pt] (-0.65ex,0) -- (0.65ex,0);%
  \draw[line width=0.5pt] (0,-0.65ex) -- (0,0.65ex);}}}\qw}
\providecommand{\PWcanTARG}{\push{\raisebox{-0.12ex}{\tikz[baseline=-0.6ex]{%
  \draw[\PWcanC,fill=\PWcanF,line width=0.6pt] (0,0) circle (0.95ex);%
  \draw[\PWcanC,line width=0.7pt] (-0.7ex,0) -- (0.7ex,0);%
  \draw[\PWcanC,line width=0.7pt] (0,-0.7ex) -- (0,0.7ex);%
  \draw[\PWcanC,line width=0.9pt] (-1.1ex,-1.1ex) -- (1.1ex,1.1ex);}}}\qw}
\colorlet{PWcanXY}{orange!75!red}
\providecommand{\PWcanCTRL}[1]{\ar @{-} |<{\tikz[baseline=-0.5ex]{%
  \useasboundingbox (-0.6ex,-0.6ex) rectangle (0.6ex,0.6ex);%
  \draw[\PWcanC,line width=0.6pt] (0,0) -- ++(0,-#1*2.0em);%
  \fill[\PWcanC] (0,0) circle (0.6ex);}} [0,-1]}

\begin{tikzpicture}[
  every node/.style={font=\small\sffamily},
  panel/.style={inner sep=0pt, align=center},
  plab/.style={font=\small\sffamily\bfseries, anchor=south west, inner sep=1pt, yshift=1pt},
]

\node[panel, anchor=north west, text width=\linewidth] (full) at (0,0) {%
  \resizebox{0.985\linewidth}{!}{\Qcircuit @C=0.34em @R=0.85em {
    \lstick{\mathsf{q_3}} & \PWactBox & \qw       & \ctrl{1}  & \qw       & \ctrl{1}  & \qw       & \qw       & \qw       & \ctrl{2}  & \qw       & \qw       & \qw       & \ctrl{2}  & \qw       & \qw       & \qw       & \qw       & \qw       & \qw       & \qw       & \ctrl{3}  & \qw       & \qw       & \qw       & \qw       & \qw       & \qw       & \qw       & \ctrl{3} \\
    \lstick{\mathsf{q_2}} & \qw       & \PWactBox & \PWTARG   & \PWactBox & \PWTARG   & \qw       & \ctrl{1}  & \qw       & \qw       & \qw       & \ctrl{1}  & \qw       & \qw       & \qw       & \qw       & \qw       & \PWcanCTRL{2}  & \qw       & \qw       & \qw       & \qw       & \qw       & \qw       & \qw       & \PWcanCTRL{2}  & \qw       & \qw       & \qw       & \qw \\
    \lstick{\mathsf{q_1}} & \qw       & \qw       & \qw       & \qw       & \qw       & \PWdelBox & \PWTARG   & \PWactBox & \PWTARG   & \PWfixBox & \PWTARG   & \PWactBox & \PWTARG   & \qw       & \PWcanCTRL{1}  & \qw       & \qw       & \qw       & \PWcanCTRL{1}  & \qw       & \qw       & \qw       & \ctrl{1}  & \qw       & \qw       & \qw       & \ctrl{1}  & \qw       & \qw \\
    \lstick{\mathsf{q_0}} & \qw       & \qw       & \qw       & \qw       & \qw       & \qw       & \qw       & \qw       & \qw       & \qw       & \qw       & \qw       & \qw       & \PWdelBox & \PWcanTARG & \PWdelBox & \PWcanTARG & \PWdelBox & \PWcanTARG & \PWdelBox & \PWTARG   & \PWdelBox & \PWTARG   & \PWdelBox & \PWcanTARG & \PWactBox & \PWTARG   & \PWactBox & \PWTARG
  }}%
};
\node[plab] at (full.north west) {(a)~Full unpruned cascade};

\node[panel, below=6mm of full.south west, anchor=north west, text width=\linewidth] (pruned) {%
  \resizebox{0.985\linewidth}{!}{\Qcircuit @C=0.5em @R=0.85em {
    \lstick{\mathsf{q_3}} & \PWactL{2\theta_0} & \qw                & \ctrl{1}  & \qw                & \ctrl{1}  & \qw       & \qw                              & \ctrl{2}  & \qw                       & \qw       & \qw                              & \ctrl{2}  & \ctrl{3}  & \qw       & \qw                & \qw       & \qw                & \ctrl{3} \\
    \lstick{\mathsf{q_2}} & \qw                & \PWactL{\theta_1}  & \PWTARG   & \PWactL{\theta_1}  & \PWTARG   & \ctrl{1}  & \qw                              & \qw       & \qw                       & \ctrl{1}  & \qw                              & \qw       & \qw       & \qw       & \qw                & \qw       & \qw                & \qw \\
    \lstick{\mathsf{q_1}} & \qw                & \qw                & \qw       & \qw                & \qw       & \PWTARG   & \PWactL{\theta_3+\tfrac{\pi}{2}} & \PWTARG   & \PWfixL{-\pi}             & \PWTARG   & \PWactL{\theta_3+\tfrac{\pi}{2}} & \PWTARG   & \qw       & \ctrl{1}  & \qw                & \ctrl{1}  & \qw                & \qw \\
    \lstick{\mathsf{q_0}} & \qw                & \qw                & \qw       & \qw                & \qw       & \qw       & \qw                              & \qw       & \qw                       & \qw       & \qw                              & \qw       & \PWTARG   & \PWTARG   & \PWactL{\theta_7} & \PWTARG   & \PWactL{\theta_7} & \PWTARG
  }}%
};
\node[plab] at (pruned.north west) {(b)~Pruned circuit};

\end{tikzpicture}}
\caption{End-to-end pruning walkthrough (Algorithm~\ref{alg:compile}) for the running example $n=4$, $S=\{0000,0001,0010,0110,1010\}$, with classified tree shown in Fig.~\ref{fig:classified-tree}. Classes:\, \textsc{Active}~$=\{\theta_0,\theta_1,\theta_3,\theta_7\}$ (4 free), \textsc{Fixed}~$=\{\theta_2{=}0,\theta_4{=}1,\theta_5{=}1,\theta_8{=}0,\theta_{10}{=}0,\theta_{12}{=}0\}$, \textsc{Inactive}~$=\{\theta_6,\theta_9,\theta_{11},\theta_{13},\theta_{14}\}$. \textbf{(a)} The full unpruned triangular circuit ($15\,R_y + 14$ CNOT); each elementary $R_y$ is coloured by the class of the tree parameter it serves -- yellow~$=$~kept (active), red~$=$~kept as a constant $\pm\pi/2$ or $\pm\pi$ (fixed), faded grey~$=$~removed by pruning -- and the two bracketing-empty CNOT pairs that cancel after $R_y$ removal are highlighted in orange. Block boundaries (left to right): position~0 the cascade on $q_3$, positions~1--4 on $q_2$, 5--12 on $q_1$, 13--28 on $q_0$. The two cascades that undergo non-trivial elimination are $q{=}0$ and $q{=}1$. \textbf{(b)} The pruned circuit: $8\,R_y + 10$ CNOT, retaining the four free angles $\theta_0,\theta_1,\theta_3,\theta_7$ together with one compiler-emitted constant angle $-\pi$. Overall reduction: $15\!\to\!8$ rotations, $14\!\to\!10$ CNOTs (cancelled pairs at positions $(14,18)$ and $(16,24)$ in panel~(a)), $15\!\to\!4$ free parameters.}
\label{fig:pruned-example}
\end{figure*}
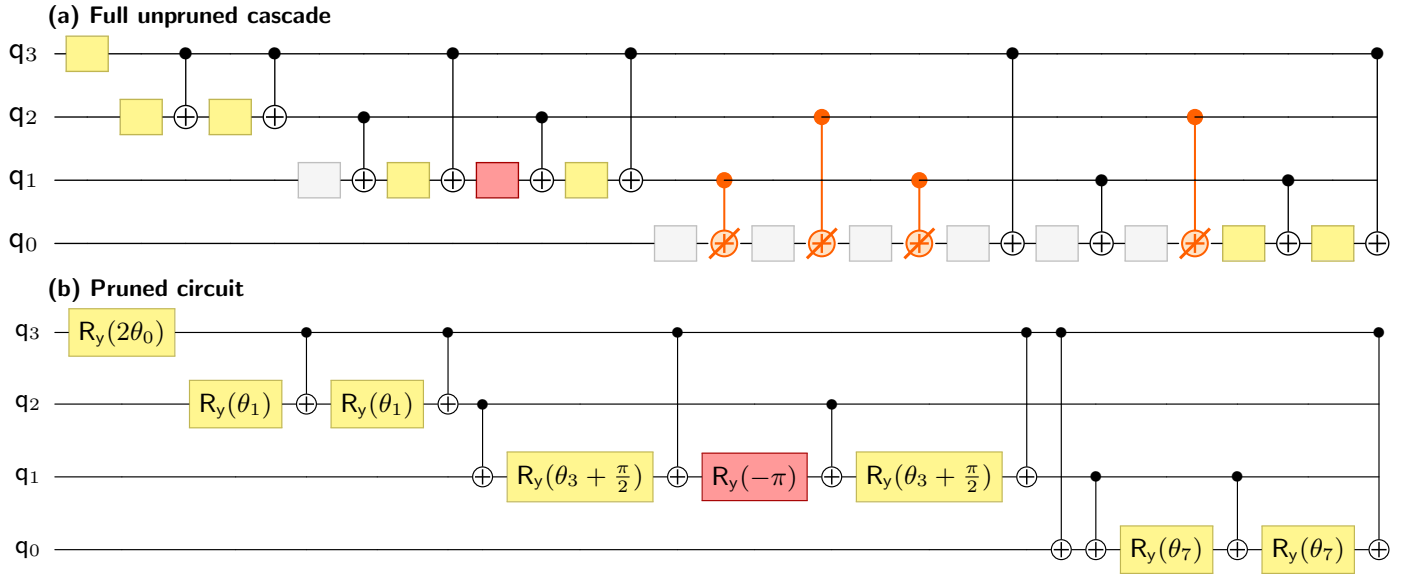

\subsection{Linear CNOT scaling in the active-leaf count}\label{sec:cnot-bound}

The CNOT count of the pruned circuit grows linearly in the number $k$ of active leaves, and this scaling is provable for any qubit ordering: the run-count identity below, combined with a bound on the surviving rotations (at most one per occupied node per level), gives $N_{\mathrm{cx}}\le 2nA_s\le 2n\bigl((n-1)k+1\bigr)=O(n^2k)$, where $A_s$ is the internal-node count of the binary subtree spanned by $S$ (Appendix~\ref{app:cnot-structure}, Proposition~\ref{prop:linear}). Three regimes are visible (Appendix~\ref{app:cnot-evidence}): a clustered sparse regime, in which the $k$ leaves fill a common depth-$\lceil\log_2 k\rceil$ subtree so that the circuit collapses to a dense state preparation on that subtree alone, with the $n$-independent cost $N_{\mathrm{cx}}\approx k$ (exactly $k-2$ when $k$ is a power of two); a spread sparse regime whose per-leaf slope grows with $n$ but stays within the structural envelope $2n$; and a dense regime saturating at the M\"ott\"onen ceiling $2(2^n-1)$. The sparse and dense bounds are
\begin{equation}\label{eq:cnot-bound}
N_{\mathrm{cx}} \;\le\; 2nk - 2 \quad\text{(sparse, with reordering; conjectured)}, \qquad N_{\mathrm{cx}} \;\le\; 2(2^n - 1) \quad\text{(dense)};
\end{equation}
The tighter sparse bound $N_{\mathrm{cx}}\le 2nk-2$ follows from the structure-dependent form $N_{\mathrm{cx}}\le 2(A_s+k-2)$ via the same envelope $A_s\le(n-1)k+1$. It sharpens the unconditional $O(n^2k)$ estimate above by replacing the $2n$ prefactor of $A_s$ with $2$, and this sharpening, not the linear scaling itself, is what requires the reordering, which the compiler currently finds by an exponential-time branch-and-bound search over qubit orderings (Appendix~\ref{app:compiler}). Through the run-count identity $N_{\mathrm{cx}}=2\sum_{q,\beta}R_{q,\beta}$ (Appendix~\ref{app:cnot-structure}) it is equivalent to the single combinatorial inequality $\sum_{q,\beta}R_{q,\beta}\le E_s-1$, which is proved for $k=2$ and, under the inactive-node-maximising reordering, verified by exhaustive enumeration through $n=4$ and sampling through $n=8$ with no violations; without the reordering it can be violated (Appendix~\ref{app:cnot-evidence}). The linear scaling contrasts with the $\Theta(2^n)$ cost of dense M\"ott\"onen preparation~\cite{Mottonen2005, Sun2023}: when $k\ll 2^n$ the pruned ansatz is exponentially cheaper, and the savings follow automatically from the active-leaves set rather than from any user-supplied gate template. Recent constructions establish (near-)optimal circuit sizes for sparse and fixed-Hamming-weight state classes~\cite{Li2025, Luo2025}, against which the linear-in-$k$ count here can be benchmarked. Indeed, by the dressing invariance of Appendix~\ref{app:invariance}, loading the $k$ amplitudes on a compact $\lceil\log_2 k\rceil$-qubit tree and routing them to their targets with such a fixed log-efficient addressing reaches the near-optimal $O(nk/\log n)$ scaling with the closed-form metric intact (Appendix~\ref{app:sparse-routing}). The direct pruned tree nonetheless remains preferable at presently relevant sizes, where its smaller constant outweighs the asymptotic gain.

\section{Numerical results}\label{sec:numerics}

The construction is illustrated across two complementary regimes. The bare ansatz (state preparation and energy, optimised or propagated by the closed-form diagonal-metric updates) is first applied to symmetry-adapted molecular electronic structure (\sref{sec:vqe}), real-time molecular dynamics (\sref{sec:mol-dynamics}), and Fermi--Hubbard quench dynamics (\sref{sec:hubbard}), where it reaches reference accuracy at markedly lower two-qubit depth than the state-of-the-art baselines. As proved in Appendix~\ref{app:simulability}, these bare computations are classically simulable in the active-leaf count, so their accuracy can be verified exactly; they establish the ansatz as an efficient, exactly controllable primitive whose composition with a hard dressing is where a quantum--classical separation can arise. Composing that primitive with a non-classically-simulable dressing $U(\bm\phi)$ then yields several further applications: dressed molecular ground states (\sref{sec:two-block}), Fubini--Study sampling for process benchmarking and infinite-temperature transport (\sref{sec:fs-apps}), and exact effective Hamiltonians trained from energy measurements (\sref{sec:decouple}). The unified method, the pruned ansatz with diagonal-metric updates and shift-rule oracle, is referred to as MinimalMetric (MM) throughout. Energies use exact double-precision statevector simulation; MinimalMetric runs use exact natural-gradient descent with the closed-form diagonal preconditioner of Eqs.~\eqref{eq:g-theta}--\eqref{eq:Minv} and the four-term shift rule of Eq.~\eqref{eq:shift-rule}, while baselines use the optimiser stated per panel (gradient descent and COBYLA for UCCSD, L-BFGS for the variational-dynamics baselines).

\subsection{Molecular electronic structure problem}\label{sec:vqe}

The active-leaves picture provides a clean way to incorporate physical symmetries: many symmetries (particle number, total spin projection $S_z$, point-group irreps under the Jordan--Wigner encoding) act diagonally on the computational basis, so the states in a given symmetry sector are exactly the computational-basis states whose bitstrings satisfy a set of linear or combinatorial constraints~\cite{Picozzi2022, Picozzi2026periodic, Bravyi2017}. Choosing $S$ to be such a sector therefore yields a pruned ansatz that exactly preserves the symmetry, with no penalty terms required and no leakage outside the sector~\cite{Gard2020}. In this implementation the user supplies an atom, a basis set, a target charge, spin projection and (optionally) irrep, and the routine \texttt{sieve\_states\_by\_symmetry} enumerates the basis states in the chosen sector and feeds them as active leaves to the compiler. For LiH in the STO-3G basis (CAS$(2e,3o)$, $n=6$ qubits), the particle-and-spin sector $\{N=2,\,S_z=0\}$ contains $k=9$ basis states; the compiler returns a pruned ansatz with $8$ free $R_y$ parameters and $32$ CNOTs, against $136$--$280$ CNOTs for the corresponding UCCSD circuit (point-group-screened singles-and-doubles to full pool). Imposing the ground-state point-group irrep reduces the pruned circuit further to $24$ CNOTs, the configuration plotted in Fig.~\ref{fig:vqe-bondcurves}, and the symmetry-adapted encoding of Refs.~\cite{Picozzi2022, Picozzi2026sae} to $6$. For BeH$_2$, H$_2$O and NH$_3$ the same procedure produces $\{N,S_z\}$ sectors of $36, 100, 100$ basis states respectively, all with linear-in-$k$ CNOT counts; Appendix~\ref{app:vqe-extra} tabulates the full resource progression.

For each molecule, two ans\"atze are compared on the same Hamiltonian in the STO-3G basis: the pruned ansatz, with active leaves chosen as the symmetry sector with the correct particle number, $S_z$ and ground-state irrep, optimized by natural-gradient descent (with derivative-free Rotosolve~\cite{Ostaszewski2021, Nakanishi2020} as a cross-check) on the active parameters; and the standard UCCSD ansatz~\cite{Romero2019} on the same orbital basis, optimized by gradient descent and COBYLA. The reference is full configuration interaction (FCI) within the chosen active space; the pruned ansatz uses the shift-rule analytic gradients of \sref{sec:shift} throughout. The multireference $\mathrm{H}_4$ square geometry is omitted from the figures because the CCSD(T)~\cite{Raghavachari1989} reference used in the pipeline is unreliable there; for the molecules shown the FCI reference is well-defined.

Figure~\ref{fig:vqe-bondcurves} shows equilibrium-geometry convergence traces for five molecules, together with the equilibrium-geometry gate-count and cost-evaluation breakdown (panels (b)--(c)). The pruned ansatz reaches sub-mHa accuracy on every closed-shell molecule in the panel. Equilibrium accuracies are recorded in Appendix~\ref{app:vqe-extra}.

\begin{figure*}[!htbp]
\centering
\includegraphics[width=\textwidth]{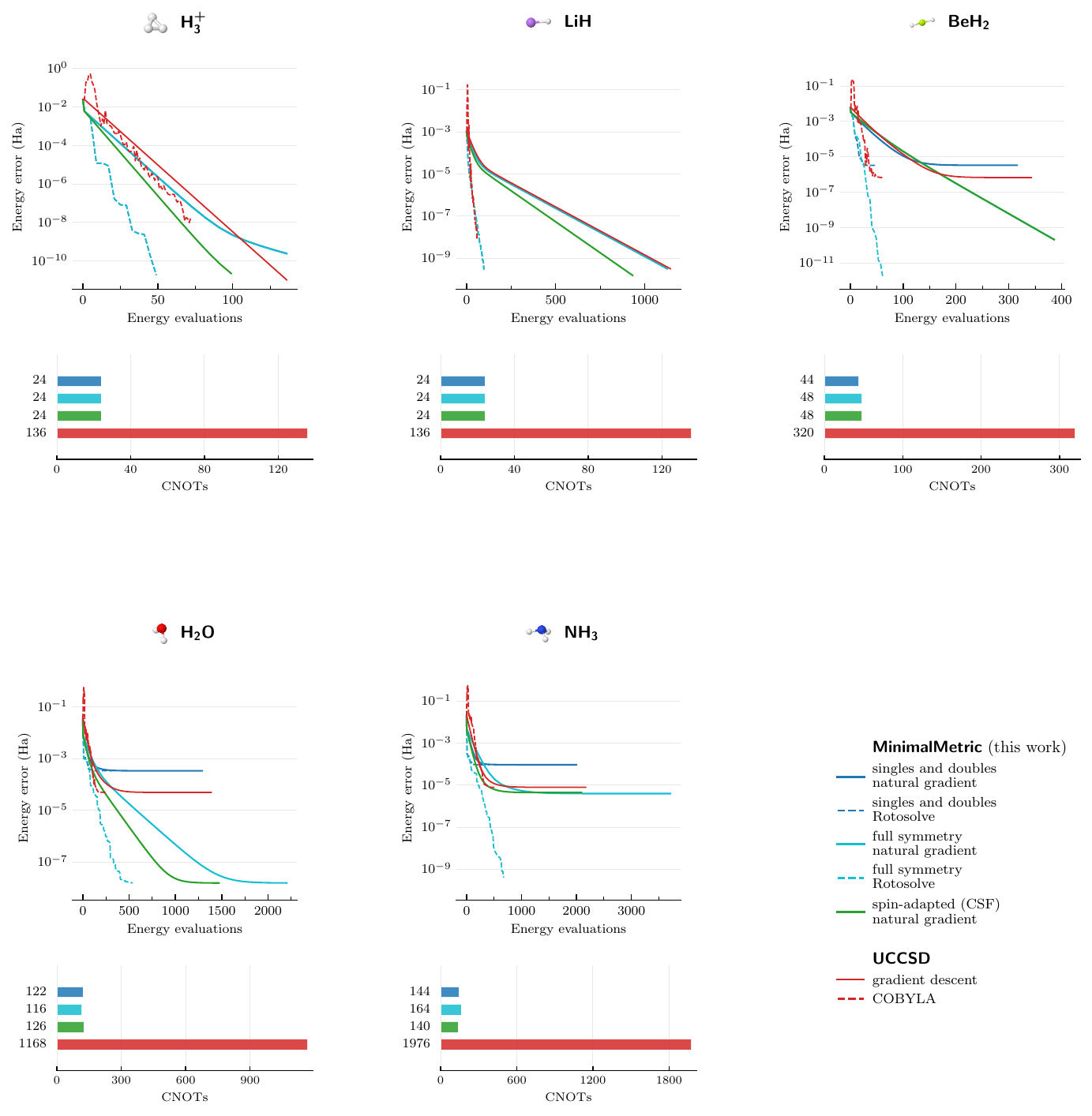}
\caption{Symmetry-adapted VQE benchmark in the STO-3G basis. Convergence at equilibrium geometry for five molecules: absolute energy error $|E-E_{\mathrm{ref}}|$ versus the cumulative number of energy evaluations on a log scale, with $E_{\mathrm{ref}}$ the CASCI (or FCI) reference. All curves share the same point-group-screened singles-and-doubles (S+D) parameter set unless noted. For MinimalMetric (MM, the symmetry-adapted pruned ansatz) two optimizers are compared on the identical ansatz, exact natural gradient (NG) and derivative-free Rotosolve, which reach the same S+D variational floor, confirming the residual error is the S+D truncation rather than an optimizer artefact; lifting MM to the full symmetry sector (MM full sym.) closes the gap to CASCI. UCCSD ($e^{T_1+T_2}$) is shown under gradient descent (GD) and COBYLA on the same S+D excitations. State-vector simulation throughout. The horizontal bars beneath each panel report the CNOT gate count of the corresponding ansatz (blue: MinimalMetric S+D; cyan: MM full sym.; green: spin-adapted; red: UCCSD).}
\label{fig:vqe-bondcurves}
\end{figure*}

\medskip\noindent\textbf{Exact spin adaptation.}\ The symmetries above act diagonally on the computational basis; total spin is the exception, its eigenfunctions (configuration state functions, CSFs) being entangled combinations of determinants, so the determinant-sector runs reach the singlet ground state only variationally, through transient spin contamination. Composing the tree with the fixed sector Schur transform $U_{\mathrm S}$, the determinant-to-CSF change of basis~\cite{Bacon2006, Burkat2025}, turns each active leaf into a CSF label. Restricting the leaves to the $S=0$ block then yields an exactly spin-adapted ansatz, with $\langle S^2\rangle=0$ at every parameter value and no penalty terms, on which the closed-form diagonal metric descends unchanged by Proposition~\ref{prop:invariance}. The singlet block is smaller by the Weyl dimension formula: absent spatial symmetry it halves the $\{N,S_z\}$ sector ($k=6/6/20/50/50$ against $9/9/36/100/100$), and the free-parameter count and metric dimension drop with it. This reduction is available for the low-symmetry molecules that dominate applications. Combined with the point-group screening of the other curves, it refines the full-symmetry ansatz of Fig.~\ref{fig:vqe-bondcurves} from $k=5/5/6/28/52$ to $4/4/6/18/28$ for H$_3^+$/LiH/BeH$_2$/H$_2$O/NH$_3$, reaching the same CASCI floor in up to $1.8\times$ fewer energy evaluations at comparable core two-qubit depth, all while holding $\langle S^2\rangle$ at machine zero. The determinant sector, by contrast, passes through spin contamination up to $\sim\!6\times10^{-6}$ (Fig.~\ref{fig:spin-contamination}). The fixed dressing, charged separately, is itself the dominant two-qubit cost of the dressed circuit on hardware (Appendix~\ref{app:schur-csf}). The transform is a fixed, parameter-free dressing, realised at the gate level as the quantum Paldus transform~\cite{Burkat2025}, and the same route reaches any spin sector. The construction, the gate-level circuit and its verification, and the resource table are given in Appendix~\ref{app:schur-csf}. That dressing cost can be avoided: Appendix~\ref{app:spin-tree} gives two gate-free constructions that reach the identical spin-adapted state with only the bare determinant-tree gates, paying for spin symmetry in polynomial classical work instead. One builds the tree directly in the CSF basis, keeping the exactly diagonal metric; the other ties a subset of the determinant tree's angles. The three routes prepare the same state and trade quantum against classical resources differently.\label{sec:schur-csf}

\subsection{Molecular real-time dynamics}\label{sec:mol-dynamics}
To complement the static VQE benchmark, dipole-kick dynamics were propagated for H$_3^+$, LiH, BeH$_2$, H$_2$O and NH$_3$ from their sector-restricted ground states under the same fixed molecular Hamiltonians, and MinimalMetric was compared against three hardware-efficient state-of-the-art baselines: projected-VQD on a hardware-efficient ansatz~\cite{Kandala2017} (pVQD-HEA), McLachlan variational real-time evolution~\cite{McLachlan1964} on the same ansatz (VarQRTE-HEA), and a second-order Suzuki--Trotter product formula~\cite{Suzuki1991, Childs2021, Childs2018}. The HEA depth is fixed per molecule at the knee of a dedicated repetition sweep. Each variational baseline is given its best initial state via a multistart L-BFGS state-preparation fit, whose residual sets a non-zero infidelity floor, and VarQRTE-HEA carries an additional ancilla for its Hadamard tests. Figure~\ref{fig:molecular-dynamics} (top panels) shows that MinimalMetric tracks exact evolution at numerical precision: the fourth-order Runge--Kutta integrator reaches the $\sim\!10^{-14}$ infidelity floor for all five systems, and first-order Euler tracks at $\sim\!10^{-11}$ (rising to $\sim\!10^{-9}$ for H$_3^+$ and NH$_3$). The variational hardware-efficient baselines plateau far above, even at this knee depth: pVQD-HEA at $\sim\!10^{-4}$--$10^{-5}$ and VarQRTE-HEA at $\sim\!10^{-5}$--$10^{-7}$ infidelity, several to ten orders of magnitude short of MinimalMetric. Their floors reflect an expressibility limit of the fixed hardware-efficient manifold, not a lack of depth: at the knee depth the manifold need neither contain the target state nor stay aligned with the exact trajectory, and enlarging it until it does drives the ansatz towards a $2$-design and the exponentially vanishing gradients of a barren plateau~\cite{Sim2019, Holmes2022, McClean2018, Leone2023, Larocca2025}. Trotter alone matches MinimalMetric's accuracy, but only by paying an enormous gate cost. The CNOT strips beneath each panel make this explicit: MinimalMetric reaches numerical precision at $50$--$418$ CNOTs, whereas the only accuracy-matching baseline (second-order Trotter) requires $1.0\times10^{5}$--$1.3\times10^{6}$ CNOTs, roughly three orders of magnitude more. The variational baselines reach their far coarser floors at two-qubit depth comparable to or greater than MinimalMetric's.

\begin{figure*}[!htbp]
\centering
\includegraphics[width=\textwidth]{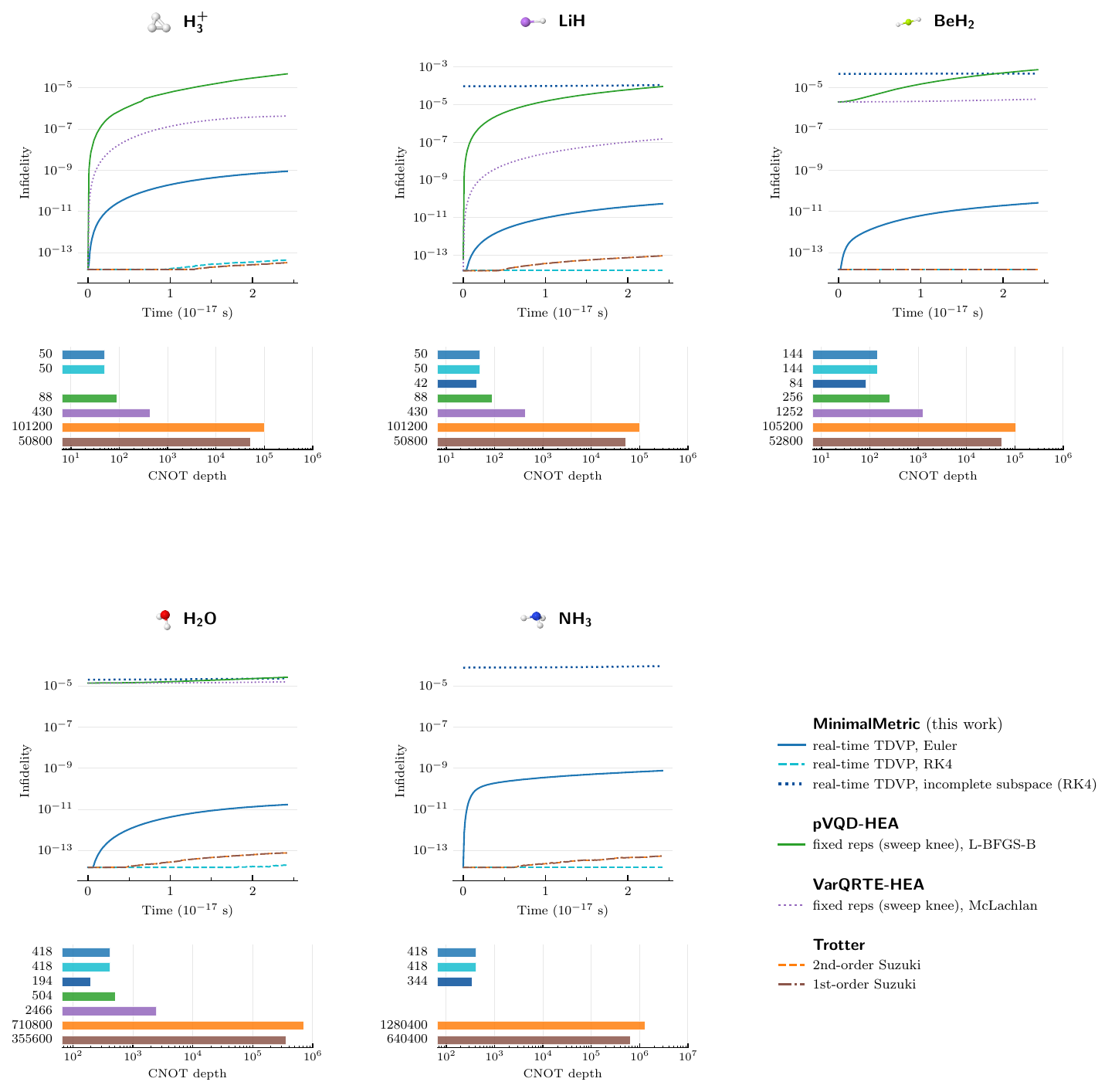}
\caption{Real-time molecular dipole-kick trajectories for H$_3^+$, LiH, BeH$_2$, H$_2$O and NH$_3$ at equilibrium geometry ($\kappa=0.05$, $T=1$; Euler and RK4 integrators). Each panel shows the trajectory infidelity against exact evolution on a log scale (top) over the coherent CNOT depth of each method (strip beneath, log scale). MinimalMetric (blue) tracks the exact trajectory to numerical precision; second-order Trotter (orange, dashed) is numerically exact; pVQD-HEA (green) and VarQRTE-HEA (purple), with per-molecule repetition count fixed at the knee of a dedicated sweep and the multistart L-BFGS state-preparation floor charged, plateau at $\sim\!10^{-4}$--$10^{-7}$ infidelity, orders of magnitude above MinimalMetric. MinimalMetric on a deliberately incomplete subspace (dark blue, dotted; four of the five molecules) tracks the trajectory to a controlled $\sim\!10^{-4}$--$10^{-5}$ floor at a fraction of the full-sector two-qubit depth (Appendix~\ref{app:restricted-subspace}). The coherent CNOT depth is the deepest circuit that must run without intermediate measurement: the per-step re-prepared circuit for the variational methods, and the full product formula for Trotter, which must be executed coherently end to end. MinimalMetric matches the numerically exact Trotter accuracy at $\sim10^{3}\times$ shallower depth.}
\label{fig:molecular-dynamics}
\end{figure*}

\medskip\noindent\textbf{Approximate dynamics on an incomplete subspace.} The runs above hand the compiler the full particle-number sector, so the reachable space contains the exact trajectory. Handing it instead a proper subset $S$ of that sector turns the same machinery, unchanged, into a controlled approximate propagator. The justification is structural: the pruned chart is complete on $\mathrm{span}(S)$, a coordinate sub-$\mathbb{C}\mathbf{P}^{k-1}$ and hence a K\"ahler submanifold, so the real-time flow~\eqref{eq:tdvp-flow} on the restricted chart is exactly Schr\"odinger evolution under the projected Hamiltonian $H_S=P_S H P_S$ from the normalised projection of the kicked state. The ansatz therefore commits no variational error beyond the subspace truncation itself, which is confirmed to $\leq 4\times10^{-15}$ infidelity (Appendix~\ref{app:restricted-subspace}). Choosing $S$ so that this truncation stays small, the reduced-subspace curves in Fig.~\ref{fig:molecular-dynamics} (dark blue, dotted) track the exact trajectory to a $\sim\!10^{-4}$--$10^{-5}$ floor, a reasonable approximation from an incomplete subspace, at a fraction of the two-qubit depth of the full-sector run ($42$--$344$ against $50$--$418$ CNOTs) and below that of the hardware-efficient baselines. Appendix~\ref{app:restricted-subspace} details the subspace construction, an a priori error certificate computable from the restricted trajectory alone, and the extension of the idea to a dressed, non-classically-simulable frame where the same controlled truncation would itself be genuinely quantum.

\subsection{Fermi--Hubbard quench dynamics}\label{sec:hubbard}

To validate the real-time variational machinery, the half-filled $1\!\times\!4$ and $2\!\times\!2$ Hubbard quenches at $U/t=4$ are evolved with the MinimalMetric integrator (first-order Euler and fourth-order Runge--Kutta around the shift-rule right-hand side, with the singular-coordinate initialiser) and compared against projected-VQD on the Hamiltonian variational ansatz (pVQD-HVA)~\cite{Wecker2015, Cade2020} and second-order Trotter baselines. Both quenches start from $U=0$ and are initialised in the symmetry-adapted noninteracting reference (for the open-shell $2\!\times\!2$ plaquette, the shell-pair singlet configuration), an exact eigenstate of the initial Hamiltonian that is supported on the pruned leaf set; every method therefore begins at machine-zero infidelity, ensuring a fair comparison. The pVQD-HVA baseline follows the literature protocol: its reference state is the same noninteracting ground state prepared by a Givens-rotation network~\cite{Kivlichan2018}, its layer count is fixed per case at the knee of a dedicated repetition sweep, and its residual multistart state-preparation fit is charged as a non-zero infidelity floor. Figure~\ref{fig:hubbard-dynamics} reports errors directly against exact diagonalisation: with the RK4 integrator the variational trajectory reaches final-time infidelity $3.3\times10^{-11}$ on the $1\!\times\!4$ chain and $8.0\times10^{-15}$ on the $2\!\times\!2$ plaquette at $T=1$ for the $N=500$ run ($dt=0.002$). The CNOT strip beneath shows the matched resource cost: MinimalMetric runs at a coherent depth of $142$ CNOTs per step, against the deeper per-step compute--uncompute circuits of pVQD-HVA and $5.6\times10^{4}$ for the single coherent second-order Trotter product formula, while reaching a final-time error many orders of magnitude lower than pVQD-HVA. The advantage extends to the total measurement budget: MinimalMetric's per-step shot noise averages down along the trajectory, whereas pVQD's per-step optimisation tolerance tightens as the timestep shrinks, so for any sufficiently fine integration MinimalMetric reaches a given accuracy at fewer total shots (Appendix~\ref{app:shot-cost}).

\begin{figure*}[!htbp]
\centering
\includegraphics[width=\textwidth]{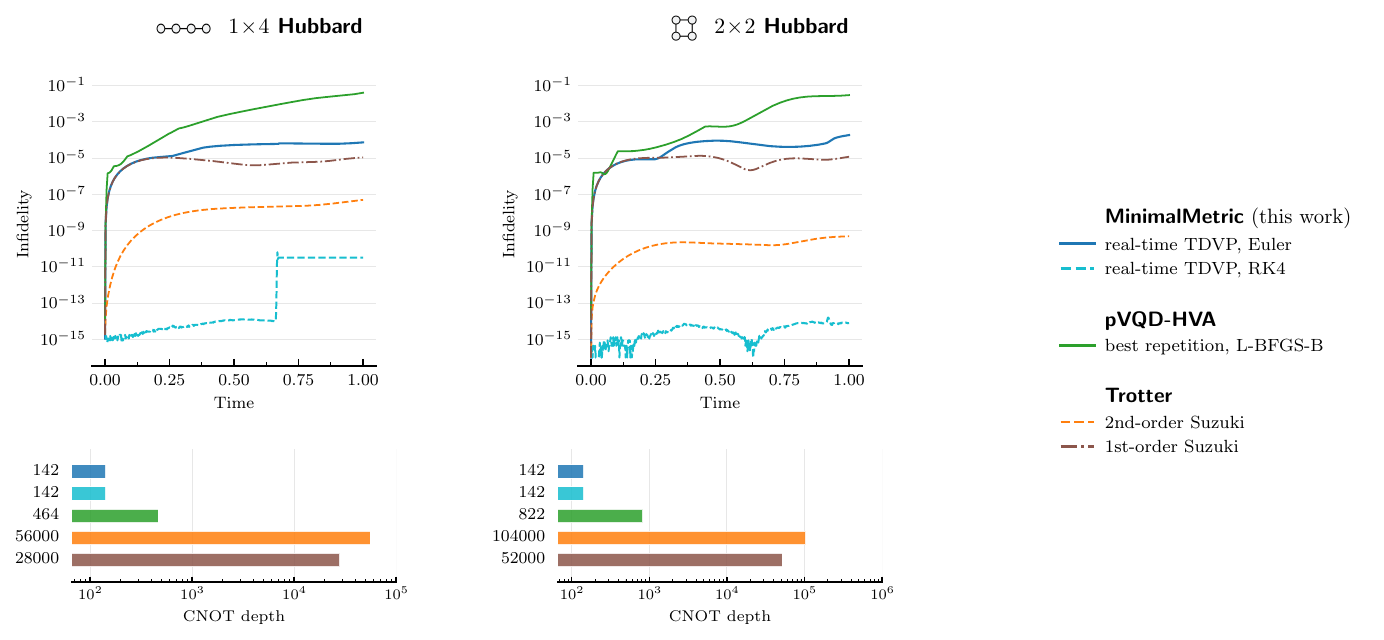}
\caption{Real-time Hubbard quench at $U/t=4$. Each panel shows the trajectory infidelity against exact diagonalisation (top, log scale) over the coherent CNOT depth of each method (strip beneath, log scale). Both quenches are initialised in the symmetry-adapted noninteracting reference, so all methods start at machine-zero infidelity. MinimalMetric (cyan, RK4) reaches final-time infidelity $3.3\times10^{-11}$ at $142$ CNOTs per step on the $1\!\times\!4$ quench ($T=1$, $N=500$, $dt=0.002$) and $8.0\times10^{-15}$ at $142$ CNOTs on the $2\!\times\!2$ quench, both outperforming pVQD-HVA (green, Givens-prepared reference, fixed repetition count at the knee: $4.0\times10^{-2}$ at $464$ CNOTs on $1\!\times\!4$, and $3.0\times10^{-2}$ at $822$ CNOTs on $2\!\times\!2$) and second-order Trotter ($5.6\times10^{4}$ CNOTs on $1\!\times\!4$ reaching $5.0\times10^{-8}$, $1.0\times10^{5}$ CNOTs on $2\!\times\!2$ reaching $4.9\times10^{-10}$, orange dashed).}
\label{fig:hubbard-dynamics}
\end{figure*}

\medskip\noindent\textbf{Composing with a non-classically-simulable dressing.}\ The benchmarks so far are classically simulable in the active-leaf count (Appendix~\ref{app:simulability}); they establish the bare ansatz as an efficient, exactly controllable primitive. A quantum--classical separation appears once the bare state is composed with a dressing that is classically hard to simulate,
\begin{equation}\label{eq:dressed-main}
\ket{\Psi(\bm\theta,\bm\omega,\bm\phi)}=U(\bm\phi)\,\ket{\psi(\bm\theta,\bm\omega)} ,
\end{equation}
for the relevant expectation values then involve the dressed Hamiltonian $\widetilde H(\bm\phi)=U^\dagger(\bm\phi)\,H\,U(\bm\phi)$, whose conjugation destroys the sparse, path-product structure that makes the bare loop efficient. The applications that follow exploit this composition, an efficient preparation-and-sampling primitive feeding a hard $U(\bm\phi)$; the classical hardness, when $U$ is deep enough to spread the efficiently-sampled support across the full Hilbert space, is that of simulating quantum dynamics (\sref{sec:discussion}).

\subsection{Two-block active-space dressing}\label{sec:two-block}

The pruned ansatz also serves as the exact state-preparation half of a two-block construction that targets correlation lying outside a chosen active space. The circuit splits into a state-preparation block $\ket{\psi}$ followed by a dressing block $U(\phi)$. The first is the symmetry-pruned binary tree of \sref{sec:vqe} with active leaves equal to a complete-active-space (CAS) sector, preparing the exact CAS ground state at fixed cost; the second is a unitary $U(\phi)=\exp\!\big(\sum_a \phi_a A_a\big)$ whose generators $A_a$ are the symmetry-adapted single and double excitations coupling the model space $P$ (the CAS) to its external complement $Q$, exactly the block-decoupling construction studied in \sref{sec:decouple}. The full pool of $P\!\to\!Q$ generators is screened by an ADAPT-style gradient criterion~\cite{Grimsley2019} down to the same number of parameters as a singles-and-doubles UCCSD ansatz in the identical symmetry-adapted encoding, so that every comparison below is at matched parameter count (the pool, screening and two-stage optimiser are detailed in Appendix~\ref{app:two-block}). This split also organises trainability. The exact state preparation is free of barren plateaus on the polynomially sized model space, and by the dressing invariance of Proposition~\ref{prop:invariance} it stays so for the core parameters at every dressing $\bm\phi$ (Appendix~\ref{app:no-bp}); any residual plateau is therefore confined to the dressing block, which the ADAPT screening keeps shallow and structured and which is warm-started at the exact CAS state. The split thus removes the plateau from the state-preparation step---where a monolithic hardware-efficient ansatz would incur it---and isolates the trainability-versus-hardness trade-off in the dressing alone.

Figure~\ref{fig:two-block} benchmarks this construction along the symmetric dissociation of a single water molecule, stretching both O--H bonds together from $R=1.5$ to $2.25\,\text{\AA}$ in steps of $0.25\,\text{\AA}$. This path interpolates from a predominantly dynamically-correlated regime near equilibrium to an increasingly static, strongly-correlated regime as the bonds stretch. All four geometries use one fixed full-space ansatz (ten qubits, with an inner $(6,5)$ complete-active-space state-preparation block) matched to a $48$-parameter UCCSD reference in the identical symmetry-adapted encoding. At every bond length the two-block ansatz (TwoBlock-MM, blue) reaches chemical accuracy, staying between $1.2\times10^{-4}$ and $2.7\times10^{-4}$\,Ha across the sweep, while using $2.4$--$3.3\times$ fewer CNOTs than UCCSD at every geometry. Undressed, the bare CAS tree (cyan) plateaus well above chemical accuracy near equilibrium ($9.5\times10^{-3}$\,Ha at $R=1.5\,\text{\AA}$) but improves monotonically as the bonds stretch and the static correlation it captures grows, narrowing to $2.4\times10^{-3}$\,Ha at $R=2.25\,\text{\AA}$; the dressing block therefore carries the residual dynamic correlation that dominates at the shorter bond lengths.

The robustness gap opens at the strongly-correlated stretched geometries. Both UCCSD baselines, the full symmetry-adapted singles-and-doubles pool (red solid) and its ADAPT-screened variant (light red dashed), track the two-block accuracy near equilibrium ($\sim\!2\times10^{-4}$\,Ha at $R\le1.75\,\text{\AA}$), but degrade past chemical accuracy as the bonds stretch, reaching $2.8$--$4.7\times10^{-3}$\,Ha at $R=2.0$ and $2.25\,\text{\AA}$, more than an order of magnitude above the two-block ansatz. The two-block ansatz, anchored by the exact CAS state preparation, descends monotonically and stays at chemical accuracy across the entire sweep, reaching it in all four geometries. This is an iso-accuracy, lower-depth result rather than an optimiser-efficiency claim: every method uses the same fixed gradient-descent schedule ($\sim\!5\times10^{4}$ energy evaluations), and the comparison is drawn at equal parameter count, where the two-block circuit is consistently shallower in two-qubit depth and more robust in the strongly-correlated regime.

\medskip\noindent\textbf{Exactly spin-adapted two-block.}\ The construction admits an exactly spin-pure variant (TwoBlock-MM-spin, green). Its state-preparation block is the gate-free spin-adapted CSF tree of \sref{sec:schur-csf}, prepared natively in the symmetry-adapted encoding (each intermediate-JW determinant is mapped to its encoded correspondent; Appendix~\ref{app:spin-tree}), so $\langle S^2\rangle=0$ holds identically and the CAS tree carries only $f_S-1=17$ free angles instead of the $27$ of the determinant sector; the dressing generators are correspondingly restricted to the spin-conserving (singlet, $S^2$-commuting) generalised single and double excitations, so the dressed state $U(\bm\phi)\ket\psi$ stays an exact singlet at every $(\bm\theta,\bm\phi)$. A spin-conserving double is a multi-Pauli operator whose two-qubit synthesis cost exceeds that of a spin-orbital excitation, so the spin-adapted pool is screened by a cost-aware ADAPT criterion (gradient per CNOT) to matched two-qubit cost rather than matched parameter count, and each rotation is realised by a single first-order Trotter step whose error at the small optimised angles ($|\bm\phi|_\infty\lesssim10^{-2}$) is negligible ($\Delta E\sim10^{-10}$\,Ha, residual $\langle S^2\rangle\sim10^{-9}$; Appendix~\ref{app:two-block}). At CNOT count equal to or below the determinant two-block ($1140/1048/776/1024$ against $1196/1108/896/1048$ across the sweep) the spin-adapted ansatz reaches chemical accuracy at every geometry with $\langle S^2\rangle\le5\times10^{-18}$, and is more accurate than the determinant two-block in the strongly-correlated regime ($8.4\times10^{-5}$ and $2.4\times10^{-5}$\,Ha at $R=2.0$ and $2.25\,\text{\AA}$, against $2.7\times10^{-4}$ and $2.4\times10^{-4}$), the determinant sector passing through the transient spin contamination that the spin-adapted ansatz never incurs.

\begin{figure*}[!htbp]
\centering
\includegraphics[width=0.90\textwidth]{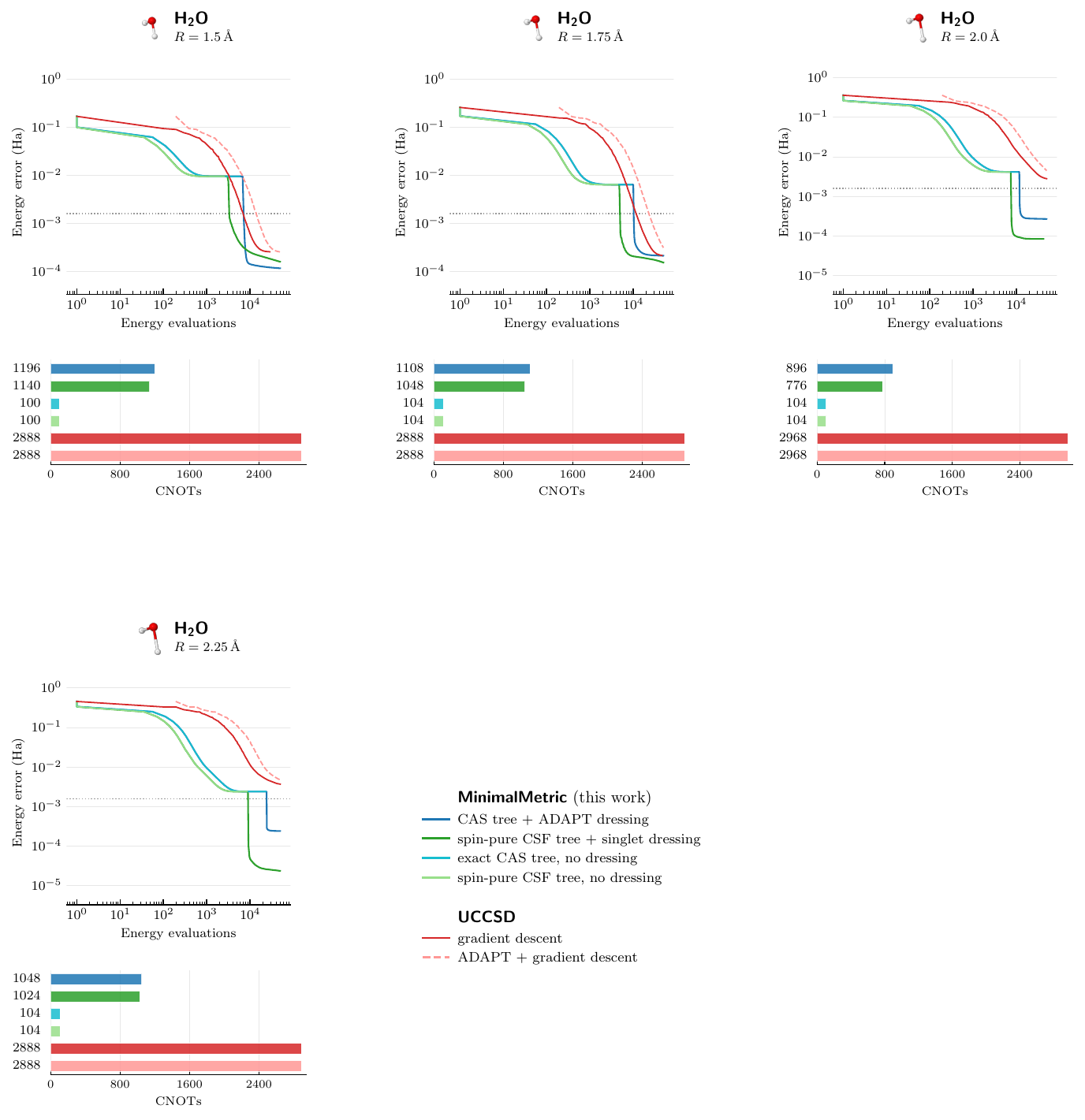}
\caption{Two-block active-space dressing along the symmetric stretch of water in the STO-3G symmetry-adapted encoding, at four O--H bond lengths $R=1.5$--$2.25\,\text{\AA}$ spanning the crossover from dynamic to static correlation. Each panel shows the absolute energy error $|E-E_{\mathrm{ref}}|$ (Ha) against the cumulative number of energy evaluations on log--log axes, with $E_{\mathrm{ref}}$ the exact (FCI) reference and the CNOT count of each ansatz in the strip beneath; the dotted grey line marks chemical accuracy ($1.6\times10^{-3}$\,Ha). All methods share one fixed full-space ansatz (ten qubits, inner $(6,5)$ CAS state-preparation block) at matched parameter count: the two-block ansatz with ADAPT-screened dressing (TwoBlock-MM, blue, this work); the undressed CAS tree (cyan); and UCCSD in the same encoding, both as the full singles-and-doubles pool (red solid) and its ADAPT-screened variant (light red dashed), each under gradient descent. The two-block ansatz reaches chemical accuracy at every bond length using $2.4$--$3.3\times$ fewer CNOTs than UCCSD and, unlike both gradient-descent UCCSD baselines, which degrade past chemical accuracy as the bonds stretch, stays at chemical accuracy across the dissociation. The exactly spin-adapted variant (TwoBlock-MM-spin, dark green) replaces the state preparation with the spin-pure CSF tree (its undressed counterpart, light green, converges to the same CAS floor as the determinant tree but from fewer free angles) and the dressing with spin-conserving (singlet) generalised excitations, selected by cost-aware ADAPT to a two-qubit cost matched to (or below) the determinant two-block; it holds $\langle S^2\rangle\le5\times10^{-18}$ throughout and reaches chemical accuracy at every geometry, more accurately than the determinant two-block in the stretched regime.}
\label{fig:two-block}
\end{figure*}

\subsection{Fubini--Study sampling: process fidelity and infinite-temperature transport}\label{sec:fs-apps}

The tree ansatz also samples ensembles. Drawing its parameters from the Fubini--Study volume measure $\sqrt{\det g^{\mathrm{FS}}}\,d\bm\theta\,d\bm\omega$ produces states distributed exactly as the unitary-invariant (sector-Haar) measure on a chosen symmetry sector. Throughout, the active-leaf count $k=|S|$ of \sref{sec:pruning} also serves as the sector's Hilbert-space dimension. The projective determinant factorises along the tree, reducing the draw to one independent $\mathrm{Beta}$-distributed angle per node~\cite{Zyczkowski1994, Mezzadri2007, Tilma2002}, so the sampler is rejection-free and exact at $\mathcal O(k)$ cost per sample even when the sector is exponentially large (Appendix~\ref{app:fs-design} gives the construction and a short proof). Being the Haar ensemble, it is an exact state design~\cite{Dankert2009}: its first moment is the maximally mixed sector state and its second moment supplies closed-form averages that computational-basis sampling, only a one-design, cannot reproduce. The sampling is classical; the quantum content is the process applied to each sample, complementing randomized-measurement schemes that instead fix the input state and randomise the readout~\cite{Huang2020, Elben2023}.

Figure~\ref{fig:haar-echo} uses this ensemble to benchmark a quantum process against its ideal target by the echo return probability $f(p)=|\braket{\psi(p)|U_{\mathrm{exact}}^\dagger U_\star|\psi(p)}|^2$ averaged over the sector (Appendix~\ref{app:fs-protocols}). Here $U_\star$ is a second-order Trotterisation of $e^{-iHT}$, so $\overline{1-f}$ is the state-dependent coherent Trotter error averaged over the half-filled Hubbard sectors ($k=36$ and $400$). The FS estimator tracks the exact sector average (known in closed form from the two-design identity, Eq.~\eqref{eq:Favg}) down to its Monte-Carlo floor across all Trotter depths, while the standard computational-basis-state average carries a persistent sector bias that widens as the sector grows from $k=36$ to $400$: precisely the off-diagonal coherent error a one-design misses. The single (ground) state echo lies far below every average: a near-eigenstate accumulates only a global phase that the echo cancels, so it registers almost no error and is not a representative probe of the averaged process.

The same ensemble estimates infinite-temperature transport, where the sector state is maximally mixed and the one-design identity turns FS sampling into a pure-state Monte-Carlo estimator of sector traces. The informative observables are dynamical correlators $C_{AB}(t)=\operatorname{Tr}[P\,A(t)\,B\,P]/k$ with $A(t)=U^\dagger(t)\,A\,U(t)$, measured by an ancilla Hadamard test (Appendix~\ref{app:fs-protocols}). Figure~\ref{fig:thermal-fs} reports the nearest-neighbour $S^zS^z$ correlator under the term-level Trotter dynamics on the $1\!\times\!4$ and $1\!\times\!6$ chains: FS sampling reproduces the exact Trotter target within error bars and, by quantum typicality~\cite{Sugiura2012, Goldstein2006, Popescu2006}, at lower single-sample variance than basis-state averaging. As $t$ grows the Heisenberg operator $U^\dagger(t)AU(t)$ spreads and a classical representation becomes expensive~\cite{Prosen2007, Schollwock2011}. On the sectors shown the exact sampler functions as an unbiased, low-variance estimator validated directly against exact diagonalisation; the clean quantum-advantage form of the construction, the same typicality composed with a support-leaving dressing on a polynomially-sized model space, is developed in \sref{sec:decouple} and \sref{sec:discussion}.

\begin{figure*}[!htbp]
\centering
\includegraphics[width=\textwidth]{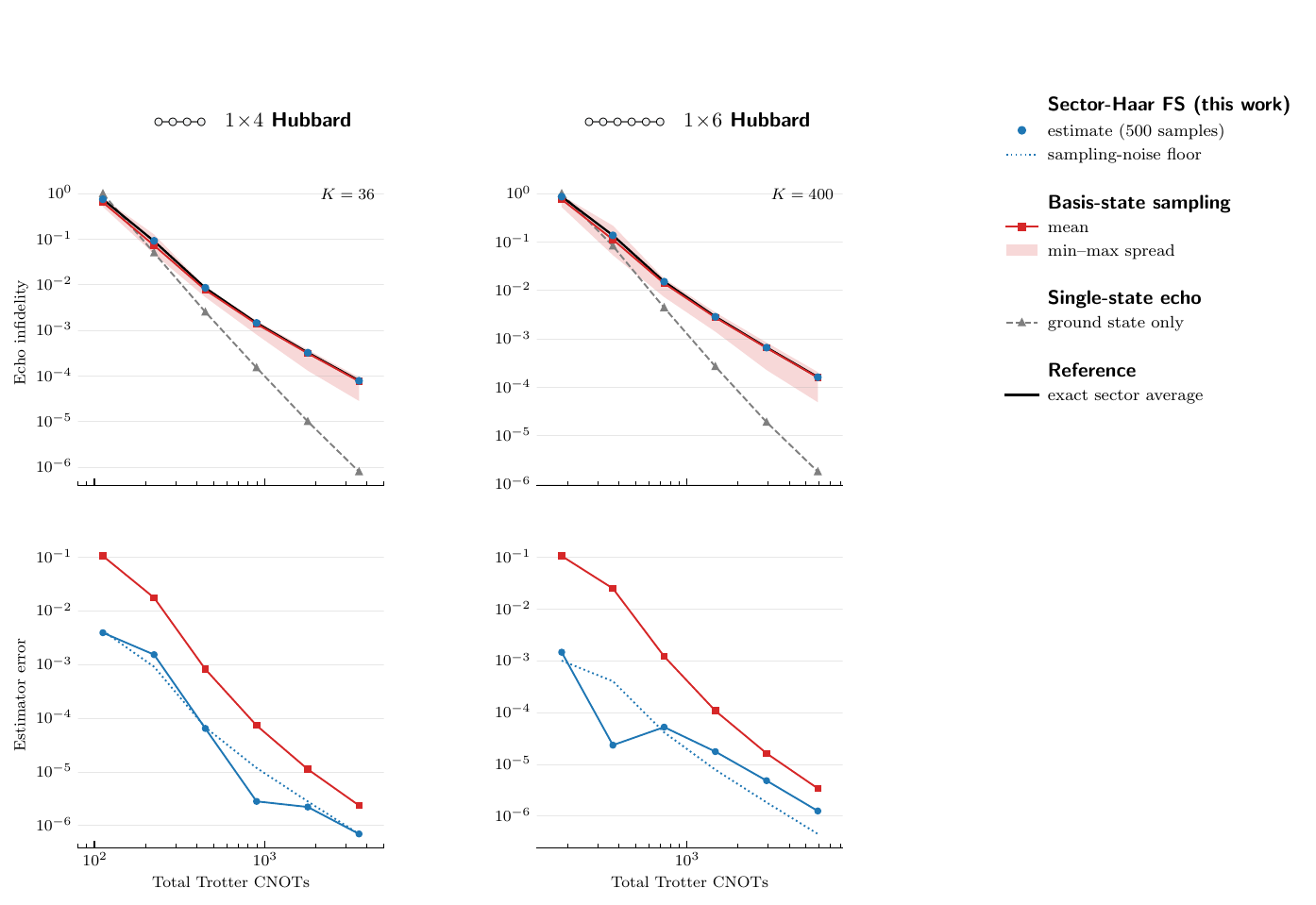}
\caption{Sector-averaged echo infidelity $1-F$ at evolution time $T=1$, estimated by sector-Haar (FS) sampling. \textbf{Top row}: for the half-filled $1\!\times\!4$ ($k=36$, left) and $1\!\times\!6$ ($k=400$, centre) Hubbard sectors, the FS estimator (blue, $500$ samples) tracks the exact sector average (black) across Trotter depths, while basis-state sampling (red) shows a persistent sector bias. The single-state (ground-state) echo (grey) sits far below every average: the echo cancels its global phase, so a near-eigenstate sees almost no coherent error and is the least representative probe of the averaged process error. \textbf{Bottom row}: the corresponding estimator error $|\,\overline{1-f}-(1-f_{\mathrm{avg}})\,|$ against the exact sector average, for each case. The FS error rides its own sampling-noise floor (dotted), i.e.\ it is unbiased down to Monte-Carlo noise, whereas basis-state averaging plateaus at a finite sector bias; the gap widens as the sector grows from $k=36$ to $k=400$. The single-state echo is not an estimator of the sector average and is omitted from the error panels.}
\label{fig:haar-echo}
\end{figure*}

\begin{figure*}[!htbp]
\centering
\includegraphics[width=\textwidth]{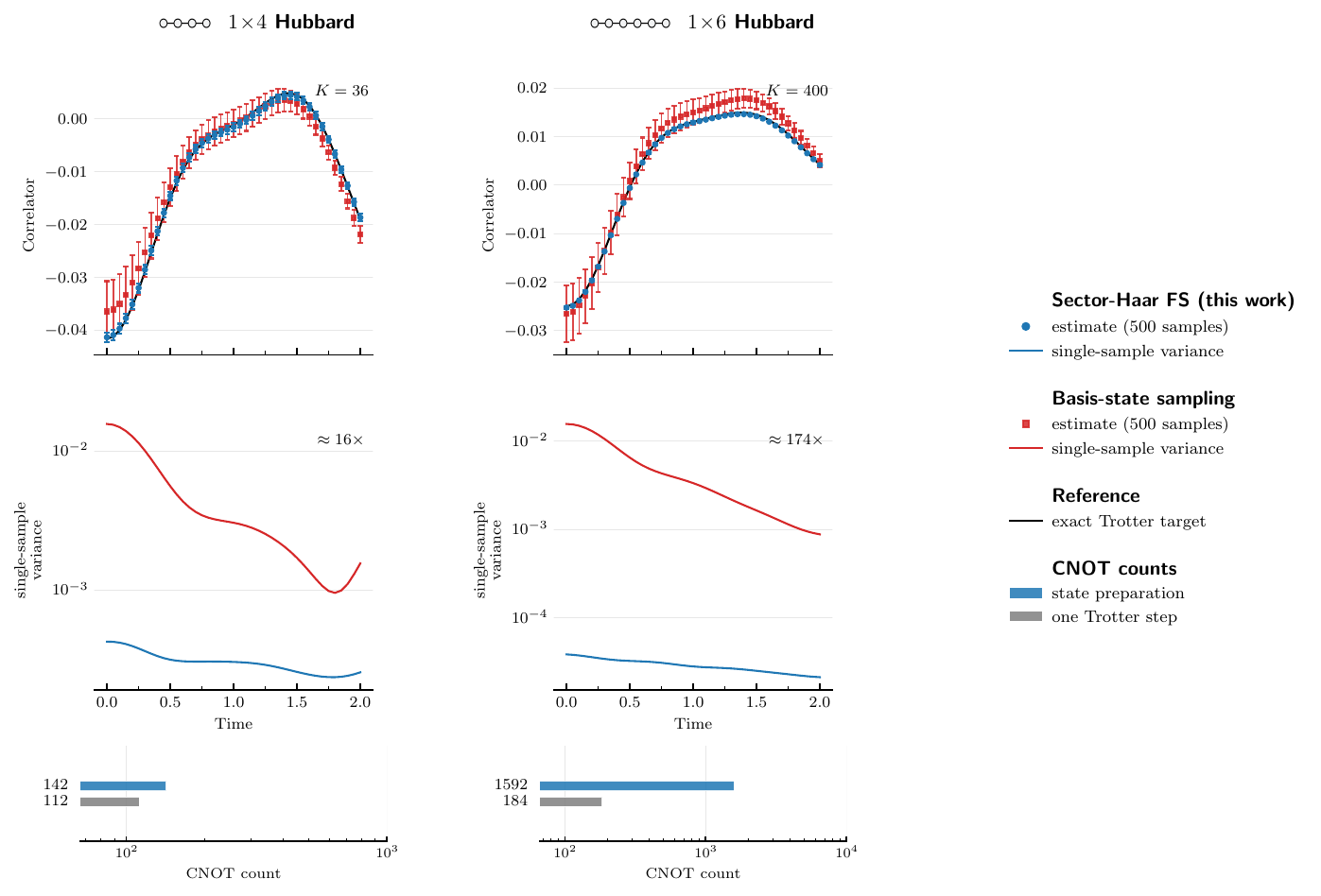}
\caption{Infinite-temperature nearest-neighbour $S^zS^z$ correlator $C_{AB}(t)$ from sector-Haar (FS) sampling on the half-filled $1\!\times\!4$ ($k=36$) and $1\!\times\!6$ ($k=400$) Hubbard sectors at $U/t=4$, under the term-level second-order Trotter dynamics. \textbf{Top of each cell}: the correlator estimated by FS sampling (blue) against the exact sector trace (black), reproduced within error bars across the evolution. \textbf{Middle}: the single-sample variance of the FS estimator versus the computational-basis-state estimator (log scale); both are unbiased, but the FS estimator concentrates by typicality and so attains a given accuracy at fewer samples. \textbf{Strip}: the coherent CNOT depth of the Trotter evolution. The per-sample quantity is read out by the ancilla Hadamard-test protocol of Appendix~\ref{app:fs-protocols}.}
\label{fig:thermal-fs}
\end{figure*}

\subsection{Exact effective Hamiltonians from energy measurements}\label{sec:decouple}

As a final application, the same machinery is used to decouple a model space exactly rather than to prepare a ground state, a quantum-native Schrieffer--Wolff transformation~\cite{Bravyi2011}, and to ask when the residual coupling can be driven not merely down but to machine zero. Lithium hydride is taken in a symmetry-adapted active space, with a compact reference model space $P$ singled out inside the full active-space sector and $Q$ its complement within that sector. The block is dressed by the same unitary $U(\phi)=\exp(\sum_a \phi_a A_a)$ used for the molecular ansatz of \sref{sec:vqe}, with $A_a$ the symmetry-adapted generalised single and double excitations connecting $P$ to $Q$.

Exact decoupling is reachable whenever the dressing pool is tangent-complete: because each active space spans at most three spatial orbitals, no model-space determinant differs from a complement determinant by more than a double excitation, so the screened singles-and-doubles pool exhausts the off-block tangent space and its dimension equals $|P|\,|Q|$ exactly ($n_\phi=6$ for CAS$(2e,3o)$ and $n_\phi=20$ for CAS$(4e,3o)$). A dressing built from a tangent-complete pool can rotate the block-off-diagonal coupling $QU^\dagger HUP$ all the way to zero, and the energy-only FS-sampled gradient surrogate $C(\phi)=\mathbb{E}_{\chi\sim\mathrm{FS}(P)}\big[\sum_a|\partial_{\phi_a}\langle\chi|U^\dagger(\phi)HU(\phi)|\chi\rangle|^2\big]$, whose global minimum is attained precisely when $QU^\dagger HUP=0$ under the tangent-completeness condition, finds that rotation (Appendix~\ref{app:sw-surrogate}). The optimiser never sees $Q$, the target spectrum, or the leakage $\|QU^\dagger HUP\|_F$; the leakage and the effective-model energy $E_0^{\mathrm{eff}}=\lambda_{\min}(PU^\dagger HUP)$ are recorded classically as held-out diagnostics only.

Figure~\ref{fig:decouple} sweeps the Li--H bond length over $R\in[1.0,3.2]\,\text{\AA}$, spanning the equilibrium and the stretched, strongly multireference regime, for the two active spaces (CAS$(2e,3o)$, $k=5$; CAS$(4e,3o)$, $k=9$). The trained dressing (blue) drives the leakage from its undressed value of order $10^{-3}$ and the effective-model energy error $|E_0^{\mathrm{eff}}-E_0|$ from order $10^{-4}$--$10^{-2}$ down in lockstep toward the machine-precision floor. For CAS$(4e,3o)$ both reach $\sim\!10^{-16}$ (leakage) and $\sim\!10^{-15}$ (energy) at every bond length; for CAS$(2e,3o)$ they do so at all but the most stretched geometry, where the energy-only fit leaves a residual $\sim\!10^{-11}$ leakage and $\sim\!5\times10^{-9}$ energy error, still far below chemical accuracy. The dressed block thus reproduces the exact sector ground state. The reconstruction succeeds at a fixed, manageable parameter count ($n_\phi=6$ and $20$) and a fixed circuit cost, state preparation plus one dressing unitary, with no access to the complement space that the effective Hamiltonian integrates out. This is the composition identified in \sref{sec:discussion} as a possible source of quantum advantage. The objective $C(\bm\phi)$ is a Fubini--Study typicality average of the dressed Hamiltonian over the polynomially-sized model space $P$, on which the exact sampler is efficient. Evaluating it classically requires propagating $U(\bm\phi)$, which carries $P$ into the exponentially larger complement $Q$ and is intractable once the dressing is deep, whereas the device pays only the $\mathcal O(k)$ preparation, one dressing, and a handful of typicality samples.

\begin{figure*}[!htbp]
\centering
\includegraphics[width=\textwidth]{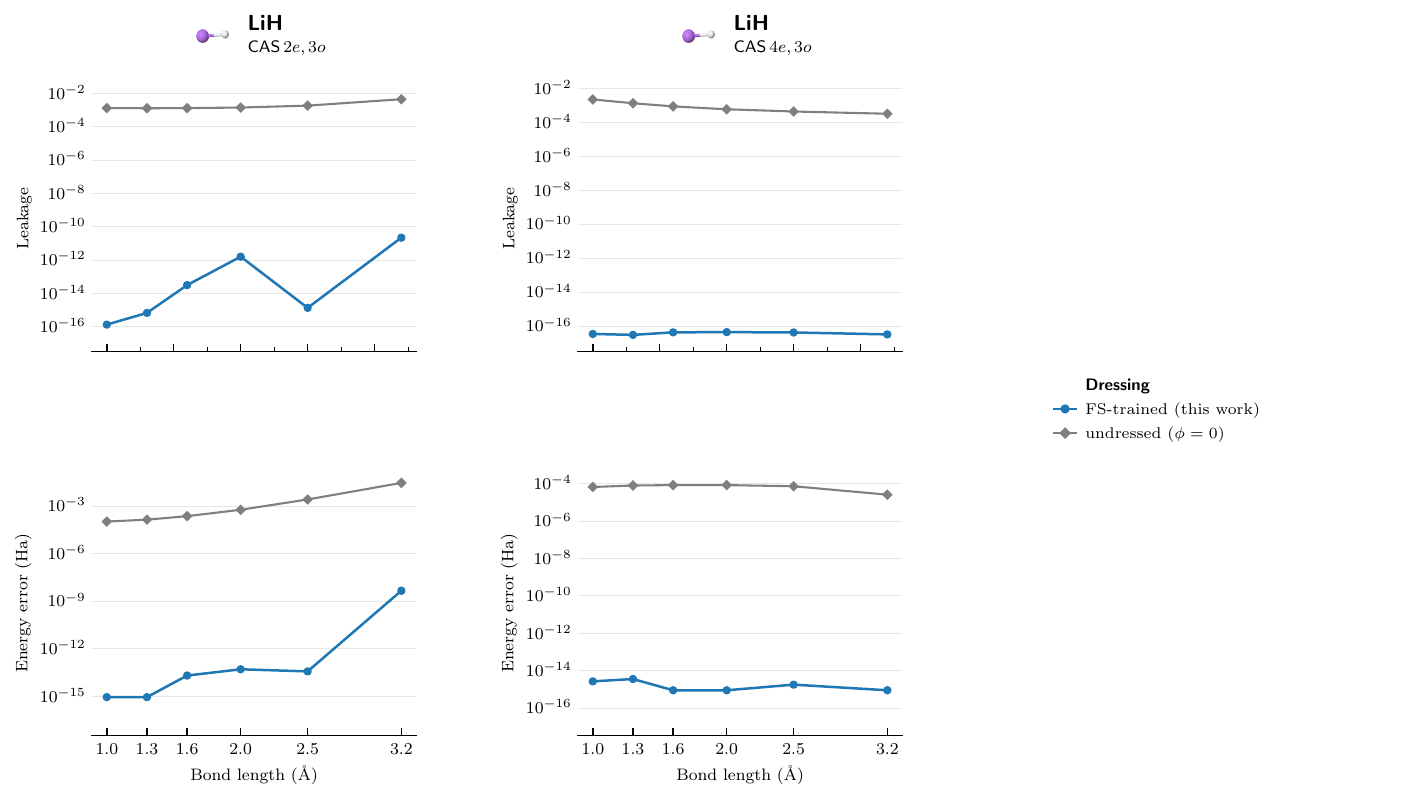}
\caption{Exact block decoupling of lithium hydride trained from energy measurements alone, across Li--H bond lengths $R\in[1.0,3.2]\,\text{\AA}$ for two active spaces (CAS$(2e,3o)$, $k=5$; CAS$(4e,3o)$, $k=9$). \textbf{Top}: block-off-diagonal leakage $\|QU^\dagger HUP\|_F^2/k_P$ for the undressed block ($\phi=0$, grey) and the FS-surrogate-trained dressing (blue, this work). \textbf{Bottom}: the effective-model energy error $|E_0^{\mathrm{eff}}-E_0|$, with $E_0^{\mathrm{eff}}=\lambda_{\min}(PU^\dagger HUP)$ the dressed $P$-block ground-state energy and $E_0$ the exact sector ground state. Because the screened singles-and-doubles pool spans the full off-block tangent space ($n_\phi=|P|\,|Q|$), the energy-only objective drives both quantities to machine precision, uniformly for CAS$(4e,3o)$, and at all but the most stretched CAS$(2e,3o)$ geometry, where a $\sim\!10^{-9}$ residual remains.}
\label{fig:decouple}
\end{figure*}

\FloatBarrier

\section{Discussion}\label{sec:discussion}

We have presented a hardware-efficient variational ansatz on $n$ qubits, built from a binary tree of uniformly controlled rotations, and shown that it has two complementary properties: an exactly diagonal Hilbert-space pullback metric, which makes quantum natural gradient, variational imaginary- and real-time evolution, and exact unitary-invariant (Haar) sampling implementable at no extra circuit cost; and a clean gauge interpretation of its redundant parameters when the target state is sparse, which yields a pruning compiler with linear-in-$k$ CNOT count. The exact metric, the pruning compiler, and the symmetry-adapted active-leaves construction can be combined freely: a symmetry-pruned chemistry ansatz with exact natural-gradient updates is, for instance, only a few lines away.

Conceptually, Theorem~\ref{thm:diag-metric} reduces the entire family of metric-aware variational quantum algorithms on this ansatz to a single computational primitive: parameter-shift evaluation of the Euclidean energy derivatives, followed by Hadamard division by the diagonal entries $(w_i,p_j)$. Quantum natural gradient and variational imaginary-time evolution use the resulting metric gradient directly as a descent direction; real-time evolution applies the projective K\"ahler rotation $-J/(2\hbar)$ to the same metric gradient, implemented by one bottom-up and one top-down pass through the active binary tree (Appendix~\ref{app:tdvp-rhs}). This eliminates the $\mathcal O(P^2)$ auxiliary metric circuits and the $\mathcal O(P^3)$ classical inversion required by generic-ansatz formulations~\cite{McArdle2019, Yuan2019}, and replaces them with $\mathcal O(P)$ classical work per step on top of the same Euclidean parameter-shift gradient. There is no separate measurement path for real-time evolution: the same gradient oracle, the same diagonal preconditioner, and the tree representation of $J$ cover both flows. The benchmarks provide consistent empirical confirmation: integrating the real-time MinimalMetric equations around the shift-rule RHS reproduces the exact unitary $\exp(-iHt)$ to integrator order. On the $1\!\times\!4$ half-filled Hubbard quench ($N=500$, $dt=0.002$, $T=1$) fourth-order Runge--Kutta reaches final-time infidelity $3.3\times 10^{-11}$ against $7.4\times 10^{-5}$ for first-order Euler; both schemes integrate the identical diagonal-metric RHS, so the gap is precisely the smaller per-step truncation error of the higher-order integrator (\sref{sec:hubbard}). On the dipole-kick molecular dynamics Runge--Kutta likewise tracks exact propagation at the $\sim\!10^{-14}$ infidelity floor across all five molecules (\sref{sec:mol-dynamics}).

The construction is, however, classically simulable in this bare form: the prepared state, its energy, its diagonal metric and its Fubini--Study sample ensembles all admit a closed-form description of cost polynomial in the active-leaf count (Appendix~\ref{app:simulability}). This is not specific to the tree. Appendix~\ref{app:simulability} promotes the observation to a general exact dequantization theorem: for any variational family whose states are supported on an enumerable polynomial-size set of basis states in some efficiently describable frame, with classically computable amplitudes and derivatives, and any frame-sparse observable, the entire metric-aware loop, energy, quantum geometric tensor, natural-gradient, imaginary- and real-time evolution, measurement sampling, sector-Haar averages, and the leakage certificate of Eq.~\eqref{eq:df-bound-full}, is classically computable exactly in polynomial time (Theorem~\ref{thm:dequantization}). It remains so under Clifford, matchgate, constant-depth-local or polynomial-bond-dimension dressings (Proposition~\ref{prop:closure}). The boundary therefore constrains not just this construction but a structural class that includes sparse-state-preparation and fixed-Hamming-weight ans\"atze in the literature~\cite{Li2025, Luo2025, Monbroussou2025}, and it connects the exact diagonal geometry here to the broader principle that the structure enabling trainability also enables classical surrogates~\cite{VandenNest2011, Goh2023, Cerezo2025sim}. Its quantum value is accordingly that of an exactly controllable preparation-and-sampling primitive, whose composition with a dressing that is classically hard to simulate is where genuinely quantum content can enter. The relevant expectation values then involve the dressed Hamiltonian $\widetilde H=U^\dagger H U$, whose conjugation destroys the sparse structure that makes the bare loop efficient and which by Corollary~\ref{cor:advantage} must simultaneously escape stabiliser, Gaussian, light-cone and tensor-network contraction; shallow, Clifford or Gaussian dressings do not suffice. This composition was exercised across four tasks built on one ingredient, Fubini--Study preparation or sampling of the bare ansatz feeding a hard $U(\bm\phi)$: dressed molecular ground states reaching chemical accuracy at $2.4$--$3.3\times$ fewer CNOTs than UCCSD and remaining at chemical accuracy where gradient-descent UCCSD degrades past it (\sref{sec:two-block}); sector-Haar process benchmarking, where exact-design sampling removes the sector bias that computational-basis averaging carries; infinite-temperature transport estimated from the same ensemble (\sref{sec:fs-apps}); and exact block decoupling driven to machine precision from energy measurements alone (\sref{sec:decouple}).

All four tasks compose the efficiently-prepared reference $\ket\psi$ with a dressing $U(\bm\phi)$, and they point to quantum advantage on two complementary footings. One is the simulation of quantum dynamics~\cite{Feynman1982, Lloyd1996, Daley2022}, generically hard to reproduce classically. Fubini--Study typicality over a polynomially-sized model space $P$ (where the exact sampler is efficient and an exact two-design), dressed by a deep $U$ that carries $P$ into an exponentially larger space, yields the model-space average $\mathrm{Tr}_P[U^\dagger O U]/k$: the device prepares the $\mathcal O(k)$ sample, applies $U$, and reads one ancilla, whereas a classical evaluation must propagate the deep $U$ over the full space and is intractable once $U$ scrambles $P$. The advantage belongs to the composed circuit as a whole; the structured reference is the indispensable input the dressing acts on, even though estimating native observables of $\ket\psi$ alone is classical. The decoupling objective (\sref{sec:decouple}) instantiates this composition; in its purest form $U$ is fixed real-time evolution, resting the claim on dynamics simulation alone. The incomplete-subspace propagator of \sref{sec:mol-dynamics} composes with the same dressing (Appendix~\ref{app:restricted-subspace}), so that approximate time evolution on a small model space acquires the same footing. The strongly-correlated ground states of \sref{sec:two-block} are a second promising target, on a practical footing. The overlap problem of Ref.~\cite{Lee2023} is precisely a shortage of efficiently-preparable references that keep appreciable overlap with the true ground state once single-reference coupled cluster~\cite{Raghavachari1989} degrades, and the construction is built to supply one: a symmetry-adapted, multireference reference, prepared and controlled exactly, feeding the same hard dressing in the regime where the density-matrix renormalization group~\cite{ChanSharma2011} and selected configuration interaction~\cite{Holmes2016} grow expensive. A provable asymptotic separation there remains open, but it turns on exactly the reference overlap this construction targets. Because the bare computations are classically checkable, each demonstration validates the underlying mechanism exactly.

Several directions remain open. The linear-in-$k$ scaling itself is unconditional (Proposition~\ref{prop:linear} gives $N_{\mathrm{cx}}\le 2nA_s=O(n^2k)$ for any qubit ordering), so what remains is the tight constant: by the run-count identity of Appendix~\ref{app:cnot-structure} (Lemma~\ref{lem:runcount}), the structure-dependent bound Eq.~\eqref{eq:cnot-bound-As}, proved here for $k=2$ and verified numerically under the inactive-node-maximising reordering, is equivalent to the single combinatorial inequality $\sum_{q,\beta}R_{q,\beta}\le E_s-1$ (Eq.~\eqref{eq:runcount-conjecture}). Proving it for general $k$, together with whether a polynomial-time qubit ordering attains it (the greedy inactive-node heuristic does not for $n\ge 9$), would close the constant-factor gap and match the tight worst-case bounds for dense state preparation. A companion question concerns connectivity: the pruned cascade maps onto a linear qubit chain at constant-factor two-qubit overhead through the nearest-neighbour synthesis of Ref.~\cite{Bergholm2005} (Appendix~\ref{app:cnot-structure}), but the achievable constant on a given sector is set by the qubit ordering that minimises the summed control depth of the surviving Gray-code parities, not by the inactive-node count the compiler currently optimises; a parity-aware nearest-neighbour synthesis matched to that ordering would turn the constant-factor claim into a provable bound, and would separate the structured supports, where a constant overhead is expected, from arbitrary spread-out supports, where it need not hold.

Beyond the compiler, the most promising directions develop the composed circuit, where the construction's quantum value resides, broaden its reach, and demonstrate it on hardware. The dequantization boundary of Appendix~\ref{app:simulability} makes the composed circuit's hardness a precise design target: Corollary~\ref{cor:advantage} requires the dressing $U(\bm\phi)$ to escape stabiliser, Gaussian, light-cone and tensor-network contraction at once, and a natural next step is to characterise the minimal such dressing and to benchmark it against the Heisenberg-picture Pauli-propagation and sparse-Pauli-dynamics simulators~\cite{Begusic2024, Angrisani2025} that address the same deep-scrambling dynamics and are not subsumed by those four classes; identifying a regime that provably evades them would establish the dynamics-based advantage of \sref{sec:discussion} on an explicit instance. The construction also reaches beyond electronic structure: since the active-leaves picture preserves any diagonally-defined sector with no penalty terms, it yields a feasibility-preserving, exactly-diagonal-metric natural-gradient ansatz for constrained combinatorial optimisation~\cite{Farhi2014, Hadfield2019} and for fixed-Hamming-weight code states~\cite{Gottesman1997}. Finally, the bare circuits here are shallow (tens to a few hundred CNOTs) and thus accessible to present-day hardware, so a real-device implementation is a natural next step: it would run on hardware the two mechanisms specific to this ansatz, the metric-free natural-gradient step, whose preconditioner is read directly from measured subtree and leaf probabilities rather than from auxiliary Hadamard tests, and the exact Fubini--Study sampler, whose two-design advantage over basis-state averaging is directly measurable at small sector dimension, and it extends naturally to the dressed circuits themselves.

\section*{Data availability}
The data underlying the figures and tables in this paper were produced with the \textsc{QuantumSymmetry} package (see Code availability); the figure-generating scripts and intermediate result files are available from the corresponding author on reasonable request.

\section*{Code availability}
The methods introduced here, the binary-tree ansatz and pruning compiler, the closed-form diagonal-metric routines, the parameter-shift gradients and singular-coordinate initialiser, the real- and imaginary-time integrators, and the sector-Haar sampler, are implemented in the open-source Python package \textsc{QuantumSymmetry}~\cite{picozzi2023quantumsymmetry, Picozzi2022} (\url{https://github.com/dariopicozzi/quantumsymmetry}, installable via \texttt{pip install quantumsymmetry}), which builds on Qiskit, OpenFermion and PySCF. Tutorials demonstrating how to implement the ideas of this paper are available in \textsc{QuantumSymmetry}. Scripts that generate the specific figures of this paper are available from the corresponding author on reasonable request.

\bibliographystyle{quantum-apsrev}
\bibliography{references}

\section*{Acknowledgements}
D.P.\ acknowledges support from the Engineering and Physical Sciences Research Council (EPSRC) under grant numbers EP/T517793/1, EP/S021582/1, and EP/W524335/1.

\paragraph*{Author contributions.} D.P.\ conceived the project, performed all analytical derivations and numerical experiments, and wrote the manuscript.

\paragraph*{Competing interests.} The author declares no competing interests.

\newcommand{\COMBINEDBUILD}{}
\ifdefined\COMBINEDBUILD
\providecommand{\SupplementaryStandaloneEnd}{}
\else
\documentclass[aps,pra,twocolumn,superscriptaddress]{revtex4-1}
\usepackage{amsmath,amssymb,graphicx,bm}
\usepackage{array}
\usepackage{adjustbox}
\usepackage[table]{xcolor}
\usepackage{subcaption}
\usepackage{booktabs}
\usepackage{multirow}
\usepackage{float}
\usepackage{sansmath}
\usepackage{hyperref}
\usepackage{xr-hyper}
\usepackage{braket}
\usepackage{tikz}
\usetikzlibrary{trees,positioning,calc,arrows.meta,shapes.geometric}
\usepackage{tikz-cd}
\usepackage[qm]{qcircuit}

\DeclareMathOperator{\atantwo}{atan2}

\newcommand{\Active}{\textsf{A}}
\newcommand{\Inactive}{\textsf{I}}

\newtheorem{theorem}{Theorem}
\newtheorem{conjecture}{Conjecture}
\newtheorem{proposition}{Proposition}
\newtheorem{lemma}{Lemma}
\newtheorem{definition}{Definition}
\newtheorem{corollary}{Corollary}
\newcounter{algorithm}

\newcommand{\sref}[1]{Sec.~\ref{#1}}

\providecommand{\Achip}[1]{\fcolorbox{yellow!70!black}{yellow!55}{\rule{0pt}{1.0ex}\,$\scriptstyle #1$\,}}
\providecommand{\Fchip}[1]{\fcolorbox{red!65!black}{red!40}{\rule{0pt}{1.0ex}\,$\scriptstyle #1$\,}}
\providecommand{\Ichip}[1]{\fcolorbox{gray!60!black}{gray!30}{\rule{0pt}{1.0ex}\,$\scriptstyle #1$\,}}
\providecommand{\TargetColFill}{orange!20}
\providecommand{\PivotRow}{cyan!22}
\providecommand{\AffectedRow}{yellow!22}
\providecommand{\SurvivorRow}{green!25}
\providecommand{\TargetPivot}{orange!42}

\externaldocument{main}

\begin{document}

\title{Supplementary Information for ``A hardware-efficient variational ansatz with an exact diagonal metric for real- and imaginary-time evolution and Haar sampling''}

\author{Dario Picozzi}
\email{picozzi.dario@gmail.com}
\affiliation{Department of Physics and Astronomy, University College London (UCL), Gower Street, London, WC1E 6BT, United Kingdom}
\affiliation{London Centre for Nanotechnology, 19 Gordon St, London, WC1H 0AH, United Kingdom}

\maketitle
\newcommand{\SupplementaryStandaloneEnd}{%
  \bibliographystyle{quantum-apsrev}
\bibliography{references}%
  \end{document}%
}
\fi

\ifdefined\COMBINEDBUILD
  \clearpage
\fi

\definecolor{headcolor}{RGB}{224,235,245}
\newcommand{\molicon}[1]{\raisebox{-0.33\height}{\includegraphics[height=4mm]{figures/molecules/#1.png}}}

\appendix

\section{Geometry primer: state manifolds, tangent spaces, and real versus imaginary time}\label{app:geometry-primer}

This appendix fixes the geometric language used in the main text. The point of the discussion is simple: the energy derivative $dE$ is first a slope, not a motion; the metric $g$ turns that slope into a gradient vector; and the complex structure $J$ rotates the gradient into the Hamiltonian wavefunction motion. With the convention $g=\mathrm{Re}\braket{\cdot|\cdot}$ used throughout the paper, physical real time also carries the factor $1/(2\hbar)$. The whole construction is the commuting triangle
\begin{equation}
\begin{tikzcd}[row sep=large, column sep=large]
dE \in T^{*} \arrow[r, "g^{-1}"] & \nabla_g E \in T \arrow[d, "{-J/(2\hbar)}"] \\
& X_E \in T \arrow[ul, "{2\hbar\,\Omega}"]
\end{tikzcd}
\label{eq:primer-chain}
\end{equation}
where $T^{*}\equiv T_{[\psi]}^{*}\mathbb{CP}^{N-1}$ and $T\equiv T_{[\psi]}\mathbb{CP}^{N-1}$ are the cotangent and tangent spaces at $[\psi]$, $g$ is the matrix of the metric components in the chosen tangent-coordinate basis, and $\Omega=gJ$ is the K\"ahler (symplectic) form. The metric inverse $g^{-1}$ raises the energy covector to the gradient vector, the complex structure $J$ rotates it into the real-time velocity, and $\Omega$ closes the loop, returning $X_E$ to $dE$ via $(2\hbar\,\Omega)X_E=dE$; the triangle commutes because the composite map around it is the identity, $(2\hbar\,\Omega)\bigl(-J/(2\hbar)\bigr)g^{-1}=-gJ^2g^{-1}=I$. The numerical algorithms set $\hbar=1$; restoring units divides every real-time velocity by $\hbar$. The same statement is the geometric reason that imaginary-time and real-time variational equations in the paper use the same metric inverse.

\paragraph*{Rays and tangent spaces.}
Physical pure states are rays, not Hilbert-space vectors with a fixed global phase:
\begin{equation}
[\psi] \;=\; \{e^{i\alpha}\ket\psi:\alpha\in\mathbb R\}.
\end{equation}
For a Hilbert space of dimension $N$, the pure-state manifold is the projective space $\mathbb{CP}^{N-1}$. A tangent vector at $[\psi]$ is an infinitesimal physical change of the ray. If one works with normalized Hilbert-space representatives, the component proportional to $i\ket\psi$ is only an infinitesimal global phase, so it is a gauge direction rather than a physical tangent direction. A standard horizontal representative removes this phase component:
\begin{equation}
\ket{\delta\psi_\perp}
\;=\;
(I-\ket\psi\!\bra\psi)\ket{\delta\psi},
\qquad
\braket{\psi|\delta\psi_\perp}=0,
\end{equation}
up to the convention used to fix the global phase. For one qubit this projective manifold is $\mathbb{CP}^1$, which is represented by the Bloch sphere in Fig.~\ref{fig:bloch-geometry-primer}.

\paragraph*{The energy slope is a covector.}
The energy expectation
\begin{equation}
E([\psi])=\braket{\psi|H|\psi}
\end{equation}
defines a scalar function on projective state space. Its differential $dE$ is a covector: it eats a tangent vector $\xi$ and returns the first-order energy change in that direction, $dE[\xi]$. By itself $dE$ is not a direction of motion on the manifold. This distinction matters because a slope and a vector live in dual spaces: $dE\in T_{[\psi]}^*\mathbb{CP}^{N-1}$, whereas wavefunction velocities live in $T_{[\psi]}\mathbb{CP}^{N-1}$.

\paragraph*{The metric turns slopes into gradients.}
The Riemannian metric $g$ identifies tangent vectors with covectors. In a local tangent basis $\{e_\mu\}$, write
\begin{equation}
g_{\mu\nu}=g(e_\mu,e_\nu),
\qquad
(dE)_\mu=dE[e_\mu].
\end{equation}
The gradient vector $\nabla_gE$ is defined by
\begin{equation}
g(\nabla_gE,\xi)=dE[\xi]
\quad\text{for every tangent vector }\xi.
\end{equation}
In coordinates this is the linear system
\begin{equation}
g\,\nabla_gE=dE,
\qquad
\nabla_gE=g^{-1}dE.
\label{eq:primer-gradient}
\end{equation}
Imaginary-time evolution is the corresponding downhill flow, $-\nabla_gE$: it uses the metric to convert the energy slope into the steepest descent direction. For physical imaginary time generated by $H$, the right-hand side is multiplied by the same overall factor $1/(2\hbar)$ that appears below; in optimization language this constant is normally absorbed into the step size.

\paragraph*{The complex structure rotates gradients into wavefunction dynamics.}
Quantum state space also has a complex structure $J$, inherited from multiplication by $i$ in Hilbert space after quotienting the unphysical phase. It satisfies $J^2=-I$ and acts as a $90^\circ$ rotation in each complex tangent plane. With the sign and normalization convention used here, the Hamiltonian vector field generated by $E$ is
\begin{equation}
X_E
\;=\;
-\,\frac{1}{2\hbar}J\nabla_gE
\;=\;
-\,\frac{1}{2\hbar}J\,g^{-1}dE.
\label{eq:primer-hamiltonian-vector}
\end{equation}
Thus the metric $g$ and the complex structure $J$ do two different jobs: $g$ converts a slope into a gradient vector, while $J$ rotates that gradient vector into the real-time wavefunction motion. If the local metric is the identity, $g=I$, the formula reduces to $X_E=-(2\hbar)^{-1}JdE$ after identifying the covector with a vector by the Euclidean metric.

\paragraph*{The K\"ahler form records the metric-plus-rotation relation.}
The compatible symplectic, or K\"ahler, form already introduced in the triangle~\eqref{eq:primer-chain} is the antisymmetric bilinear form obtained from the metric and complex structure, $\Omega(\xi,\eta)=g(\xi,J\eta)$, with matrix $\Omega=gJ$ in local coordinates. The Hamiltonian equation can therefore be written equivalently as $\bigl(2\hbar\,\Omega\bigr)X_E=dE$, which is the same content as Eq.~\eqref{eq:primer-hamiltonian-vector} and the closing arrow of the triangle: the metric and complex structure together turn the energy slope into the real-time tangent vector.

Equivalently, $g$ and $\Omega$ are the symmetric and antisymmetric parts of the quantum geometric tensor $Q_{\mu\nu}=\braket{\partial_\mu\psi|\partial_\nu\psi}-\braket{\partial_\mu\psi|\psi}\braket{\psi|\partial_\nu\psi}$~\cite{Kolodrubetz2017}: its real part $\operatorname{Re}Q_{\mu\nu}=g_{\mu\nu}$ is the Fubini--Study metric, while its imaginary part is the Berry curvature $F_{\mu\nu}=-2\operatorname{Im}Q_{\mu\nu}$~\cite{Berry1984, Shapere1989, Xiao2010}, so that $\Omega=\tfrac12 F$. The identification $\Omega=gJ\propto F$ is a property of the K\"ahler structure of projective Hilbert space and holds because the chart is complete and its pruned restriction maps onto a coordinate sub-$\mathbb{CP}$, which is itself a K\"ahler submanifold; for a generic variational ansatz, whose manifold need not be K\"ahler, $gJ$ and the pulled-back Berry curvature differ. The metric part of $Q$ is one quarter of the quantum Fisher information, $g=\tfrac14\,\mathcal F_Q$. All of these identifications hold on the projective state manifold, i.e.\ up to the global-phase gauge fixed in Appendix~\ref{app:ng-omega-equiv}.

\begin{figure*}[t]
\centering
\resizebox{0.95\textwidth}{!}{\input{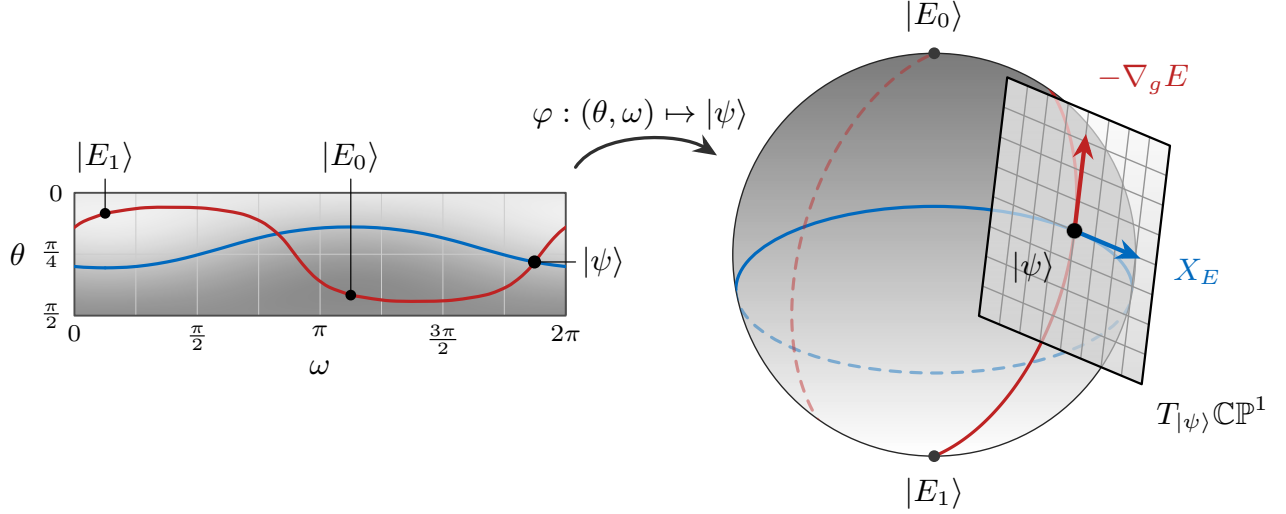}}
\caption{One-qubit picture of the projective state manifold. \textbf{Right:} Bloch-sphere view, with darker shading indicating lower energy. At the state $\ket{\psi}$, the metric turns the energy differential $dE$ into the gradient vector $\nabla_gE$ in the tangent plane $T_{\ket{\psi}}\mathbb{CP}^1$. Imaginary-time evolution follows the downhill direction $-\nabla_gE$ (red), while real-time evolution follows the Hamiltonian vector field $X_E=-(2\hbar)^{-1}J\nabla_gE$ (blue). The complex structure $J$ rotates tangent vectors by $90^\circ$, so the two directions are perpendicular in the Fubini--Study metric. \textbf{Left:} pullback to the ansatz coordinates represented by $\varphi(\theta,\omega)=\cos\theta\,\ket{0}+e^{i\omega}\sin\theta\,\ket{1}$, with $\theta\in[0,\pi/2]$ and periodic $\omega\in[0,2\pi)$. Geometric angles in the ansatz manifold are measured by the pullback metric rather than by the Euclidean metric suggested by the ansatz coordinates.}
\label{fig:bloch-geometry-primer}
\end{figure*}

\paragraph*{Pullback to variational parameters.}
The algorithms in the paper do not move freely on all of $\mathbb{CP}^{N-1}$; they move on the variational submanifold traced by $\bm\lambda\mapsto[\psi(\bm\lambda)]$. Pulling the ambient metric back to parameter space gives
\begin{equation}
g_{\mu\nu}
\;=\;
\mathrm{Re}\bigl[\braket{\partial_\mu\psi|\partial_\nu\psi}
-\braket{\partial_\mu\psi|\psi}\braket{\psi|\partial_\nu\psi}\bigr],
\end{equation}
where $g$ is the projective metric in independent coordinates; Appendix~\ref{app:ng-omega-equiv} relates it to the diagonal chart metric. The two flows are:
\begin{align}
\dot{\lambda}^{\mu}_{\mathrm{imag}}
&\;=\;
-\sum_\nu (g^{-1})^{\mu\nu}\,\partial_\nu E,\\
\dot{\lambda}^{\mu}_{\mathrm{real}}
&\;=\;
-\,\frac{1}{2\hbar}\sum_{\nu,\rho} J^{\mu}{}_{\nu}\,(g^{-1})^{\nu\rho}\,\partial_\rho E,
\end{align}
where the velocity $\dot\lambda^\mu$ and the inverse metric $(g^{-1})^{\mu\nu}$ carry upper (contravariant) indices, the energy covector $\partial_\nu E$ carries a lower index, and $J^{\mu}{}_{\nu}$ is the complex structure as a $(1,1)$ tensor; the labels $\mathrm{imag}$ and $\mathrm{real}$ are written as subscripts to avoid colliding with the contravariant component index.
For a generic ansatz this metric is dense and must be estimated and inverted. For the binary-tree chart, the metric raise is implemented by diagonal division in redundant leaf phases, followed by re-gauging to the chosen coordinates (Appendix~\ref{app:ng-omega-equiv}). The real- and imaginary-time flows then differ by the complex-structure rotation and the stated physical time-unit factor.

\section{Diagonal chart metric: statement and proof}\label{app:diag-metric}

The central structural result of the paper is that the binary-tree chart makes the pullback metric exactly diagonal in closed form, summarised in \sref{sec:metric}. It is restated formally and proved below.

\begin{theorem}[Diagonal chart metric]\label{thm:diag-metric}
For the binary-tree ansatz of \sref{sec:circuit}, the chart pullback metric $g_{\mu\nu}=\mathrm{Re}\,\braket{\partial_\mu\psi|\partial_\nu\psi}$ is block-diagonal between the amplitude parameters $\bm\theta$ and the leaf phases $\bm\omega$, and each block is itself diagonal, with entries and weights given by Eqs.~\eqref{eq:g-theta}--\eqref{eq:g-cross}.
\end{theorem}

The two weight formulas can be read off the tree by walking down from the root and accumulating one factor at each step: $\cos^2\theta_m$ on a downward branch, $\sin^2\theta_m$ on an upward branch. For the tree of Fig.~\ref{fig:classified-tree}, the parameter $\theta_7$ sits at depth three on the leftmost root-to-leaf path with ancestors $\theta_0,\theta_1,\theta_3$ all reached by downward steps, so $w_7 = \cos^2\theta_0\cos^2\theta_1\cos^2\theta_3$; the leaf $\ket{0000}$ below it adds one further downward step, giving $p_0 = \cos^2\theta_0\cos^2\theta_1\cos^2\theta_3\cos^2\theta_7$.

\noindent\textbf{Proof.} The diagonal structure is easy to see at the level of state derivatives. The amplitude derivative $\partial_{\theta_i}\ket\psi$ lives entirely on the subtree below node $i$ and carries a prefactor $\sqrt{w_i}$ from the sines and cosines above it. For two distinct amplitude derivatives $\partial_{\theta_i}\ket\psi$ and $\partial_{\theta_j}\ket\psi$ there are two cases. If neither node is an ancestor of the other, their lowest common ancestor routes them into orthogonal child subtrees, so they have disjoint support and the inner product vanishes outright. If instead one node, say $i$, is an ancestor of $j$, then on their only shared support, the subtree below $j$, the higher derivative $\partial_{\theta_i}\ket\psi$ is proportional to the undifferentiated branch state $\ket{\varphi_j}=\cos\theta_j\ket{L}+\sin\theta_j\ket{R}$ at node $j$ (with $\ket{L},\ket{R}$ the orthonormal child-subtree states), while $\partial_{\theta_j}\ket\psi$ is proportional to $\partial_{\theta_j}\ket{\varphi_j}=-\sin\theta_j\ket{L}+\cos\theta_j\ket{R}$; these are orthogonal, $\braket{\varphi_j|\partial_{\theta_j}\varphi_j}=0$, so the inner product again vanishes. The leaf-phase derivative $\partial_{\omega_j}\ket\psi = i c_j\ket{j}$ is supported on a single computational-basis leaf, so distinct leaf-phase derivatives are exactly orthogonal in Hilbert space. The cross-block vanishes leaf by leaf: writing $c_j=r_j\,e^{i\omega_j}$ with $r_j$ real, the only nonzero overlap at leaf $j$ is $(\partial_{\theta_i}c_j)^{*}\,(\partial_{\omega_j}c_j)=(\partial_{\theta_i}r_j)\,e^{-i\omega_j}\cdot i\,r_j\,e^{i\omega_j}=i\,r_j\,\partial_{\theta_i}r_j$, which is pure imaginary and so contributes no real part.\hfill$\square$

\section{Metric invariance under fixed dressing}\label{app:invariance}

The diagonal metric of \sref{sec:metric} is a property of the bare tree map $(\bm\theta,\bm\omega)\mapsto\ket{\psi(\bm\theta,\bm\omega)}$, and is preserved under any fixed post-composition. This both grounds the bare/dressed split used throughout the applications and broadens the ansatz class.

\begin{proposition}[Invariance under fixed dressing]\label{prop:invariance}
Let $\ket{\psi(\bm\xi)}$ be any parametrized state and let $U$ be any unitary independent of the parameters $\bm\xi$. Then the derivative Gram tensor $\braket{\partial_\mu\psi|\partial_\nu\psi}$, and hence both the chart metric $g_{\mu\nu}=\mathrm{Re}\,\braket{\partial_\mu\psi|\partial_\nu\psi}$ and the imaginary-part two-form $\mathrm{Im}\,\braket{\partial_\mu\psi|\partial_\nu\psi}$, is unchanged under post-composition by $U$: the metric of $U\ket{\psi}$ equals that of $\ket{\psi}$.
\end{proposition}

\noindent\textbf{Proof.} Since $U$ does not depend on $\bm\xi$, $\partial_\mu\bigl(U\ket{\psi}\bigr)=U\ket{\partial_\mu\psi}$, so $\braket{\partial_\mu(U\psi)|\partial_\nu(U\psi)}=\braket{\partial_\mu\psi|U^\dagger U|\partial_\nu\psi}=\braket{\partial_\mu\psi|\partial_\nu\psi}$ by unitarity; taking real and imaginary parts gives the metric and the imaginary-part two-form.\hfill$\square$

Three consequences fix the role of this result in the main text.

\textbf{(i) An enlarged ansatz class.} Because the diagonal metric of Eqs.~\eqref{eq:g-theta}--\eqref{eq:g-cross} is a property of the bare tree map and not of any particular circuit realizing it, Proposition~\ref{prop:invariance} carries it unchanged to every member of the family $\{U\ket{\psi(\bm\theta,\bm\omega)}:U\ \text{fixed}\}$. The closed-form toolkit, natural gradient, variational imaginary- and real-time evolution (\sref{sec:kahler}), and exact sector-Haar sampling (\sref{sec:fs-apps}), therefore applies not only to the tree but to any fixed dressing of it.

\textbf{(ii) Fixed dressings reduce to the core.} For a Hamiltonian $H$ the dressed energy is $\braket{U\psi|H|U\psi}=\braket{\psi|\widetilde H|\psi}$ with $\widetilde H=U^\dagger H U$; by the Proposition the geometry seen by the chart parameters is unchanged, so a natural-gradient or imaginary-time step on the dressed state for $H$ is exactly such a step on the bare core for $\widetilde H$, using the diagonal preconditioner of Eq.~\eqref{eq:Minv}. The metric thus stays free, while the conjugation $\widetilde H=U^\dagger H U$ carries the quantum content, precisely the bare/dressed boundary of \sref{sec:decouple}.

\textbf{(iii) Parametrised dressings.} When the dressing carries its own parameters, $U=U(\bm\phi)$ as in Eq.~\eqref{eq:dressed-main}, the Proposition still applies at each fixed $\bm\phi$: differentiating $U(\bm\phi)\ket{\psi(\bm\theta,\bm\omega)}$ in the core parameters alone, holding $\bm\phi$ fixed, reproduces the core's diagonal metric exactly, for every $\bm\phi$. Embedding the bare core under a parametrised dressing thus leaves its closed-form geometry untouched at every point of an optimisation, not only at a fixed dressing. The two applications use this differently. In the decoupling of \sref{sec:decouple} the dressing is screened and trained on its own, and the metric enters only through the bare-core Fubini--Study sampling over the model space $P$, which the dressing does not touch. In the two-block calculation of \sref{sec:two-block} the core and the dressing are refined together: at each step the core parameters take a natural-gradient step preconditioned by the closed-form diagonal metric (which, by the Proposition, is exactly the core block of the dressed metric at the current $\bm\phi$), while the dressing parameters take an ordinary gradient step. The core-dressing block of the full metric is never formed. Dropping that block does not rest on its being small, and in general it is not: the objective does not depend on the metric, so its minima are unchanged; the diagonal preconditioner is positive-definite, so the joint step is a descent direction that the line search keeps monotone; and the core step is, by the Proposition, the exact natural gradient of the core at the current dressing. The genuine boundary of the result lies elsewhere: trainable layers inserted within the tree change the core map itself, and with it the closed-form geometry.

Two fixed-dressing instances are developed in this Supplement: Appendix~\ref{app:sparse-routing}, in which $U$ is a routing permutation, and Appendix~\ref{app:schur-csf}, in which $U$ is the sector Schur transform that carries the leaf basis from Slater determinants to configuration state functions.

\section{Two metrics, one update}\label{app:ng-omega-equiv}

This appendix gives the self-contained derivation behind the relation-to-Fubini--Study discussion of \sref{sec:metric}. The chart pullback metric of Theorem~\ref{thm:diag-metric} and the projective Fubini--Study form of Eq.~\eqref{eq:g-FS} differ by the rank-one term $-pp^\top$ on the leaf-phase block, and this difference is shown to be unobservable for the physical updates used in the paper. For energy minimization and variational imaginary-time evolution the energy gradient is automatically sum-zero in the leaf phases. For real-time evolution the projective K\"ahler flow is defined only up to a uniform phase velocity, so the two choices of leaf-phase metric differ only by the gauge that is projected out in Appendix~\ref{app:tdvp-rhs}.

\paragraph*{Setting.} On the leaf-phase block the two metrics are
\begin{equation}\label{eq:two-metrics}
\widetilde g_\omega \;=\; D \;:=\; \mathrm{diag}(p), \qquad g^{\mathrm{FS}}_\omega \;=\; D - p p^\top,
\end{equation}
with $p\in\mathbb R_{>0}^{|S|}$ and $\sum_j p_j = 1$ (the prepared state is normalized). The first is the raw real part of $\braket{\partial_{\omega_j}\psi|\partial_{\omega_k}\psi}$ in the redundant one-phase-per-leaf chart. The second is the Provost--Vall\'ee form $g^{\mathrm{FS}}_{\mu\nu} = \mathrm{Re}\bigl[\braket{\partial_\mu\psi|\partial_\nu\psi} - \braket{\partial_\mu\psi|\psi}\braket{\psi|\partial_\nu\psi}\bigr]$~\cite{Provost1980}, which subtracts the (single) global-phase direction and is the canonical Riemannian metric on $\mathbb{CP}^{N-1}$. In components, the leaf-phase block of the Provost--Vall\'ee form is
\begin{equation}\label{eq:g-FS}
g^{\mathrm{FS}}_{\omega_j\omega_k} \;=\; p_j\,\delta_{jk} - p_j p_k,
\end{equation}
the rank-one phase subtraction that removes the global-phase direction $\bm 1$.

\paragraph*{Kernel and pseudoinverse of $g^{\mathrm{FS}}_\omega$.} A direct computation gives $g^{\mathrm{FS}}_\omega\bm 1 = p - p\,(\bm 1^\top p) = 0$, so $\ker g^{\mathrm{FS}}_\omega = \mathrm{span}(\bm 1)$ and $g^{\mathrm{FS}}_\omega$ has rank $|S|-1$. Let $P_\perp := I - \tfrac{1}{|S|}\bm 1\bm 1^\top$ denote the orthogonal projector onto $\bm 1^\perp$; the Moore--Penrose pseudoinverse $g^{\mathrm{FS},+}_\omega$ is uniquely determined by $g^{\mathrm{FS}}_\omega g^{\mathrm{FS},+}_\omega = g^{\mathrm{FS},+}_\omega g^{\mathrm{FS}}_\omega = P_\perp$, together with $g^{\mathrm{FS},+}_\omega\bm 1=0$.

\paragraph*{Imaginary time: gauge invariance kills the correction.} For the energy objective $E(\bm\theta,\bm\omega)=\braket{\psi|H|\psi}$, the chain rule gives $\partial_{\omega_j} E = 2\,\mathrm{Re}\braket{\partial_{\omega_j}\psi|H|\psi}$. Summing over $j$ and using $\sum_j\partial_{\omega_j}\ket\psi = i\ket\psi$,
\begin{equation}\label{eq:omega-sum-zero}
\bm 1^\top\nabla_\omega E \;=\; 2\,\mathrm{Re}\braket{i\psi|H|\psi} \;=\; -2\,\mathrm{Im}\braket{\psi|H|\psi} \;=\; 0,
\end{equation}
since $\braket{\psi|H|\psi}\in\mathbb R$. Equivalently, $E$ is invariant under $\omega_j\mapsto\omega_j+c$ for any $c\in\mathbb R$, so its gradient is automatically sum-zero. For any sum-zero $u$ (i.e.\ $\bm 1^\top u = 0$) we have
\begin{equation}\label{eq:agreement}
g^{\mathrm{FS},+}_\omega\,u \;=\; P_\perp\bigl(D^{-1} u\bigr) \;=\; \widetilde g_\omega^{-1}\,u \;-\; \tfrac{1}{|S|}\bm 1\bigl(\bm 1^\top D^{-1} u\bigr).
\end{equation}
\textbf{Proof.} Applying $g^{\mathrm{FS}}_\omega$ to $P_\perp D^{-1} u$, and using $g^{\mathrm{FS}}_\omega P_\perp = g^{\mathrm{FS}}_\omega$ together with $p^\top D^{-1} = \bm 1^\top$,
\begin{equation}
g^{\mathrm{FS}}_\omega\bigl(P_\perp D^{-1} u\bigr) = (D - p p^\top) D^{-1} u = u - p\,(\bm 1^\top u) = u.
\end{equation}
The pseudoinverse satisfies $g^{\mathrm{FS}}_\omega\,g^{\mathrm{FS},+}_\omega u = P_\perp u = u$. Both candidates therefore solve $g^{\mathrm{FS}}_\omega x = u$ on $\bm 1^\perp$, on which $g^{\mathrm{FS}}_\omega$ is invertible, so they coincide. $\square$

The Provost--Vall\'ee preconditioned natural gradient $g^{\mathrm{FS},+}_\omega\nabla_\omega E$ and the diagonal preconditioned gradient $\widetilde g_\omega^{-1}\nabla_\omega E = D^{-1}\nabla_\omega E$ therefore differ only by a uniform shift of all $\omega_j$ along $\bm 1$. The prepared ray is invariant under that shift (it is exactly the global-phase ambiguity of the $\omega$-parametrization), so the two updates produce the same physical state after each step. This is the formal justification for using $\widetilde g_\omega^{-1}$ throughout imaginary-time evolution and quantum natural gradient.

\paragraph*{Real time: the same gauge equivalence.} The projective real-time flow is
\[
\dot{\bm\lambda}=-(2\hbar)^{-1}Jg^{-1}dE.
\]
Its phase-gradient input is again $u=\nabla_\omega E$, and Eq.~\eqref{eq:omega-sum-zero} still gives $\bm 1^\top u=0$. Replacing $D^{-1}u$ by $g^{\mathrm{FS},+}_\omega u$ therefore changes only the uniform component of the metric-raised phase gradient. Both metric-raised vectors represent the same projective tangent vector, so the projective complex structure acts identically on their equivalence classes. In Hilbert space, this means taking the horizontal representative before applying multiplication by $i$. Appendix~\ref{app:tdvp-rhs} implements this directly: after applying the diagonal reciprocal, the tree recursion applies $J$ and then subtracts the mean of $\dot{\bm\omega}$ as a gauge choice.

\paragraph*{Worked example.} Take $|S|=2$, $H=Z$, and $\ket\psi=\sqrt{p_0}\ket0+e^{i(\omega_1-\omega_0)}\sqrt{p_1}\ket1$. The energy is $E=p_0-p_1$. The real-time K\"ahler rule gives the relative phase velocity $\dot\omega_1-\dot\omega_0=2/\hbar$, matching $\exp(-iZt/\hbar)\ket\psi$. In the probability-weighted zero-mean (horizontal) gauge, $p^\top\dot{\bm\omega}=0$, this may be represented as $\dot{\bm\omega}=(-2p_1/\hbar,\,2p_0/\hbar)$; adding the uniform velocity $-E/\hbar$ gives the equally valid Schr\"odinger representative $(-1/\hbar,+1/\hbar)$. The difference is only the global phase.

\section{Coordinate transformations: ranges and inverse map}\label{app:c2p-detail}

This appendix collects the standard details of the polyspherical chart map. For each internal node bifurcating into subtrees with leaf-index sets $S_-$ (downward) and $S_+$ (upward), the amplitude angle reads off the squared subtree masses by
\begin{equation}\label{eq:c2p-ry}
\theta_i \;=\; \atantwo\!\left( \sqrt{\textstyle\sum_{k\in S_+}|c_k|^2},\; \sqrt{\textstyle\sum_{k\in S_-}|c_k|^2}\right),
\end{equation}
and the leaf phase $\omega_j$ is the argument of the corresponding amplitude $c_j$. The leaf-phase angles read off the complex amplitudes,
\begin{equation}\label{eq:c2p-rz}
\omega_j = \atantwo\bigl(\Im c_j,\; \Re c_j\bigr),
\end{equation}
with one phase fixed (e.g.\ $\omega_0=0$) to remove the global-phase redundancy. The natural ranges for $\theta_i$ are $[0,\pi/2]$ when both subtrees contain at least one leaf, and one of $[-\pi/2,\pi/2]$ or $[0,\pi]$ when one subtree is empty (the choice fixes which child carries the surviving amplitude). The elementary $R_y$ and $R_z$ angles emitted by the cascade are obtained from the tree parameters $(\bm\theta,\bm\omega)$ by the inverse Walsh--Hadamard transform on each level, re-indexed by the Gray code~\cite{Mottonen2004, Mottonen2005}; together with Eq.~\eqref{eq:c2p-ry}, this defines the chart-to-circuit map used throughout the paper.

\section{Closed-form real-time RHS from the tree recursion}\label{app:tdvp-rhs}

The real-time right-hand side is derived directly from the K\"ahler geometry of projective Hilbert space. With the convention of the main text,
\[
g(\xi,\eta)=\mathrm{Re}\braket{\xi|\eta},
\]
Schr\"odinger evolution $i\hbar\,\partial_t\ket\psi=H\ket\psi$ induces, on the projective state manifold,
\begin{equation}\label{eq:kahler-flow-units}
\dot{\bm\lambda}
\;=\;
-\,\frac{1}{2\hbar}\,J\,g^{-1}dE
\end{equation}
up to the unobservable global-phase direction. The factor $1/2$ comes from $\partial_\mu E=2\,\mathrm{Re}\braket{\partial_\mu\psi|H|\psi}$. The formulas below keep $\hbar$ explicit; the numerical work sets $\hbar=1$.

The binary tree supplies an efficient coordinate expression for the complex-structure rotation $J$. One phase coordinate per active leaf is temporarily kept; a uniform shift of all phases is a gauge direction and is projected out at the end. If the implemented chart fixes one reference phase, its derivative is recovered from the gauge identity $\sum_j\partial_{\omega_j}E=0$, so this temporary restoration introduces no extra circuit evaluations. Let $a$ be an active internal node, with left and right child subtrees denoted $L(a)$ and $R(a)$, and write
\[
s_a:=\sin\theta_a,\qquad c_a:=\cos\theta_a .
\]
For any subtree $B$, define the collective phase coordinate $\phi_B$ to shift every leaf phase in $B$ by the same amount. Its derivative and diagonal metric entry are just leaf sums,
\begin{align}
D^\omega_B &\;:=\; \partial_{\phi_B}E
        \;=\; \sum_{j\in B}\partial_{\omega_j}E,
        \label{eq:subtree-phase-derivative}\\
G^\omega_B &\;:=\; g_{\phi_B\phi_B}
        \;=\; \sum_{j\in B}g_{\omega_j\omega_j}.
        \label{eq:subtree-phase-metric}
\end{align}
The metric-raised subtree phase derivative is
\begin{equation}\label{eq:Gamma-def}
\Gamma_B \;:=\; (G^\omega_B)^{-1}D^\omega_B .
\end{equation}
Geometrically, $\Gamma_B$ is the phase component of $g^{-1}dE$ averaged over subtree $B$.

\paragraph*{Bottom-up pass for the angle velocities.}
At a leaf $j$, initialise $D^\omega_j=\partial_{\omega_j}E$ and $G^\omega_j=g_{\omega_j\omega_j}$. A single postorder walk over the active tree computes, for every internal subtree $B$ with children $L$ and $R$,
\begin{equation}\label{eq:Gamma-recursion}
D^\omega_B=D^\omega_L+D^\omega_R,\qquad
G^\omega_B=G^\omega_L+G^\omega_R,\qquad
\Gamma_B=\frac{D^\omega_B}{G^\omega_B}.
\end{equation}
The real-time $-J/(2\hbar)$ rotation sends the phase-gradient contrast between the children of $a$ into the velocity of the split angle:
\begin{equation}\label{eq:thetadot-tree}
\boxed{\displaystyle
\dot\theta_a
\;=\;
\frac{\sin(2\theta_a)}{4\hbar}
\left(\Gamma_{R(a)}-\Gamma_{L(a)}\right).}
\end{equation}
Expanded out, this is
\begin{equation}
\dot\theta_a
\;=\;
\frac{\sin(2\theta_a)}{4\hbar}
\left[
\frac{\sum_{j\in R(a)}\partial_{\omega_j}E}
     {\sum_{j\in R(a)}g_{\omega_j\omega_j}}
-
\frac{\sum_{j\in L(a)}\partial_{\omega_j}E}
     {\sum_{j\in L(a)}g_{\omega_j\omega_j}}
\right],
\end{equation}
but Eq.~\eqref{eq:Gamma-recursion} is the linear-time way to evaluate all such subtree averages at once.

\paragraph*{Top-down pass for the phase velocities.}
The complementary part of the real-time $-J/(2\hbar)$ rotation sends the metric-raised amplitude derivative at node $a$ into a phase-velocity contrast between its two child subtrees. Define the metric-raised split derivative
\begin{equation}\label{eq:Lambda-def}
\Lambda_a
\;:=\;
\frac{g_{\theta_a\theta_a}^{-1}\,\partial_{\theta_a}E}{\sin(2\theta_a)}
\;=\;
\frac{\partial_{\theta_a}E}
     {g_{\theta_a\theta_a}\,\sin(2\theta_a)} .
\end{equation}
Let $\nu_B$ denote the phase velocity assigned uniformly to the subtree $B$ before the descendants of $B$ add their own contrasts. Choose an arbitrary global-phase gauge $\nu_{\mathrm{root}}=\chi(t)$; the projective state is independent of this choice. A preorder walk over the active tree propagates
\begin{align}
\nu_{L(a)} &\;=\; \nu_a + \frac{s_a^2}{\hbar}\,\Lambda_a,
        \label{eq:nu-left}\\
\nu_{R(a)} &\;=\; \nu_a - \frac{c_a^2}{\hbar}\,\Lambda_a.
        \label{eq:nu-right}
\end{align}
These weights are fixed by two requirements: the child contrast is
\[
\nu_{R(a)}-\nu_{L(a)}=-\,\frac{1}{\hbar}\Lambda_a,
\]
and the probability-weighted child average equals the parent velocity,
$c_a^2\nu_{L(a)}+s_a^2\nu_{R(a)}=\nu_a$. At a leaf,
\begin{equation}\label{eq:omegadot-tree}
\boxed{\displaystyle
\dot\omega_j=\nu_{\mathrm{leaf}(j)}.}
\end{equation}
Equivalently, after choosing a root gauge,
\begin{equation}\label{eq:omegadot-path}
\dot\omega_j
\;=\;
\chi(t)
+\frac{1}{\hbar}\sum_{a\in\mathrm{anc}(j)}
\alpha_{a,j}\,
\frac{g_{\theta_a\theta_a}^{-1}\partial_{\theta_a}E}
     {\sin(2\theta_a)},
\qquad
\alpha_{a,j}=
\begin{cases}
\sin^2\theta_a, & j\in L(a),\\
-\cos^2\theta_a, & j\in R(a).
\end{cases}
\end{equation}

\paragraph*{Global-phase gauge fix.}
The uniform component of $\dot{\bm\omega}$ is unobservable. For numerical hygiene it is projected out after the top-down pass,
\begin{equation}\label{eq:gauge-fix}
\dot\omega_j \;\leftarrow\; \dot\omega_j \;-\; \tfrac{1}{|S|}\sum_k \dot\omega_k ,
\end{equation}
which leaves the prepared ray unchanged. For a chart with $\omega_r=0$, instead subtract $\dot\omega_r$ from every phase velocity and retain the components $j\ne r$.

\paragraph*{Cost and singular coordinates.}
The real-time RHS uses the same Euclidean derivatives as one ordinary shift-rule gradient: $4|\bm\theta|+2|\bm\omega|$ energy evaluations, followed by one bottom-up pass and one top-down pass through the active tree. On a pruned support with $k$ active leaves, the active split tree has $k-1$ free $\theta$ angles and $\mathcal O(k)$ nodes, so the classical conversion from $(\nabla_{\bm\theta}E,\nabla_{\bm\omega}E)$ to $(\dot{\bm\theta},\dot{\bm\omega})$ is linear in $k$. No auxiliary metric circuits, dense matrix assembly, matrix inversion, or per-leaf phase-shifted real-part reconstruction is required. The formulas above assume an open chart, $\sin(2\theta_a)\neq0$ and positive metric entries; the chart-singularity limits are handled by the singular-coordinate initialisation of Appendix~\ref{app:singular-init}.

\section{Chart-singularity initialisation in \texorpdfstring{$(\bm\theta,\bm\omega)$}{(theta,omega)}}\label{app:singular-init}

A chart singularity is a parameter configuration at which an active angle $\theta_a$ has reached one of its boundary values $\theta_{\mathrm{snap}}\in\{0,\pm\pi/2,\pi\}$, so that one of the two child branches of node $a$ carries zero amplitude. Write $R$ for that closed, zero-amplitude subtree, $\mathcal L(R)\subseteq S$ for its active leaves, and $m_a=|\mathcal L(R)|$. By Theorem~\ref{thm:diag-metric} the weight of any parameter $\theta_b$ lying strictly under $R$ is a path product (Eq.~\eqref{eq:g-theta}) that now contains the vanishing factor $\sin^2\theta_a$ or $\cos^2\theta_a$, so $w_b=0$ and the natural-gradient update $\dot\theta_b=-\partial_{\theta_b}E/w_b$ takes the indeterminate form $0/0$. The entire subtree $R$ is thus underdetermined by the metric. Such configurations are measure-zero and arise in practice only at a structured start: a Hartree--Fock initialisation pins every active angle to $0$ or $\pi/2$, placing the whole tree at a singularity with only the Hartree--Fock leaf occupied.

The resolution rests on a single observation: at the singularity the update is set by the Hamiltonian, not by the metric. On any open chart the natural-gradient / imaginary-time step $-g^{-1}\nabla E$ reproduces, to first order in the step size, the orthogonal projection of the exact Hilbert-space flow $-\eta(H-E)\ket{\psi_o}$ onto the tangent space of the ansatz (this is the defining property of the metric-raised update; see Appendix~\ref{app:ng-omega-equiv}). As $R$ is opened, the tangent directions that drop out of the metric (the buried parameters under $R$, together with $\theta_a$ itself) are precisely the ones that span the closed-leaf subspace $\mathrm{span}\{\ket j: j\in\mathcal L(R)\}$. The well-defined limit of the metric-raised update is therefore the projection of the exact flow onto that subspace,
\begin{equation}\label{eq:singinit-target}
-\,\eta\!\sum_{j\in\mathcal L(R)} F_j\,\ket j,\qquad F_j \;:=\; \braket{j|H|\psi_o},
\end{equation}
where $\ket{\psi_o}$ is the current state at $\theta_a=\theta_{\mathrm{snap}}$ (it has zero amplitude on $\mathcal L(R)$, so the $-E$ term drops). This is the $0/0$ limit evaluated in closed form: the gradient $\partial_{\theta_b}E$ and the weight $w_b$ vanish together as $R$ opens, and their ratio is fixed entirely by the local Hamiltonian amplitudes $F_j$. The same projection with $\eta\mapsto i\,dt$ gives the real-time target $-i\,dt\sum_j F_j\ket j$. The rest of this appendix does two things: it reads the $F_j$ off the energy oracle using $(\bm\theta,\bm\omega)$ shift-rule queries alone (no auxiliary metric circuit, no Cartesian state vector), and it writes the resulting amplitudes back into the chart in one analytic step. From a Hartree--Fock start the whole sweep costs $8(|S|-1)$ shift-rule evaluations, comparable to the $4|\bm\theta|+2|\bm\omega|=6(|S|-1)$ evaluations of a single Euclidean gradient, so the treatment is asymptotically free.

\paragraph*{Geometry of the singular node.} The snap value $\theta_{\mathrm{snap}}$ records which child of $a$ is closed. Define the snap sign
\begin{equation}\label{eq:snap-sign}
\sigma_a \;:=\;
\begin{cases}+1, & \theta_{\mathrm{snap}}\in\{0,-\pi/2\},\\[2pt] -1, & \theta_{\mathrm{snap}}\in\{\pi,\pi/2\},\end{cases}
\end{equation}
as the unique sign for which the chart step $\theta_a\leftarrow\theta_{\mathrm{snap}}+\sigma_a\alpha$ opens the closed branch with a positive amplitude $+\alpha+\mathcal O(\alpha^3)$ for small $\alpha>0$ (check $\sin(\theta_{\mathrm{snap}}+\sigma_a\alpha)$ when the sine branch closes, $\cos(\theta_{\mathrm{snap}}+\sigma_a\alpha)$ when the cosine branch closes). Let
\begin{equation}\label{eq:ra-def}
r_a \;:=\; \prod_{b\in\mathrm{anc}(a)} f_b(\theta_b), \qquad f_b\in\{\cos,\sin\},
\end{equation}
be the amplitude funnelled down to node $a$ along its open ancestors, with $f_b$ selected as in Theorem~\ref{thm:diag-metric} ($\cos$ on a downward step, $\sin$ on an upward step); equivalently $r_a=\sqrt{w_a}$.

\paragraph*{The Hamiltonian amplitudes are frame-independent.} The number $F_j=\braket{j|H|\psi_o}$ depends only on the current state $\ket{\psi_o}$, which carries zero amplitude throughout $R$. It is therefore independent of how the closed subtree is parametrised, and the inner angles of $R$ may be chosen freely when probing a particular $F_j$. This freedom is used twice: first to concentrate the (not-yet-present) amplitude of $R$ on a single chosen leaf, so that one shift-rule query isolates one $F_j$; and afterwards to overwrite those inner angles with the amplitude profile to be installed.

\paragraph*{Reading $F_j$ from the energy oracle.} Fix a leaf $j\in\mathcal L(R)$ and set the inner angles of $R$ so that opening $\theta_a$ would send all of $R$'s amplitude onto $j$ with phase $\omega_j$ (walk the root-to-$j$ path inside $R$, setting each active angle to $0$ on a downward step and $\pi/2$ on an upward step). Holding everything else fixed, the state as a function of $\theta_a$ alone is
\begin{equation}\label{eq:psi-of-thetaA}
\ket{\psi(\theta_a)} \;=\; \ket{\chi} \;+\; r_a\bigl[\,f^{\mathrm{op}}(\theta_a)\,\ket{\phi} \;+\; f^{\mathrm{cl}}(\theta_a)\,e^{i\omega_j}\ket{j}\,\bigr],
\end{equation}
where $\ket\chi$ is the fixed part of the state on leaves outside node $a$'s cone, $\ket\phi$ is the normalised open-sibling state, and $f^{\mathrm{op}},f^{\mathrm{cl}}\in\{\cos,\sin\}$ are the open- and closed-branch factors. At the singularity $f^{\mathrm{cl}}(\theta_{\mathrm{snap}})=0$, while $f^{\mathrm{op}}(\theta_{\mathrm{snap}})=\pm1$ sits at an extremum; hence the two branch derivatives are $\partial_{\theta_a}f^{\mathrm{op}}|_{\mathrm{snap}}=0$ and $\partial_{\theta_a}f^{\mathrm{cl}}|_{\mathrm{snap}}=\sigma_a$ (the latter by the defining property of $\sigma_a$). Therefore $\partial_{\theta_a}\ket\psi|_{\mathrm{snap}}=\sigma_a r_a\,e^{i\omega_j}\ket j$, and since $E(\theta_a)=\braket{\psi(\theta_a)|H|\psi(\theta_a)}$,
\begin{equation}\label{eq:dE-at-snap}
\partial_{\theta_a}E\big|_{\theta_{\mathrm{snap}}}
\;=\; 2\,\mathrm{Re}\,\braket{\partial_{\theta_a}\psi|H|\psi}\big|_{\mathrm{snap}}
\;=\; \sigma_a\cdot 2 r_a\,\mathrm{Re}\bigl[e^{-i\omega_j}F_j\bigr].
\end{equation}
The four-term shift rule~\eqref{eq:shift-rule} applied to $\theta_a$ at the current $\omega_j$ thus returns $\mathrm{Re}[e^{-i\omega_j}F_j]$ up to the known factor $\sigma_a\,2r_a$. The four cases $\theta_{\mathrm{snap}}\in\{0,\pi,\pi/2,-\pi/2\}$ are collected in Table~\ref{tab:snap-cases}.

\begin{table}[h]
\centering
\begin{tabular}{cccc}
\toprule
\rowcolor{headcolor}[\tabcolsep][\tabcolsep]
$\theta_{\mathrm{snap}}$ & closed & $\partial_{\theta_a}E/(2r_a\mathrm{Re}[e^{-i\omega_j}F_j])$ & $\sigma_a$ \\
\midrule
$0$       & right ($\sin$) & $+1$ & $+1$ \\
$\pi$     & right ($\sin$) & $-1$ & $-1$ \\
$\pi/2$   & left ($\cos$)  & $-1$ & $-1$ \\
$-\pi/2$  & left ($\cos$)  & $+1$ & $+1$ \\
\bottomrule
\end{tabular}
\caption{Snap sign $\sigma_a$ at each chart singularity. The two natural definitions, the sign of $\partial_{\theta_a}E$ and the direction of the opening step, agree.}\label{tab:snap-cases}
\end{table}

Applying the four-term shift rule~\eqref{eq:shift-rule} to $\theta_a$ at the current $\omega_j$ thus returns $D^{\mathrm{Re}}_j:=\sigma_a\,2r_a\,\mathrm{Re}[e^{-i\omega_j}F_j]$; repeating it in the rotated frame $\omega_j\mapsto\omega_j+\pi/2$ returns the imaginary part, since $e^{-i(\omega_j+\pi/2)} = -i\,e^{-i\omega_j}$ exchanges $\mathrm{Re}\leftrightarrow\mathrm{Im}$. The two outputs are therefore
\begin{align}
D^{\mathrm{Re}}_j &\;=\; \sigma_a\cdot 2 r_a\,\mathrm{Re}\bigl[e^{-i\omega_j}F_j\bigr], \label{eq:Dre}\\
D^{\mathrm{Im}}_j &\;=\; \sigma_a\cdot 2 r_a\,\mathrm{Im}\bigl[e^{-i\omega_j}F_j\bigr], \label{eq:Dim}
\end{align}
and the complex amplitude is recovered as
\begin{equation}\label{eq:Fj-recon}
F_j \;=\; \sigma_a\,e^{+i\omega_j}\,\frac{D^{\mathrm{Re}}_j + i\,D^{\mathrm{Im}}_j}{2 r_a}.
\end{equation}
Each leaf costs $8$ shift-rule evaluations (four per frame), and the full profile $\{F_j\}_{j\in\mathcal L(R)}$ is reconstructed in $8 m_a$ queries; the inner-path encoding of each leaf is a classical parameter rewrite and needs no extra circuit.

\paragraph*{The initialisation step.} With $\{F_j\}$ in hand, the target amplitudes of Eq.~\eqref{eq:singinit-target} are installed in chart coordinates: open the singular angle by $\theta_a\leftarrow\theta_{\mathrm{snap}}+\sigma_a\alpha$, distribute the closed-leaf magnitudes through the inner split angles of $R$, and set the leaf phases. Writing $\|F\|^2:=\sum_{j\in\mathcal L(R)}|F_j|^2$, the leading-order match gives
\begin{align}
Y_j \;&=\; \frac{|F_j|}{\|F\|}, &
\alpha \;&=\; \kappa\,\frac{\|F\|}{r_a}, \label{eq:singinit-step}\\
\phi_j \;&=\; \omega_j + \mathrm{atan2}\!\bigl(\sigma_a D^{\mathrm{Im}}_j,\,\sigma_a D^{\mathrm{Re}}_j\bigr) + \varphi, \notag
\end{align}
where
\begin{equation}\label{eq:singinit-kappa}
(\kappa,\varphi) \;=\;
\begin{cases}
(\eta,\;+\pi), & \text{imaginary time (VITE),}\\
(dt,\;-\pi/2), & \text{real time.}
\end{cases}
\end{equation}
The non-negative magnitudes $\{Y_j\}$ are loaded into the inner angles of $R$ by the standard $\mathrm{atan2}(\|Y_R\|,\|Y_L\|)$ recursion of \sref{sec:circuit}, which lands in $[0,\pi/2]$ at every node; the per-leaf phases $\{\phi_j\}$ are written directly into the leaf coordinates $\omega_j$ (the chart's phase block is per-leaf, so no inner recursion is required). If some $|F_j|=0$ the corresponding branch of $R$ is itself closed, creating a nested singularity one level down; the step then recurses into $R$ to resolve it.

\paragraph*{Correctness.} Using the closed-branch factor $f^{\mathrm{cl}}(\theta_{\mathrm{snap}}+\sigma_a\alpha)=+\alpha+\mathcal O(\alpha^3)$ (guaranteed by the choice of $\sigma_a$) together with Eqs.~\eqref{eq:singinit-step}--\eqref{eq:singinit-kappa}, the chart amplitude installed on each leaf $j\in\mathcal L(R)$ is
\begin{equation}\label{eq:singinit-match}
\alpha\,r_a\,Y_j\,e^{i\phi_j}
\;=\;
\begin{cases}
-\,\eta\,F_j, & \text{(VITE)},\\
-\,i\,dt\,F_j, & \text{(real time)},
\end{cases}
\end{equation}
which is exactly the target~\eqref{eq:singinit-target} (and its real-time analogue) to first order. For real $H$ and real initial $\bm\omega$ every $D^{\mathrm{Im}}_j$ vanishes, $\phi_j$ collapses to $\omega_j+\{0,\pi\}$, and the step reduces to the real-amplitude version in which leaf signs are absorbed into the inner $\mathrm{atan2}$ recursion.

\paragraph*{Role of the snap sign.} The reconstruction~\eqref{eq:Fj-recon} carries $\sigma_a$ explicitly. The magnitudes $Y_j=|F_j|/\|F\|$ and the step size $\alpha=\kappa\|F\|/r_a$ are sign-invariant; only $\arg F_j$ inherits the $\sigma_a$ correction, which is why the bare $\mathrm{atan2}(D^{\mathrm{Im}}_j,D^{\mathrm{Re}}_j)$ is replaced by $\mathrm{atan2}(\sigma_a D^{\mathrm{Im}}_j,\sigma_a D^{\mathrm{Re}}_j)$ in~\eqref{eq:singinit-step}. The same sign sets the direction of the opening step $\theta_a\leftarrow\theta_{\mathrm{snap}}+\sigma_a\alpha$.

\paragraph*{Cost.} Resolving one singular node $a$ costs $8 m_a$ shift-rule evaluations and writes its entire closed subtree at once. At a Hartree--Fock start only the root is singular, with closed subtree containing every active leaf but the Hartree--Fock determinant ($m_{\mathrm{root}}=|S|-1$), for a total of $8(|S|-1)$ evaluations. Nested singularities occur only on the measure-zero profiles with some $|F_j|=0$, so in generic use no recursion is triggered and the sweep cost is exactly $8 m_{\mathrm{root}}$, comparable to the $4|\bm\theta|+2|\bm\omega|=6(|S|-1)$ shift-rule queries of one Euclidean gradient (Eq.~\eqref{eq:shift-rule}), hence asymptotically free relative to the variational iterations themselves.

\section{Derivation of the four-term shift rule}\label{app:shift-proof}

The shift rule of Eq.~\eqref{eq:shift-rule} is derived here. The argument is elementary and does not invoke the general parameter-shift machinery of Ref.~\cite{Wierichs2022}; the connection to that family is recorded at the end of the appendix.

\paragraph*{Multi-controlled form.} Each tree parameter $\theta_i$ controls exactly one uniformly-controlled rotation in the cascade of \sref{sec:circuit}: an $R_y(2\theta_i)$ on a target qubit $q_i$, controlled on the upper qubits $q_i{+}1,\dots,n{-}1$ being in a specific bitstring $b_i$ determined by the position of node $i$ in the tree.\footnote{Equivalently, the elementary $R_y$ angles produced by the Walsh--Hadamard expansion of the cascade are integer-affine combinations of the $\theta_i$, but the resulting circuit is unitarily equivalent to the multi-controlled rotation. The shift rule is most naturally derived from the controlled form.} Its generator is
\begin{equation}\label{eq:shift-generator}
G_i \;=\; -\,\ket{b_i}\!\bra{b_i}_{\mathrm{ctrl}} \otimes Y_{q_i},
\end{equation}
a rank-2 Hermitian operator with spectrum $\sigma(G_i)=\{-1,\,0,\,+1\}$.

\paragraph*{Energy as a degree-2 trigonometric polynomial.} Write $P_i := G_i^2 = \ket{b_i}\!\bra{b_i}\otimes I$, the rank-2 projector onto the subspace where the controls match $b_i$. From $G_i^3 = G_i$ the matrix exponential collapses to
\begin{equation}\label{eq:U-collapse}
U_i(\theta_i) \;=\; e^{-i\theta_i G_i} \;=\; (I - P_i) \;+\; \cos\theta_i\,P_i \;-\; i\sin\theta_i\,G_i .
\end{equation}
Hence $U_i^\dagger H U_i$ is bilinear in $\{(I-P_i),\,\cos\theta_i\,P_i,\,-i\sin\theta_i\,G_i\}$, and
\begin{equation}\label{eq:E-form}
E(\theta_i) \;=\; a_0 + a_1\cos\theta_i + b_1\sin\theta_i + a_2\cos 2\theta_i + b_2\sin 2\theta_i ,
\end{equation}
i.e.\ $E$ is a real trigonometric polynomial of degree at most $2$ in $\theta_i$. Differentiating and grouping by frequency,
\begin{align}
\partial_{\theta_i} E(\theta_i) &\;=\; f_1(\theta_i) + 2\,f_2(\theta_i),\\
f_r(\theta) &\;:=\; b_r\cos r\theta - a_r\sin r\theta .
\end{align}

\paragraph*{Inversion at $\pi/4$ and $3\pi/4$.} The shifted differences pick out exactly these two harmonics: a direct sum-to-product calculation gives, for any shift $s$,
\begin{equation}\label{eq:shift-diff}
E(\theta_i + s) - E(\theta_i - s) \;=\; 2\sin s\,f_1(\theta_i) + 2\sin 2s\,f_2(\theta_i) .
\end{equation}
Choosing the symmetric pair of shifts $s_1 = \pi/4$ and $s_2 = 3\pi/4$ and demanding that the linear combination $\alpha\,[E(\theta_i{+}s_1){-}E(\theta_i{-}s_1)] - \beta\,[E(\theta_i{+}s_2){-}E(\theta_i{-}s_2)]$ reproduce $f_1 + 2f_2$ yields the $2\times 2$ system
\begin{equation}
\begin{pmatrix}\sqrt 2/2 & \sqrt 2/2 \\ 1 & -1\end{pmatrix}\!\!\begin{pmatrix}\alpha\\-\beta\end{pmatrix} = \begin{pmatrix}1/2\\1\end{pmatrix},
\end{equation}
whose unique solution is $\alpha=(2+\sqrt 2)/4$ and $\beta=(2-\sqrt 2)/4$, giving Eq.~\eqref{eq:shift-rule}.

The rule applies uniformly to every tree parameter, irrespective of $n$, of the depth of node $i$, or of the active-leaves set $S$ of \sref{sec:pruning}: the seemingly large integer prefactors that appear when the cascade is expanded into elementary $R_y$+CNOT primitives (the Walsh--Hadamard coefficients of the synthesis) cancel identically, since the underlying multi-controlled rotation always has rank-$2$ generator and frequency support contained in $\{1,2\}$.

\paragraph*{Connection to the Wierichs equispaced family.} The same rule is the $R=2$ instance of the equispaced parameter-shift family of Wierichs, Izaac and Killoran~\cite{Wierichs2022}: for any expectation $E(\theta)$ that is a real trigonometric polynomial with positive integer frequencies bounded by $R$,
\begin{align}
\partial_\theta E &\;=\; \sum_{\ell=1}^{R} w_\ell\bigl[E(\theta+\mu_\ell) - E(\theta-\mu_\ell)\bigr], \label{eq:wierichs}\\
\mu_\ell &\;=\; \frac{(2\ell-1)\pi}{2R}, \label{eq:wierichs-shifts}\\
w_\ell &\;=\; \frac{(-1)^{\ell-1}}{4R\sin^2(\mu_\ell/2)}. \label{eq:wierichs-weights}
\end{align}
Equation~\eqref{eq:E-form} shows that the hypothesis of Eq.~\eqref{eq:wierichs} is met with $R=2$. Direct substitution gives the same shifts $\mu_1=\pi/4,\mu_2=3\pi/4$ and weights $w_1=(2+\sqrt 2)/4$, $w_2=-(2-\sqrt 2)/4$. Among all pairs of equispaced shifts that produce an exact rule for degree-$2$ trigonometric polynomials, this choice minimises $w_1^2+w_2^2$ and therefore the shot-noise variance of the gradient estimator under uniform per-circuit shot budgets.

\paragraph*{Two-term rule for the leaf phases.} Each leaf phase $\omega_j$ enters the prepared state through a single factor $e^{i\omega_j}$ acting on the leaf basis state $\ket{j}$, so its generator is the rank-one projector
\begin{equation}
G^\omega_j \;=\; -\,\ket{j}\!\bra{j},
\end{equation}
with spectrum $\sigma(G^\omega_j) = \{-1,\,0\}$. The exponential $V_j(\omega_j) = e^{-i\omega_j G^\omega_j} = (I-\ket{j}\!\bra{j}) + e^{i\omega_j}\ket{j}\!\bra{j}$ then makes $V_j^\dagger H V_j$ linear in $\{(I-\ket{j}\!\bra{j}),\,\cos\omega_j\,\ket{j}\!\bra{j},\,\sin\omega_j\,\ket{j}\!\bra{j}\}$, so
\begin{equation}\label{eq:E-form-omega}
E(\omega_j) \;=\; a_0 + a_1\cos\omega_j + b_1\sin\omega_j,
\end{equation}
a real trigonometric polynomial of degree at most $1$ in $\omega_j$. Specialising Eq.~\eqref{eq:shift-diff} to $R=1$ with $s = \pi/2$ gives $E(\omega_j+\tfrac{\pi}{2}) - E(\omega_j-\tfrac{\pi}{2}) = 2\,(b_1\cos\omega_j - a_1\sin\omega_j) = 2\,\partial_{\omega_j}E$, which is the two-term rule of Eq.~\eqref{eq:shift-rule-omega}; equivalently, it is the $R=1$ instance of Eq.~\eqref{eq:wierichs} with $\mu_1 = \pi/2$ and $w_1 = 1/2$. The rule applies uniformly to every leaf, irrespective of $n$ and of the active-leaves set $S$, and combines with Eq.~\eqref{eq:shift-rule} to give an exact Euclidean gradient at a cost of $4|\bm\theta| + 2|\bm\omega|$ energy evaluations.

\paragraph*{Two-evaluation reduction for real targets.}\label{app:shift-real}

When the prepared state and the observable are real, the amplitude-only chart $\bm\omega=\bm 0$ that underlies ground-state search and imaginary-time evolution, the four-term rule of Eq.~\eqref{eq:shift-rule} reduces to two evaluations per tree parameter. The argument is the self-inverse-decomposition of Refs.~\cite{Mari2021, Kottmann2021}, specialised to the rank-2 generator $G_i$ of Eq.~\eqref{eq:shift-generator}.

\paragraph*{Self-inverse halves.} Let $P_0^{(i)} := I - G_i^2 = I - P_i$ be the projector onto the null eigenspace of $G_i$ (the subspace where the controls do not match $b_i$). Define
\begin{equation}\label{eq:self-inverse-halves}
G_i^{\pm} \;:=\; G_i \pm P_0^{(i)} .
\end{equation}
Each half is a self-inverse involution: $G_i^\pm$ acts as $\pm 1$ on the null space and as $G_i$ on its support, so $(G_i^\pm)^2 = P_0^{(i)} + G_i^2 = I$. The two halves commute (they are simultaneously diagonal in the eigenbasis of $G_i$), and $G_i = \tfrac12(G_i^+ + G_i^-)$, so the rotation factorises into two commuting self-inverse rotations,
\begin{equation}\label{eq:U-factor}
U_i(\theta_i) \;=\; e^{-i\theta_i G_i} \;=\; e^{-i\frac{\theta_i}{2} G_i^+}\,e^{-i\frac{\theta_i}{2} G_i^-} .
\end{equation}
Each factor obeys an ordinary two-term shift rule with shift $\pi/2$ in its own half-angle. Treating the two as independent and recombining reproduces Eq.~\eqref{eq:shift-rule}; this is the generic four-evaluation case, valid for any Hermitian $H$ and any state.

\paragraph*{Collapse under reality.} The generator $G_i = -\ket{b_i}\!\bra{b_i}\otimes Y_{q_i}$ is purely imaginary, so complex conjugation $K$ obeys $K G_i K = -G_i$, hence $K P_\pm K = P_\mp$ for the rank-one $\pm 1$ eigenprojectors of $G_i$ and $K P_0^{(i)} K = P_0^{(i)}$. When $\ket\psi$ and $H$ are real, $\ket{\psi_-} = K\ket{\psi_+}$ and $\ket{\psi_0}$ is real, and the two harmonic amplitudes that distinguish the $G_i^+$ and $G_i^-$ half-derivatives become equal. The two half-rules therefore contribute identically, and the four evaluations of Eq.~\eqref{eq:shift-rule} collapse to a single shifted pair carried by the $G_i^+$ half alone:
\begin{equation}\label{eq:shift-real-app}
\partial_{\theta_i}E \;=\; \braket{H}_{V_+} - \braket{H}_{V_-}, \qquad V_\pm \;=\; U_i\!\bigl(\theta_i \pm \tfrac{\pi}{4}\bigr)\,U_0^{(\pm)}, \qquad U_0^{(\pm)} \;=\; e^{\pm i\frac{\pi}{4} P_i}.
\end{equation}
The appended factor $U_0^{(\pm)}$ is exactly $e^{-i\frac{\pi}{4} G_i^\mp}$ stripped of the rotation $U_i(\mp\pi/4)$ that the shifted evaluation already supplies; up to an irrelevant global phase it is the diagonal phase $e^{\pm i\frac{\pi}{4} P_i}$, a phase $e^{\pm i\pi/4}$ on the subspace where the ancestor qubits match $b_i$ and the identity elsewhere.

\paragraph*{The correction reuses the phase-block compiler.} Because $P_i = \ket{b_i}\!\bra{b_i}\otimes I$ is diagonal in the computational basis and is supported entirely on the control (ancestor) qubits of node $i$, the unitary $U_0^{(\pm)}$ commutes with $U_i$ and with every gate acting at or below the target qubit $q_i$; it can therefore be moved to the end of the circuit and folded into the observable. There it coincides with the diagonal $R_z$ phase block of the complex ansatz, restricted to the qubits above $q_i$ (the ancestors, not including the target), with the per-leaf phases chosen to place $\pm\pi/4$ on every active leaf in the subtree of node $i$ and $0$ on the rest of the support. Concretely, the subtree-indicator phase pattern $\varphi_i(s) = \pm\tfrac{\pi}{4}\,\mathbf{1}[s\in\mathrm{subtree}(i)]$ lies in the range of the affine $\bm\omega\mapsto$ leaf-phase map synthesised by the pruning compiler, so $U_0^{(\pm)}$ is compiled by the same exact-integer Gaussian elimination of Appendix~\ref{app:compiler} and inherits the linear-in-$k$ CNOT scaling of Eq.~\eqref{eq:cnot-bound}. No multi-controlled gate and no ancilla are introduced. As with the four-term rule, and as with the fermionic UCCSD shift rule~\cite{Anselmetti2021, Kottmann2021}, each evaluated derivative carries its own correction block (its own subtree pattern $b_i$), appended for that gradient evaluation only.

\paragraph*{Cost and scope.} The real-target tree gradient thus costs $2|\bm\theta|$ energy evaluations, and one quantum-natural-gradient or imaginary-time step is computable from $2P+1$ evaluations rather than $4P+1$. The reduction is specific to the real amplitude chart: the complex-amplitude real-time runs, which carry genuine leaf phases $\bm\omega\neq\bm 0$, break the reality condition $\ket{\psi_-}=K\ket{\psi_+}$ and retain the four-term rule of Eq.~\eqref{eq:shift-rule}. Eq.~\eqref{eq:shift-real-app} has been verified, with $U_0^{(\pm)}$ compiled through the phase-block construction above, against the four-term gradient to machine precision on the symmetry-adapted molecular supports.

\section{Pruning compiler: gauge-fixing elimination and end-to-end algorithm}\label{app:compiler}

This appendix collects the full algorithmic content of the pruning compiler summarised in \sref{sec:gauge}, and proves that its gauge-fixing step runs in exact integer arithmetic.

\paragraph*{Tree parameters and circuit parameters.} The cascade is assembled one level at a time. At level $\ell$ (acting on target qubit $q=n-1-\ell$, with $2^\ell$ control patterns) the uniformly controlled $R_y$ rotation is synthesised as a row of $2^\ell$ elementary $R_y$ gates interleaved with CNOTs whose controls run in Gray-code order. The $2^\ell$ angles $\bm\theta^{(\ell)}$ carried by the tree nodes at that level are the tree parameters, and the $2^\ell$ elementary $R_y$ angles $\bm\alpha^{(\ell)}$ actually emitted into the cascade the circuit parameters. The two are related by the Gray-reindexed Walsh--Hadamard transform~\cite{Mottonen2004, Mottonen2005},
\begin{equation}\label{eq:walsh-block}
\bm\alpha^{(\ell)} \;=\; 2^{1-\ell}\,G_\ell\,\bm\theta^{(\ell)},
\qquad
G_\ell \;=\; P_\ell\,H_\ell .
\end{equation}
Here $H_\ell=H_1^{\otimes\ell}$ is the Sylvester--Hadamard (Walsh--Hadamard) matrix, $(H_\ell)_{ij}=(-1)^{\langle i,j\rangle}$, and $P_\ell$ is the permutation that reorders its rows into Gray-code order, $(P_\ell)_{i,\mathrm{gray}(i)}=1$. Columns $j$ index the tree parameters; rows $i$ index the circuit parameters (the elementary gates, in Gray order). Each column inherits the active/fixed/inactive label of its tree node (\sref{sec:classify}). Stacking the blocks $G_\ell$ over all levels gives the block-diagonal affine map from tree parameters to circuit parameters, and gauge-fixing is performed independently within each block.

\paragraph*{Stage 1: gauge-fixing by integer Gaussian elimination.} The inactive (gauge) columns of a block do not affect the prepared state, so they may be chosen to zero as many circuit parameters as possible. The elimination operates on the augmented integer array $(\,G_\ell \mid \bm c\,)$, the constant column $\bm c\in\mathbb Z^{2^\ell}$ starting at $\bm 0$:
\begin{itemize}
\item[(i)] \textbf{Fold the fixed nodes.} Each fixed column $j$, pinned to $v_j\in\{0,1\}$ (in units of $\pi/2$), is added into the constant, $\bm c\mathrel{+}=v_j\,G_\ell[:,j]$, and removed.
\item[(ii)] \textbf{Eliminate the inactive columns.} For each inactive column $c$ in turn, take the first row $r$ (in Gray order) with $(G_\ell)_{rc}\neq0$ as pivot, $p:=(G_\ell)_{rc}$, and clear that column from every other row,
\begin{equation}\label{eq:elim}
\mathrm{row}_o \;\leftarrow\; \mathrm{row}_o \;-\; \frac{(G_\ell)_{oc}}{p}\;\mathrm{row}_r
\qquad (o\neq r),
\end{equation}
acting on $G_\ell$ and $\bm c$ simultaneously; the spent pivot row $r$ is then deleted.
\end{itemize}
This is ordinary Gaussian elimination, restricted to the gauge columns and carried out on the Walsh--Hadamard block. Each inactive node consumes one elementary $R_y$ gate (the pivot row), so the number of surviving circuit parameters at level $\ell$ drops by exactly the number of inactive nodes there. A surviving row that is identically zero is an angle pinned to $0$ and is deleted from the circuit; each remaining row $i$ gives a surviving circuit parameter
\begin{equation}\label{eq:surviving-angle}
\alpha_i \;=\; 2^{1-\ell}\!\sum_{a}(G_\ell)_{ia}\,\theta_a \;+\; \tfrac{\pi}{2}\,2^{1-\ell} c_i ,
\end{equation}
an affine function of the active tree parameters $\theta_a$ alone. The word ``Gaussian elimination'' is meant literally here, with one caveat, now made precise: it is performed in exact integer arithmetic.

\begin{lemma}[Exact integer elimination]\label{lem:integer-elim}
Carried out on the integer array $(\,G_\ell\mid\bm c\,)$, the elimination~\eqref{eq:elim} uses only integer row operations. At every step the pivot $p$ is $\pm$ a power of two and divides every entry of its pivot row, so $\mathrm{row}_r/p$ is an integer vector and the working array stays integer-valued throughout. Consequently every surviving circuit parameter~\eqref{eq:surviving-angle} is an exact integer-affine function of the active tree parameters (up to the fixed scale $2^{1-\ell}$ and the unit $\pi/2$), and which circuit parameters vanish is decided exactly, with no floating-point tolerance.
\end{lemma}

\noindent The only step in~\eqref{eq:elim} that can leave the integers is the division of the pivot row by $p$; once $p$ divides every entry of its own row, each multiplier $(G_\ell)_{oc}/p$ is an integer and every update preserves integrality, so the lemma reduces to a single divisibility statement about the Walsh--Hadamard matrix. The Gray permutation $P_\ell$ only relabels rows and so leaves both the pivot order and every pivot value unchanged; the claim is therefore really about $H_\ell$, and it is proved (in the sharper form $p=\pm 2^{\,|r|}=\pm\gcd(\mathrm{row}_r)$, where $|r|$ is the Hamming weight of the pivot row index) in Appendix~\ref{app:integrality}.

\paragraph*{Stage 2: CNOT cancellation.} The triangular cascade has palindromic structure: each pair of CNOTs on a given control--target pair brackets a sub-cascade of elementary $R_y$ gates. When Stage~1 has emptied the sub-cascade between two identical CNOTs, the two CNOTs compose to the identity and cancel. The cancellation is local and is iterated to fixpoint within each level.

\paragraph*{Constant bits and qubit reordering.} A constant-bit removal pass runs before the gauge fixing: a bit position whose value is the same for every leaf in $S$ is fully redundant, so the corresponding qubit is initialized to that fixed value and removed from the variational tree, reducing the effective qubit count. The number of inactive nodes, and therefore the size of the gauge group available to Stage~1, also depends on the order in which the varying qubits are placed in the tree, and the compiler searches over their permutations with a depth-first branch-and-bound procedure that maximizes the inactive-node count, at a cost that grows exponentially in the number of varying bits. The CNOT bounds of \sref{sec:cnot-bound} assume this reordering.

\paragraph*{End-to-end algorithm.} The pipeline takes as input the qubit count $n$ and the active-leaves set $S$, and returns a parametrized circuit (Algorithm~\ref{alg:compile}). Both the surviving angles and the CNOT count depend only on $n$ and $S$.

\begin{figure}[h]
\centering
\fbox{\begin{minipage}{0.95\columnwidth}
\textbf{Algorithm 1.} \textsc{PruneAnsatz}$(n, S)$\\
\textbf{Input:} qubit count $n$, active-leaves set $S\subset\{0,1\}^n$.\\
\textbf{Output:} parametrized circuit on the active parameters $\bm\theta_A$.
\begin{enumerate}
    \item[1.] Strip constant bit positions from $S$; let $n'\le n$ be the residual qubit count.
    \item[2.] (Optional; assumed by the CNOT bounds of \sref{sec:cnot-bound}.) Search for a qubit permutation $\pi^\star$ that maximises the inactive-node count, by depth-first branch-and-bound over the symmetric group on the $n'$ varying bits; otherwise $\pi^\star$ is the identity.
    \item[3.] Build the binary tree on $\pi^\star(S)$ and label every internal node \textsc{Active}/\textsc{Fixed}/\textsc{Inactive} (\sref{sec:classify}).
    \item[4.] Form the level-block matrices $G_\ell$ and the constant vector $\bm c$ of Eq.~\eqref{eq:walsh-block}; for each level, perform the exact integer Gaussian elimination~\eqref{eq:elim} on the inactive columns to maximise the number of zeroed circuit parameters (Lemma~\ref{lem:integer-elim}).
    \item[5.] Emit the triangular cascade level by level, deleting every $R_y$ whose angle was zeroed in step~4 and cancelling pairs of CNOTs that bracket an empty sub-cascade.
\end{enumerate}
\end{minipage}}
\refstepcounter{algorithm}\label{alg:compile}
\end{figure}

\section{Integrality of the gauge-fixing elimination}\label{app:integrality}

This appendix proves Lemma~\ref{lem:integer-elim}: the gauge-fixing elimination~\eqref{eq:elim} stays in the integers, and each pivot is $\pm$ a power of two equal to the gcd of its row. Because the Gray permutation $P_\ell$ in $G_\ell=P_\ell H_\ell$ only relabels rows (it permutes the pivot order but changes no pivot value), it suffices to treat the lowest-index Gaussian elimination of an arbitrary set of columns of the Walsh--Hadamard matrix $H_\ell$ itself. Throughout, $\langle i,j\rangle$ is the dot product of the binary strings $i,j\in\{0,1\}^\ell$ modulo $2$, $|i|$ is the Hamming weight of $i$, and $i\subseteq j$ means the bits of $i$ are contained in those of $j$.

The argument has three steps. First, a factorization $H_\ell=Z^\top D Z$ that exposes all the arithmetic. Second, a unimodularity lemma for the integer triangular factor $Z$. Third, a flag-minor identity that combines the two and pins each pivot to $\pm2^{|r|}$.

\paragraph*{An exact triangular factorization.}

Let $J=\left(\begin{smallmatrix}1&1\\0&1\end{smallmatrix}\right)$ and set
\[
Z \;=\; J^{\otimes\ell},
\qquad
D \;=\; \operatorname{diag}\!\big((-2)^{|i|}\big)_{i\in\{0,1\}^\ell}.
\]
The matrix $Z$ is the zeta matrix of the Boolean lattice: $Z_{ij}=[\,i\subseteq j\,]$, unit upper-triangular with $0,1$ entries, hence integer and unimodular. The diagonal $D$ records a signed power of two per index.

\begin{lemma}[Triangular factorization of the Walsh--Hadamard matrix]\label{lem:factorization}
$H_\ell = Z^\top D\,Z$.
\end{lemma}

\noindent\textbf{Proof.}
Expanding the product entrywise and using $Z_{ki}=[k\subseteq i]$,
\[
(Z^\top D Z)_{ij}
=\sum_{k}Z_{ki}\,D_{kk}\,Z_{kj}
=\sum_{k\subseteq i\cap j}(-2)^{|k|}
=\prod_{b\in i\cap j}\big(1+(-2)\big)
=(-1)^{|i\cap j|}
=(-1)^{\langle i,j\rangle},
\]
where the third equality factorizes the sum over subsets $k$ of $i\cap j$ bit by bit. This is $(H_\ell)_{ij}$.\hfill$\square$

This is the LU factorization of $H_\ell$ with unit lower-triangular factor $L=Z^\top$ and upper factor $U=DZ$:
\begin{equation}\label{eq:factorization-display}
H_\ell \;=\; \underbrace{Z^\top}_{L}\;\underbrace{D\,Z}_{U},
\qquad
Z=J^{\otimes\ell},\quad D=\operatorname{diag}\big((-2)^{|i|}\big).
\end{equation}
In particular the Smith normal form of $H_\ell$ is $\operatorname{diag}(2^{|i|})$: the elementary divisors are exactly the powers of two $2^{0},2^{1},\dots,2^{\ell}$, with multiplicities $\binom{\ell}{0},\dots,\binom{\ell}{\ell}$. This already shows that only powers of two can appear as pivots; the work below is to show that the lowest-index pivot order realises them exactly.

\paragraph*{Greedy unimodularity of the triangular factor.}

Fix a set of columns $C$ of $Z$. Scanning rows in increasing index order and keeping a row whenever it is independent of those kept so far yields the greedy least-index basis $R$ of the column space, the row set the lowest-index elimination actually pivots on.

\begin{lemma}[Greedy unimodularity]\label{lem:greedy-unimodular}
For every column set $C$ of $Z=J^{\otimes\ell}$ with greedy least-index basis $R$, the square submatrix $Z[R,C]$ is unimodular: $\det Z[R,C]=\pm1$.
\end{lemma}

\noindent\textbf{Proof.}
By induction on $\ell$. For $\ell=1$, $Z=J=\left(\begin{smallmatrix}1&1\\0&1\end{smallmatrix}\right)$ and every greedy minor is $\pm1$. For the step, split indices on the most significant bit and write
\[
Z_\ell=\begin{pmatrix}A&A\\0&A\end{pmatrix},\qquad A:=Z_{\ell-1},
\]
so that rows and columns carrying a leading $0$ (the ``top'' block) have strictly smaller index than those carrying a leading $1$ (the ``bottom'' block); the greedy scan therefore exhausts the top rows before any bottom row. Let $Y_0$ be the lower-order parts of the columns of $C$ with leading bit $0$, and $Y_1$ those with leading bit $1$.

Apply the unimodular column operations subtracting each leading-$1$ column from the matching leading-$0$ column when both are present, i.e.\ on the shared part $Y_0\cap Y_1$. Ordering the columns as $\{(0,y):y\in Y_0\}$, then $\{(1,y):y\in Y_1\setminus Y_0\}$, then $\{(1,y):y\in Y_0\cap Y_1\}$, brings the matrix to the block-staircase form
\[
\begin{pmatrix}
A[R_0,\,Y_0] & A[R_0,\,Y_1\setminus Y_0] & 0\\[2pt]
0 & A[R_1,\,Y_1\setminus Y_0] & A[R_1,\,Y_0\cap Y_1]
\end{pmatrix}.
\]
The bottom-only columns $\{(1,y):y\in Y_0\cap Y_1\}$ vanish on the top rows, so the generalized Laplace expansion along them has a single surviving term, giving
\begin{equation}\label{eq:greedy-factorization}
\det Z_\ell[R,C]
=\pm\,\det A\big[R_0,\,Y_0\cup Y_1\big]\;\cdot\;\det A\big[R_1,\,Y_0\cap Y_1\big],
\end{equation}
where $Y_0\cup Y_1=Y_0\sqcup(Y_1\setminus Y_0)$. It remains to identify $R_0,R_1$ as greedy bases of $A$, so that induction applies to both factors.

The top rows are $(0,x)\mapsto\big(A[x,Y_0],\,A[x,Y_1\setminus Y_0],\,0\big)$. As $Y_0$ and $Y_1\setminus Y_0$ are disjoint there is no column repetition, so a top row is independent of its predecessors iff $A[x,\,Y_0\cup Y_1]$ is; hence $R_0$ is exactly the greedy least-index basis of $A$ on the columns $Y_0\cup Y_1$, and its span is the set of vectors that vanish on the $Y_0\cap Y_1$ coordinates. Reducing each bottom row $(1,x)\mapsto\big(0,\,A[x,Y_1\setminus Y_0],\,A[x,Y_0\cap Y_1]\big)$ modulo $\operatorname{span}(R_0)$ collapses it to its $Y_0\cap Y_1$ component $A[x,\,Y_0\cap Y_1]$; thus the greedy selection of bottom rows is the greedy least-index basis $R_1$ of $A$ on the columns $Y_0\cap Y_1$. Both factors in~\eqref{eq:greedy-factorization} are then $\pm1$ by the induction hypothesis, and so is their product.\hfill$\square$

\noindent The lowest-index pivot order is essential: $Z$ is not totally unimodular (its minors reach $3$ for $\ell\ge4$), so a generic choice of pivot rows would break unimodularity. It is precisely the greedy least-index choice, the one the compiler makes, that the lemma certifies.

\paragraph*{The pivots are powers of two.}

The two ingredients are now combined. Let $C_k$ be the first $k$ gauge columns processed and $R_k$ the rows pivoted on after $k$ steps. Standard fraction-free elimination identifies the $k$-th pivot with a ratio of leading minors,
\begin{equation}\label{eq:pivot-ratio}
p_k \;=\; \frac{\det H_\ell[R_k,\,C_k]}{\det H_\ell[R_{k-1},\,C_{k-1}]}.
\end{equation}
The next lemma evaluates these minors.

\begin{lemma}[Flag-minor identity]\label{lem:flag-minor}
For the greedy least-index basis $R_k$ of any column set $C_k$ of $H_\ell$,
\[
\det H_\ell[R_k,\,C_k] \;=\; (-2)^{\sum_{r\in R_k}|r|}\;\det Z[R_k,\,C_k]
\;=\;\pm\,2^{\sum_{r\in R_k}|r|}.
\]
\end{lemma}

\noindent\textbf{Proof.}
By the factorization $H_\ell=Z^\top DZ$ and the Cauchy--Binet formula,
\[
\det H_\ell[R_k,C_k]
=\sum_{K}\det\!\big(Z^\top\big)[R_k,K]\;\det\!\big(DZ\big)[K,C_k]
=\sum_{K}\det Z[K,R_k]\;\Big(\textstyle\prod_{m\in K}(-2)^{|m|}\Big)\det Z[K,C_k],
\]
the sum over $k$-subsets $K$ of row indices. Since $L=Z^\top$ is unit lower-triangular, $\det Z[K,R_k]=\det L[R_k,K]$ vanishes unless $K\preceq R_k$ componentwise. By Gale's theorem the greedy least-index basis $R_k$ is componentwise $\le$ every basis of the column space of $C_k$; hence any $K\preceq R_k$ with $K\neq R_k$ fails to be a basis, so $\det Z[K,C_k]=0$. Only $K=R_k$ survives, where $\det L[R_k,R_k]=1$, leaving
\[
\det H_\ell[R_k,C_k]=(-2)^{\sum_{r\in R_k}|r|}\,\det Z[R_k,C_k].
\]
The second factor is $\pm1$ by Lemma~\ref{lem:greedy-unimodular}.\hfill$\square$

Substituting the flag-minor identity into~\eqref{eq:pivot-ratio} and using that $R_k=R_{k-1}\cup\{r_k\}$ adds exactly the new pivot row $r_k$,
\[
p_k=\frac{\pm2^{\sum_{r\in R_k}|r|}}{\pm2^{\sum_{r\in R_{k-1}}|r|}}=\pm\,2^{\,|r_k|}.
\]
Thus the lowest-index pivot at row $r_k$ is $\pm2^{|r_k|}$, a power of two equal to $\gcd(\mathrm{row}_{r_k})$: the gcd cannot be smaller, since the pivot row of $H_\ell$ lies in $2^{|r_k|}\mathbb Z^{2^\ell}$ by the Smith normal form, and cannot be larger, since the pivot divides the unimodular minor $\det Z[R_k,C_k]=\pm1$. Each elimination~\eqref{eq:elim} therefore divides the pivot row by a power of two that divides every one of its entries, so the working array stays integral and Lemma~\ref{lem:integer-elim} follows. The Gray permutation $P_\ell$, applied on the left in $G_\ell=P_\ell H_\ell$, merely relabels which physical row plays the role of $r_k$ and leaves every pivot value $\pm2^{|r_k|}$ unchanged. \hfill$\square$

\section{CNOT cost as a Gray-coordinate run count}\label{app:cnot-structure}

The surviving CNOT count of the pruned circuit has an exact combinatorial form: it counts maximal runs of ones in the Gray coordinates of the surviving elementary rotations. This is established here; it is the identity that turns the scaling bound (Appendix~\ref{app:cnot-evidence}) into the single inequality Eq.~\eqref{eq:runcount-conjecture}. Throughout, level $q$ refers to the uniformly controlled rotation on target qubit $q$, realised as the cascade $T(q,q)$ followed by the appended $\mathrm{CX}(q,0)$, with $T(0,q)=[R_y@q]$ and $T(d,q)=T(d{-}1,q)\,\mathrm{CX}(q,q{-}d)\,T(d{-}1,q)$. Number its $2^q$ elementary $R_y$ gates by position $r=0,\dots,2^q-1$; let $\mathcal S\subseteq\{0,\dots,2^q-1\}$ be the positions surviving Stage~1, and write $v_2(i)$ for the $2$-adic valuation of $i\ge1$ and $x_r:=\mathrm{gray}(r)_\beta\in\{0,1\}$ for a fixed bit $\beta$.

\begin{lemma}[Per-level independence]\label{lem:perlevel}
$N_{\mathrm{cx}}=\sum_{q} c_q$, where $c_q$ is the number of CNOTs of the level-$q$ block surviving Stage~2.
\end{lemma}
\noindent\textbf{Proof.} Every CNOT of the level-$q$ block has target $q$, and levels use distinct target qubits, so two CNOTs of different levels are never identical; the Stage-2 rule, which removes only identical pairs, therefore acts within each level. $\square$

\begin{lemma}[Run-count identity]\label{lem:runcount}
For control type $t\in\{0,\dots,q-1\}$ put $\beta=q-1-t$ and let $R_{q,\beta}$ be the number of maximal runs of $1$ in $(x_r)_{r\in\mathcal S}$ (survivors in increasing position order). Then the number of surviving control-$t$ CNOTs is $2R_{q,\beta}$, and
\begin{equation}\label{eq:runcount-identity}
N_{\mathrm{cx}}\;=\;2\sum_{q}\sum_{\beta=0}^{q-1}R_{q,\beta}.
\end{equation}
\end{lemma}

\noindent\textbf{Proof.} The argument proceeds in three steps.

\textbf{(B.0) Control positions.} By induction on $d$, the CNOT immediately before position $i$ ($1\le i\le2^d-1$) in $T(d,q)$ has control $q-1-v_2(i)$. The base $d=1$ has its single CNOT $\mathrm{CX}(q,q{-}1)$ before position $1$ with $q-1-v_2(1)=q-1$. For the step: the center CNOT $\mathrm{CX}(q,q{-}d)$ sits before position $2^{d-1}$ with $q-1-v_2(2^{d-1})=q-d$; positions in the left copy reuse the induction hypothesis verbatim; a position $i$ in the right copy maps to $i'=i-2^{d-1}\in[1,2^{d-1})$, and since $i'<2^{d-1}$ has its lowest set bit below bit $d{-}1$, $v_2(i)=v_2(i')$, so the control is again $q-1-v_2(i)$. As $\mathrm{gray}(i)$ and $\mathrm{gray}(i{-}1)$ differ exactly in bit $v_2(i)$, the control-$t$ CNOTs are precisely the flip-points of the sequence $(x_r)$; there are $2^t$ of them, and for $t=0$ the appended $\mathrm{CX}(q,0)$ is the closing flip, so the flip count is even and (since $x_0=0$) the sequence returns to $0$.

\textbf{(B.1) Reduction to adjacent matching.} The Stage-2 scan removes only CNOTs and treats other-control CNOTs and removed $R_y$ gates as transparent; identical means same control. Hence the control-$t$ CNOTs are processed independently of all other types, and each free $p_j$ (the $j$-th control-$t$ CNOT) matches the next one $p_{j+1}$ exactly when the segment between them contains no surviving $R_y$. This is greedy left-to-right matching on the path $p_1,\dots,p_{N_t}$ with an edge $(j,j{+}1)$ present iff segment $j$ is survivor-free.

\textbf{(B.2) Parity.} Greedy path matching leaves $s\bmod2$ vertices unmatched in each maximal run of $s$ vertices joined by edges; the occupied (survivor-containing) segments are the barriers separating these runs, so the surviving count is the number of odd-size blocks. Since $x$ flips at each $p_j$ with $x_0=0$, segment $s$ has value $x_s=s\bmod2$; a block between consecutive occupied segments $c_a<c_{a+1}$ has size $c_{a+1}-c_a$, odd iff $x_{c_a}\ne x_{c_{a+1}}$, and the two boundary blocks (using $N_t$ even) are odd iff their extreme occupied segment has $x=1$. Writing $O$ for ``occupied, $x{=}1$'',
\[
\#\text{surviving}=[c_0{=}O]+\#\{O/E\text{ transitions}\}+[c_L{=}O]=2\,(\#\,O\text{-runs}).
\]
Finally a survivor has $x=1$ iff it lies in an $O$ segment, and an $O$-run is broken only by an occupied $x{=}0$ survivor, so the $O$-runs are exactly the maximal $1$-runs of $(x_r)_{r\in\mathcal S}$, i.e.\ $R_{q,\beta}$. Summing over $t$ gives Eq.~\eqref{eq:runcount-identity}. $\square$

\noindent A corollary is that the pruned circuit can be emitted directly from $\mathcal S$: for each $\beta$, place the two bracketing CNOTs around each maximal $1$-run of $(x_r)$, without building the full cascade or running the Stage-2 scan. The identity~\eqref{eq:runcount-identity} holds for any qubit ordering; it is the inequality $\sum_{q,\beta}R_{q,\beta}\le E_s-1$ (Eq.~\eqref{eq:runcount-conjecture}), equivalent to Eq.~\eqref{eq:cnot-bound-As}, that requires the reordering.

\medskip
The identity already yields an unconditional linear-in-$k$ bound, independent of the qubit ordering, by controlling the number of surviving rotations.

\begin{lemma}[Rotation count]\label{lem:rotcount}
Let $|\mathcal S_q|$ be the number of surviving elementary $R_y$ gates at level $q$ and $m_q$ the number of occupied (active or fixed) nodes at that level. Then $|\mathcal S_q|\le m_q$, and the pruned circuit retains $N_{R_y}=\sum_q|\mathcal S_q|\le A_s$ elementary rotations.
\end{lemma}

\noindent\textbf{Proof.} At level $q$ the $2^q$ angles of the nodes at that level map to the $2^q$ elementary $R_y$ angles of the cascade through the Gray-reindexed Walsh--Hadamard matrix, which is invertible. The $2^q-m_q$ inactive nodes contribute that many of its columns; being a subset of the columns of an invertible matrix they are linearly independent, so the block-local integer Gaussian elimination on the inactive columns (Appendix~\ref{app:integrality}) pivots on $2^q-m_q$ distinct rows and zeros the corresponding elementary angles, leaving at most $m_q$ nonzero. Summing over levels and using $\sum_q m_q=A_s$ gives $N_{R_y}\le A_s$. $\square$

\begin{proposition}[Unconditional linear scaling]\label{prop:linear}
For every qubit ordering, $N_{\mathrm{cx}}\le 2n\,A_s\le 2n\bigl((n-1)k+1\bigr)=O(n^2k)$, and in particular $N_{\mathrm{cx}}$ is linear in $k$.
\end{proposition}

\noindent\textbf{Proof.} By Lemma~\ref{lem:runcount}, $N_{\mathrm{cx}}=2\sum_q\sum_{\beta=0}^{q-1}R_{q,\beta}$. Each $R_{q,\beta}$ counts maximal $1$-runs in a sequence of length $|\mathcal S_q|$, so $R_{q,\beta}\le|\mathcal S_q|\le m_q$ by Lemma~\ref{lem:rotcount}; with at most $n$ coordinates $\beta$ per level, $N_{\mathrm{cx}}\le 2\sum_q q\,m_q\le 2n\sum_q m_q=2nA_s$. The envelope $A_s\le(n-1)k+1$ (Appendix~\ref{app:cnot-evidence}) gives the $O(n^2k)$ form. $\square$

\noindent Proposition~\ref{prop:linear} fixes the linear-in-$k$ scaling unconditionally; the tighter structure-dependent bound $N_{\mathrm{cx}}\le 2(A_s+k-2)$ (Appendix~\ref{app:cnot-evidence}) sharpens the $A_s$-prefactor from $2n$ to $2$, and that sharpening is the part that requires the inactive-node-maximising reordering.

\paragraph*{Nearest-neighbour architectures.} The counts above assume all-to-all connectivity; restricting to a linear nearest-neighbour chain is benign. By step (B.0) of Lemma~\ref{lem:runcount}, the CNOT preceding position $i$ at level $q$ has control qubit $q-1-v_2(i)$, i.e.\ control--target distance $d=v_2(i)+1$, so distance $d$ occurs $2^{q-d}$ times among the CNOTs of the full level-$q$ cascade (plus the single closing $\mathrm{CX}(q,0)$ at distance $q$): long-range CNOTs are geometrically rare. Since a distance-$d$ CNOT costs $O(d)$ nearest-neighbour CNOTs, the full cascade maps onto the chain at constant-factor overhead, $\sum_{d\ge1}2^{q-d}\,O(d)=O(2^q)$ per level, preserving the dense ceiling of Eq.~\eqref{eq:cnot-bound} up to a constant; nearest-neighbour realisations of uniformly controlled gates are developed in Ref.~\cite{Bergholm2005}. For the pruned circuit the same holds on the symmetry sectors of interest. Because the gauge fixing of Appendix~\ref{app:integrality} depends only on the (basis-independent) Walsh--Hadamard matrix, the surviving rotations are unchanged by the choice of two-qubit realisation, so each level's uniformly controlled rotation may be laid on the chain by that construction. Routing the pruned circuits of \sref{sec:vqe} onto a linear chain with a standard swap network reproduces the target unitary exactly at a small constant two-qubit overhead ($\le 3.4\times$ the all-to-all count across LiH, BeH$_2$, H$_2$O, NH$_3$ and the half-filled Hubbard sectors, with no qubit reordering). The unconditional envelope $O(n)\,N_{\mathrm{cx}}$, giving $O(n^2k)$ under Eq.~\eqref{eq:cnot-bound-As} and $O(n^3k)$ in general, is attained only by arbitrary sparse supports whose active leaves spread across the register so that the per-level control distance grows with $n$.

\section{CNOT-bound evidence and structural form}\label{app:cnot-evidence}

This appendix collects the structure-dependent form of the CNOT bound stated in \sref{sec:cnot-bound}, a comparison with prior sparse-state-preparation constructions (Table~\ref{tab:sparse-prep}), the underlying counting argument, the supporting numerical evidence, and tight-configuration examples. Linear-in-$k$ scaling itself is established unconditionally in Appendix~\ref{app:cnot-structure} (Proposition~\ref{prop:linear}, $N_{\mathrm{cx}}\le 2nA_s=O(n^2k)$ for any ordering); the evidence collected here concerns the tight constant, the sharp bound $N_{\mathrm{cx}}\le 2(A_s+k-2)$ and the empirical per-leaf slope, which is the part controlled by the inactive-node-maximising reordering.

\begin{table*}[t]
\centering
\renewcommand{\arraystretch}{1.2}
\begin{tabular}{llcc}
\toprule
\rowcolor{headcolor}[\tabcolsep][\tabcolsep]
Method & Asymptotic CNOTs & Arbitrary state & Diagonal metric \\
\midrule
Möttönen et al.~\cite{Mottonen2005}        & $\Theta(2^n)$    & yes    & no \\
Shende--Bullock--Markov~\cite{Shende2006, Shende2004}  & $\sim\frac{23}{48}\cdot 4^n$ & yes    & no \\
Plesch--Brukner~\cite{Plesch2011}          & $\Theta(2^n)$    & yes    & no \\
Iten et al.~\cite{Iten2016}                & $O(nk)$          & sparse & no \\
Gleinig--Hoefler~\cite{Gleinig2021}        & $O(nk)$          & sparse & no \\
\textbf{This work}                         & $\le 2nk-2$  & sparse$^\dagger$ & \textbf{yes} \\
\bottomrule
\end{tabular}
\caption{Comparison with sparse and dense state-preparation constructions. $^\dagger$Complete as a circuit family (\sref{sec:circuit}); in the pruned regime $k\ll 2^n$ the two-qubit count is provably $O(n^2k)$ for any ordering (Proposition~\ref{prop:linear}), and the sharper $\le 2nk-2$ shown, matching the metric-free $O(nk)$ constructions above, is proved for $k=2$ and conjectured for $k\ge3$ under the inactive-node-maximising reordering of Appendix~\ref{app:compiler}. The distinguishing feature is the exactly diagonal metric (last column), which the other constructions lack.}
\label{tab:sparse-prep}
\end{table*}

\paragraph*{Spanned-subtree count $A_s$.}
Let $A_s$ denote the number of internal nodes of the binary subtree spanned by the active-leaves set $S$ (the union of the root-to-leaf paths of $S$), evaluated after the constant-bit removal of Appendix~\ref{app:compiler}. The total active and fixed parameters of the compiler are by construction the elements of this subtree, so $A_s=|\bm\theta_A|+|\bm\theta_F|$, and the subtree has $E_s=A_s+k-1$ edges.

\paragraph*{Structure-dependent bound.}
With the qubit reordering of Appendix~\ref{app:compiler}, the sharpest single-quantity bound is
\begin{equation}\label{eq:cnot-bound-As}
N_{\mathrm{cx}} \;\le\; 2(A_s + k - 2) \;=\; 2(E_s - 1),
\end{equation}
with $A_s$ evaluated on the reordered support. The reordering is necessary: for the identity ordering Eq.~\eqref{eq:cnot-bound-As} can be violated (e.g.\ $S=\{94,109,146,159\}$ at $n=8$ gives $N_{\mathrm{cx}}=52>50=2(A_s+k-2)$, with $A_s=23$ for that ordering), and the permutation that maximises the inactive-node count lowers $A_s$ to $21$ and removes the violation. Since $A_s\le\min(2^n-1,\,(n-1)k+1)$ (the right inequality from the leaf-path counting argument below; the left because $A_s$ counts internal nodes of the complete binary tree on $n$ levels), Eq.~\eqref{eq:cnot-bound-As} implies the two regime bounds of Eq.~\eqref{eq:cnot-bound}.

\paragraph*{Active--fixed split and $F$-dependence.}
Every active node has both children occupied by construction, so the active-internal-node count of the spanned subtree equals the number of pairwise merges of the leaves' root-to-leaf paths, which for $k$ distinct leaves is exactly $k-1$ (verified exhaustively over all $k$-subsets through $n=8$). Writing $F:=|\bm\theta_F|$ for the number of fixed internal nodes in the spanned subtree, $A_s=(k-1)+F$ and Eq.~\eqref{eq:cnot-bound-As} becomes
\begin{equation}\label{eq:cnot-bound-F}
N_{\mathrm{cx}} \;\le\; 4k - 6 + 2F.
\end{equation}
Here $F$ counts the ``single-child'' stretches of the spanned subtree, segments where the active leaves do not branch, and is the only support-dependent quantity in the bound. Its extremal values are determined by $(n,k)$ alone:
\begin{align}
F_{\min}(n,k) &= n - \lceil \log_2 k\rceil, \label{eq:F-min}\\
F_{\max}(n,k) &= k\!\left(n - \lfloor\log_2 k\rfloor\right). \label{eq:F-max}
\end{align}
The minimum is achieved when the $k$ leaves cluster in a single depth-$\lceil\log_2 k\rceil$ subtree: a length-$(n-\lceil\log_2 k\rceil)$ root-to-cluster path of fixed nodes is followed by a (near-)complete binary tree of merges, contributing no further fixed nodes when $k$ is a power of two. The maximum, for $k=2^d$ with $d\le n$, is realised by placing the $k$ leaves in $k$ disjoint depth-$d$ subtrees and descending all $n-d$ remaining levels independently to a single leaf in each subtree; the same construction interpolates for non-powers-of-two. Both extremal formulas were verified by enumeration of all $k$-subsets for $n\le 4$ and by sampling for $n\le 8$, with no violations.

\paragraph*{From structure to $(n,k)$.}
Combining Eq.~\eqref{eq:cnot-bound-F} with Eq.~\eqref{eq:F-max} gives the tight closed-form $(n,k)$ bound
\begin{equation}\label{eq:cnot-bound-nk-tight}
N_{\mathrm{cx}} \;\le\; 4k - 6 + 2k\!\left(n - \lfloor\log_2 k\rfloor\right),
\end{equation}
which exhibits the $\log_2 k$ saving when the leaves cluster. Dropping the $\log$ via the loose envelope $F\le (n-2)k+2$, a leaf-path counting argument, $\sum_l n = nk$ pairs (leaf, ancestor in spanned subtree), with each of the $k-1$ active nodes consuming at least two and each fixed node at least one, yields the simpler form
\begin{equation*}
N_{\mathrm{cx}} \;\le\; 2nk - 2
\end{equation*}
quoted as the sparse bound of Eq.~\eqref{eq:cnot-bound} in the main text. The two bounds coincide at $k=1$ ($N_{\mathrm{cx}}\le 2n-2$); for $k\ge 2$ Eq.~\eqref{eq:cnot-bound-nk-tight} is strictly tighter, by $2k\lfloor\log_2 k\rfloor - 4$, and the gap grows with $k$. In particular, when $k$ is a power of two within a constant factor of $2^n$ the tight bound~\eqref{eq:cnot-bound-nk-tight} collapses to an $n$-independent per-leaf slope $2(n-\log_2 k)=O(1)$, matching the clustered floor, where a support filling a minimal depth-$\lceil\log_2 k\rceil$ subtree costs only $N_{\mathrm{cx}}\approx k$.

\paragraph*{Numerical evidence.}
Under the inactive-node-maximising reordering of Appendix~\ref{app:compiler}, Eq.~\eqref{eq:cnot-bound-As} holds with no violations across exhaustive enumeration of all supports for $n\le 4$ and uniform sampling for $n=5,\dots,8$ (Table~\ref{tab:cnot-evidence}). The bound is saturated by specific supports: e.g.\ $S=\{94,109,146,159\}$ at $n=8$ gives $A_s=21$ and $N_{\mathrm{cx}}=46=2(A_s+k-2)$, whereas the identity ordering on the same support gives $N_{\mathrm{cx}}=52$ and violates the bound. The reordering is essential but is computed by an exponential branch-and-bound; a polynomial-time ordering attaining Eq.~\eqref{eq:cnot-bound-As} is not known, the greedy inactive-node heuristic can violate it for $n\ge 9$, and both that question and the general $k\ge3$ proof are left to future work (\sref{sec:discussion}).

\begin{table}[h]
\centering
\renewcommand{\arraystretch}{1.2}
\begin{tabular}{rrcc}
\toprule
\rowcolor{headcolor}[\tabcolsep][\tabcolsep]
$n$ & configurations tested & maximum ratio & violations \\
\midrule
3 & $247$ (exhaustive)     & $0.60$ & $0$ \\
4 & $65{,}519$ (exhaustive) & $0.71$ & $0$ \\
5 & $600$ (sampled)        & $0.83$ & $0$ \\
6 & $600$ (sampled)        & $0.88$ & $0$ \\
7 & $600$ (sampled)        & $0.91$ & $0$ \\
8 & $600$ (sampled)        & $0.96$ & $0$ \\
\bottomrule
\end{tabular}
\caption{Numerical evidence for Eq.~\eqref{eq:cnot-bound-As} under the inactive-node-maximising reordering. The ratio is $N_{\mathrm{cx}}/[2(A_s+k-2)]$; it is saturated (ratio $1$) by specific supports such as $\{94,109,146,159\}$ at $n=8$, and never exceeded in the sample.}
\label{tab:cnot-evidence}
\end{table}

\paragraph*{Empirical slope summary.}
Proposition~\ref{prop:linear} guarantees linear-in-$k$ scaling but only an $O(n^2)$ per-leaf slope; the empirical slopes are far smaller, the evidence that the tight constant of Eq.~\eqref{eq:cnot-bound-As} is the operative one. For each $n=4,\ldots,8$, up to $80$ uniformly random supports were sampled per $k\in[2,2^n]$ and recorded $N_{\mathrm{cx}}/k$. The minimum (clustered), median (typical) and maximum (worst) observed slopes are reported in Table~\ref{tab:cnot-slopes}. The worst observed slope remains well below the structural envelope $2n$ across the range tested, and the median (typical) slope grows roughly linearly with $n$, from $1.6$ at $n=4$ to $4.6$ at $n=8$. The dense limit $k=2^n$ gives $N_{\mathrm{cx}}=2(2^n-1)$ exactly (Möttönen baseline).

\begin{table}[h]
\centering
\renewcommand{\arraystretch}{1.2}
\begin{tabular}{rcccc}
\toprule
\rowcolor{headcolor}[\tabcolsep][\tabcolsep]
$n$ & best slope & median slope & worst slope & envelope $2n$ \\
\midrule
4 & $0.0$ & $1.6$ & $5.0$  & $8$ \\
5 & $0.0$ & $2.3$ & $7.0$  & $10$ \\
6 & $0.7$ & $3.0$ & $8.0$  & $12$ \\
7 & $2.0$ & $3.8$ & $10.7$ & $14$ \\
8 & $2.7$ & $4.6$ & $12.5$ & $16$ \\
\bottomrule
\end{tabular}
\caption{Empirical slope $N_{\mathrm{cx}}/k$ over uniformly random supports, for $k$ in the sparse regime $k\le n^2$. The bound $2n$ of Eq.~\eqref{eq:cnot-bound} is the structural envelope; the worst observed slope remains well within it across $n=4,\ldots,8$.}
\label{tab:cnot-slopes}
\end{table}

\paragraph*{Proof status.} The unconditional linear-in-$k$ bound $N_{\mathrm{cx}}\le 2nA_s=O(n^2k)$ holds for any qubit ordering (Proposition~\ref{prop:linear}). The tighter structure-dependent bound Eq.~\eqref{eq:cnot-bound-As}, which improves the $A_s$-prefactor from $2n$ to $2$, is proved analytically for $k=2$ after constant-bit removal. For general $k$ the exact identity $N_{\mathrm{cx}}=2\sum_{q,\beta}R_{q,\beta}$ of Appendix~\ref{app:cnot-structure} (Lemma~\ref{lem:runcount}) reduces it to the single combinatorial inequality
\begin{equation}\label{eq:runcount-conjecture}
\sum_{q,\beta} R_{q,\beta} \;\le\; E_s - 1,
\end{equation}
where $R_{q,\beta}$ is the number of maximal $1$-runs of Gray-coordinate $\beta$ over the surviving rows at level $q$, holding under the inactive-node-maximising reordering. Eq.~\eqref{eq:runcount-conjecture} is verified as in Table~\ref{tab:cnot-evidence}, but its general proof, and the question of whether a polynomial-time ordering attains it, are left to future work (\sref{sec:discussion}).

\section{Log-optimal sparse preparation with preserved metric}\label{app:sparse-routing}

The dressing invariance of Appendix~\ref{app:invariance} has a direct consequence for the two-qubit cost of \sref{sec:cnot-bound}, through the sparse-preparation strategy of Li and Luo~\cite{Li2025}. Their construction prepares a $k$-sparse target $\sum_{j}\alpha_j\ket{q_j}$ in two steps: first the $k$ amplitudes are loaded onto the lowest $\lceil\log_2 k\rceil$ qubits as a dense state $\sum_j\alpha_j\ket{j}$, with the remaining qubits left in $\ket{0}$; then a permutation $\sigma$ of the $2^n$ computational basis states relabels each loaded index to its target address, $\ket{j}\mapsto\ket{q_j}$. The circuit is thus
\begin{equation}\label{eq:separable-prep}
\ket{\psi}=U_\sigma\,\bigl(\ket{\psi_{\mathrm{core}}(\bm\theta,\bm\omega)}\otimes\ket{0}\bigr),
\end{equation}
and its two factors play cleanly separated roles. The amplitude core $\ket{\psi_{\mathrm{core}}}$ is the binary-tree ansatz of \sref{sec:polyspherical}, the same construction used throughout the paper, but on a compact $\lceil\log_2 k\rceil$-qubit register rather than the full $n$ qubits, so it carries every chart parameter $(\bm\theta,\bm\omega)$ and the closed-form diagonal metric. The routing $U_\sigma$ is a permutation matrix: it merely sends one basis state to another, is fixed once the support $\{q_j\}$ is chosen, and contains neither amplitudes nor variational parameters. It is therefore exactly the fixed unitary of Proposition~\ref{prop:invariance}, so the metric in the core coordinates $(\bm\theta,\bm\omega)$ is left unchanged by the routing and remains the diagonal metric of Eqs.~\eqref{eq:g-theta}--\eqref{eq:g-cross}. The construction thus differs from the direct pruned tree of \sref{sec:cnot-bound} only in how the same $k$-sparse state is reached: rather than loading amplitudes and placing them on their target basis states within a single $n$-qubit tree, it loads them on a compact tree and places them with a separate, parameter-free permutation.

The amplitude core costs $O(k)$ two-qubit gates. Synthesising the fixed routing $U_\sigma$ with the sparse-permutation construction of Li and Luo~\cite{Li2025} costs $O(nk/\log n)$, so the total is $O(nk/\log n)$, near-optimal for $k$-sparse preparation in general, and information-theoretically optimal when $k=\mathrm{poly}(n)$ (the regime of the sparse states prepared here), while the closed-form diagonal metric is retained. The routing was implemented in full, the matrix-reduction kernel that maps a block of $2m$ bitstrings to canonical positions, the block reduction, and the cycle decomposition into disjoint transpositions; it was verified by exhaustive basis-state checks over the reachable range, and its scaling confirmed by direct two-qubit-gate counting up to $n=2048$: the count grows sub-$O(nk)$, is linear in $k$, and tracks $nk/\log n$.

Two qualifications are essential. First, the advantage is asymptotic: the measured constant ($\approx 2\,nk$ at reachable sizes, against $\sim nk$ for the direct pruned tree of \sref{sec:cnot-bound}) leaves the direct tree cheaper until $n$ is far larger than any near-term register. Second, the construction works in the compact-tree chart and uses a single clean ancilla (the with-ancilla variant of Ref.~\cite{Li2025}; a strictly ancilla-free realisation follows from the borrowed-ancilla Toffoli of the same work). The result is therefore conceptual: the diagonal-metric framework can be pushed to the same near-optimal $O(nk/\log n)$ scaling, a guarantee that preparation-only constructions, which carry no closed-form geometry, do not provide.

\section{Measurement cost of real-time evolution: MinimalMetric versus pVQD}\label{app:shot-cost}

This appendix compares the two variational real-time schemes of \sref{sec:hubbard} by the total number of measurement shots needed to propagate a state over a fixed physical time $T$ while holding the final-state infidelity below a target $\delta$. The argument is general: it uses only the form of the two estimators and the unbiasedness of the parameter-shift gradient, and is independent of the particular Hamiltonian, lattice, or initial state.

\paragraph*{Setup.} Fix a total evolution time $T$ propagated in $N$ uniform steps of size $dt=T/N$. Both schemes advance one step from a finite-shot estimator. MinimalMetric estimates the energy gradient through parameter-shift evaluations of $\braket{H}$ and takes one explicit integration step. pVQD minimizes the one-step infidelity through an inner optimization over overlap (compute--uncompute) measurements. Write the Hamiltonian as $H=\sum_j c_j P_j$ and let $\lambda:=\sum_j|c_j|$ denote its one-norm.

\paragraph*{Per-evaluation variance.} A single-shot estimate of an energy $\braket{H}$ has variance bounded by $\lambda^2$, so resolving it to precision $\epsilon$ costs $M\sim\lambda^2/\epsilon^2$ shots. A single-shot overlap estimate is a Bernoulli outcome with variance at most $1/4$, independent of system size, so resolving it to precision $\epsilon$ costs $M\sim1/\epsilon^2$ shots. The overlap is thus the cheaper primitive per evaluation; the comparison below tracks how many evaluations, and at what precision, each scheme requires.

\paragraph*{MinimalMetric: shot noise averages down.} Write one integration step as
\begin{equation}
\bm\lambda_{k+1}=\bm\lambda_k+dt\,\bm F(\bm\lambda_k)+dt\,\bm\xi_k,
\end{equation}
where $\bm F$ is the exact velocity and $\bm\xi_k$ the shot-noise error of the estimated velocity. Because the parameter-shift gradient is unbiased, $\mathbb E[\bm\xi_k]=\bm0$ with per-component variance $\mathrm{Var}(\bm\xi_k)\sim\lambda^2/M$ for $M$ shots per energy evaluation. Independent errors across steps accumulate as a random walk,
\begin{equation}
\mathrm{Var}\bigl(\Delta\bm\lambda_{\mathrm{final}}\bigr)\;\sim\;\sum_{k=1}^{N}dt^2\,\frac{\lambda^2}{M}\;=\;N\,dt^2\,\frac{\lambda^2}{M}\;=\;T\,dt\,\frac{\lambda^2}{M},
\end{equation}
and the induced final-state infidelity is proportional, through the bounded Fubini--Study metric, to this variance. Fixing the infidelity budget $\delta$ gives the required shots per evaluation,
\begin{equation}
M_{\mathrm{MM}}\;\sim\;\frac{T\,dt\,\lambda^2}{\delta}.
\end{equation}
With $G$ parameter-shift energy evaluations per step, a fixed count set by the chart dimension and the integrator, independent of $dt$, and with the Pauli-group overhead of each evaluation already absorbed into the factor $\lambda^2$, the total shot count is
\begin{equation}\label{eq:shot-mm}
S_{\mathrm{MM}}\;=\;G\,N\,M_{\mathrm{MM}}\;\sim\;G\,\frac{T^2\lambda^2}{\delta},
\end{equation}
where the time step has cancelled ($N\,dt=T$). Refining the time resolution costs MinimalMetric nothing: more steps each carry proportionally less accumulated noise.

\paragraph*{pVQD: the per-step tolerance tightens with resolution.} To hold the cumulative infidelity below $\delta$ over $N$ steps with independent per-step errors, each step must be solved to residual infidelity $\sim\delta/N$. An overlap estimator of single-evaluation variance $1/4$ resolves a per-step infidelity floor $\sim1/M$, so
\begin{equation}
M_{\mathrm{pVQD}}\;\sim\;\frac{N}{\delta}.
\end{equation}
With $R$ overlap evaluations per step (the inner loop, fixed by the ansatz and independent of $dt$), the total is
\begin{equation}\label{eq:shot-pvqd}
S_{\mathrm{pVQD}}\;=\;R\,N\,M_{\mathrm{pVQD}}\;\sim\;R\,\frac{N^2}{\delta}\;=\;R\,\frac{T^2}{dt^2\,\delta}.
\end{equation}
The single-projector saving is real per evaluation, but the per-step tolerance $\delta/N$ makes the total grow as $1/dt^2$ as the grid is refined. Both schemes are charged here under the same independent-per-step-error assumption, and pVQD under a gradient-based inner loop (shift-rule fidelity gradients, as in the literature protocol~\cite{Barison2021}), for which unbiased gradient noise $\sim M^{-1/2}$ against the $\mathcal O(1)$ curvature of the fidelity at its maximum yields the $\sim 1/M$ floor. Both choices favour pVQD: a derivative-free inner loop resolving noisy fidelity values directly floors at the larger $\sim M^{-1/2}$, and coherently aligned per-step residuals would accumulate as $N^2$ rather than $N$. The asymmetry of the comparison below thus stems solely from the $\mathcal O(dt)$ scaling of the MinimalMetric per-step error, not from a different accumulation law.

\paragraph*{Crossover.} The ratio of the two totals is
\begin{equation}\label{eq:shot-ratio}
\frac{S_{\mathrm{MM}}}{S_{\mathrm{pVQD}}}\;\sim\;\frac{G\,\lambda^2}{R}\,dt^2.
\end{equation}
Both the physical time $T$ and the accuracy target $\delta$ cancel: the comparison is governed solely by the time step. MinimalMetric is cheaper whenever
\begin{equation}\label{eq:shot-crossover}
dt\;<\;dt^\star\;\sim\;\sqrt{\frac{R}{G\,\lambda^2}}.
\end{equation}
$S_{\mathrm{MM}}$ is independent of $dt$ while $S_{\mathrm{pVQD}}$ grows as $dt^{-2}$: refining the time resolution leaves MinimalMetric's measurement budget unchanged but inflates that of pVQD without bound. Faithful integration drives $dt$ down (the more so the smaller the target $\delta$, since the integrator error $\mathcal O(dt^p)$ must itself sit below $\delta$), so any sufficiently accurate simulation lands on the MinimalMetric-favourable side $dt<dt^\star$, where MinimalMetric reaches a given trajectory accuracy at strictly lower total measurement cost and its advantage grows as $dt\to0$. The overlap saving of pVQD dominates only at coarse resolution, where the integration itself is uncontrolled.

\section{Restricted-subspace approximate dynamics: protocol and numerics}\label{app:restricted-subspace}

This appendix records the protocol and numerics behind the incomplete-subspace experiment of \sref{sec:mol-dynamics} (the reduced-subspace curves of Fig.~\ref{fig:molecular-dynamics}); the dependence on the subspace size $k$ is collected in Fig.~\ref{fig:restricted-subspace}.

\paragraph*{Projection identity.} Let $S$ be any active-leaves set, $k=|S|$, and $\mathcal M_S$ the image of the pruned chart. Since the chart is complete on its support, $\mathcal M_S$ is the full projective space of $\mathrm{span}(S)$, and its tangent space at any point is all of $\mathrm{span}(S)$ modulo the global-phase gauge. The Dirac--Frenkel condition, that the residual $(\tfrac{d}{dt}+iH)\ket\psi$ be orthogonal to the tangent space, therefore reduces to $i\,\partial_t\ket\psi=P_SH\ket\psi$ with $\ket\psi\in\mathrm{span}(S)$ at all times, whose solution is $\ket{\psi(t)}=e^{-iH_St}\ket{\psi(0)}$ with $H_S=P_SHP_S$, up to the global phase fixed by the gauge projection of Appendix~\ref{app:tdvp-rhs}. The variational integrator of \sref{sec:kahler} on the restricted chart hence carries no method error beyond the subspace choice; every deviation from $e^{-iH_St}$ is integrator truncation. Numerically, the restricted RK4 chart trajectories agree with $e^{-iH_St}$ (evaluated by one dense eigendecomposition of $H_S$) to maximum infidelity $4.0\times10^{-15}$ over the $72$ verification runs below ($48$ at $\kappa=0.05$, $24$ at $\kappa=1$), including runs that traverse chart singularities via the initialiser of Appendix~\ref{app:singular-init}.

\paragraph*{Error certificate.} With $\psi_S(0)=P_S\psi_0/\|P_S\psi_0\|$ and initial captured weight $w_0=\|P_S\psi_0\|^2$, the distance to the exact trajectory splits as
\begin{equation}\label{eq:df-bound-full}
\bigl\|e^{-iHt}\psi_0-e^{-iH_St}\psi_S(0)\bigr\|\;\le\;\underbrace{\sqrt{2\bigl(1-\sqrt{w_0}\bigr)}}_{\text{initial projection}}\;+\;\underbrace{\int_0^t\bigl\|(1-P_S)H\,\psi_S(s)\bigr\|\,ds}_{\text{integrated leakage}},
\end{equation}
the standard Dirac--Frenkel a-priori bound applied to the linear manifold $\mathrm{span}(S)$; both terms are functionals of the restricted trajectory alone. On hardware, $\|(1-P_S)H\psi\|^2=\braket{\psi|H^2|\psi}-\sum_{j\in S}|\braket{j|H|\psi}|^2$ is measurable from the $H^2$ expectation plus one Hadamard-test overlap per (trivially prepared) basis state $\ket j$, $j\in S$. Figure~\ref{fig:restricted-subspace} (bottom row) confirms that the squared bound majorises the observed infidelity uniformly along every trajectory, within about an order of magnitude (the bound charges leaked amplitude coherently, whereas part of it rephases).

\paragraph*{Protocol.} Identical to the dipole-kick benchmark of \sref{sec:mol-dynamics} ($\kappa=0.05$, $T=1$, $200$ steps, RK4), for LiH, BeH$_2$, H$_2$O and NH$_3$ at equilibrium geometry. For each molecule a nested family of incomplete subspaces $S$ of increasing size $k$, chosen so that the truncation error stays small, is handed to the unmodified compiler, and the trajectory is reported as a function of $k$ (Fig.~\ref{fig:restricted-subspace}); the reduced-subspace curve of Fig.~\ref{fig:molecular-dynamics} is one representative $k$ per molecule. CNOT counts are those of the pruned complex circuit ($R_y$ cascade plus phase block, natural qubit ordering, no reordering search), the same compiled-circuit convention as the coherent-depth strips of Fig.~\ref{fig:molecular-dynamics}; the per-step measurement cost is set by the parameter count $2(k-1)$ through the shift rules of \sref{sec:shift}. The CNOT count is not monotone in $k$: highly structured sets (the full sector, the symmetry-reachable support) expose more Gray-code cancellations in the compiler than an arbitrary subset of nearly the same size.

\paragraph*{Results.} As the subspace grows the trajectory infidelity falls monotonically and reaches machine precision once $S$ contains the symmetry-reachable support (Fig.~\ref{fig:restricted-subspace}, top). The fixed Hamiltonian and the $x$-dipole kick obey point-group selection rules that act diagonally on the determinants, so the exact trajectory is confined to an invariant active-leaves set ($9/9$ for LiH, $18/36$ for BeH$_2$, $52/100$ for H$_2$O, $100/100$ for NH$_3$; NH$_3$'s $C_{3v}$ rules do not decouple the kick, hence full support), and restriction to it is exact to machine precision. This extends the symmetry sieve of \sref{sec:vqe} from statics to dynamics. The pruned CNOT count grows with $k$ but non-monotonically (Fig.~\ref{fig:restricted-subspace}, middle). A strong-kick variant ($\kappa=1$, same grid) spreads the trajectory across the sector, so the captured weight at fixed $k$ drops and the infidelities shift up, while remaining certified by Eq.~\eqref{eq:df-bound-full} and holding the projection identity to the same precision. Table~\ref{tab:restricted} lists representative sizes.

\begin{table*}[t]
\centering
\renewcommand{\arraystretch}{1.2}
\begin{tabular}{clccccccc}
\toprule
\rowcolor{headcolor}[\tabcolsep][\tabcolsep]  &  & $K$ & support & \multicolumn{2}{c}{$k\approx0.3K$} & \multicolumn{2}{c}{$k\approx0.5K$} & full $K$ \\
\cmidrule(lr){5-6}\cmidrule(lr){7-8}
\rowcolor{headcolor}[\tabcolsep][\tabcolsep]  &  &  &  & maximum infidelity & CNOTs & maximum infidelity & CNOTs & CNOTs \\
\midrule
\molicon{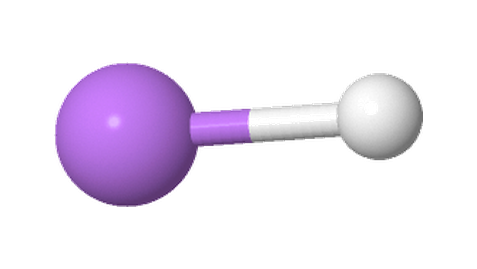} & LiH   & $9$   & $9$   & $1.8\times10^{-3}$ ($k{=}3$)  & $16$  & $1.9\times10^{-4}$ ($k{=}5$)  & $34$  & $50$ \\
\molicon{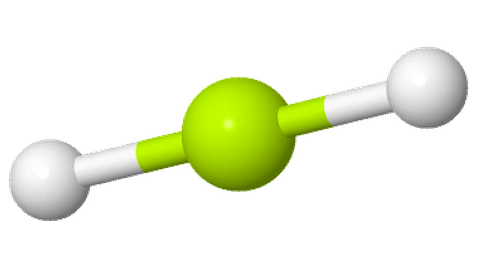} & BeH$_2$ & $36$  & $18$  & $1.0\times10^{-6}$ ($k{=}11$) & $98$ & $1.1\times10^{-15}$ ($k{=}18$) & $118$ & $144$ \\
\molicon{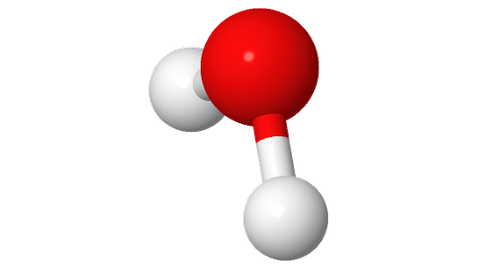} & H$_2$O  & $100$ & $52$  & $3.1\times10^{-6}$ ($k{=}30$) & $226$ & $3.6\times10^{-8}$ ($k{=}50$) & $250$ & $418$ \\
\molicon{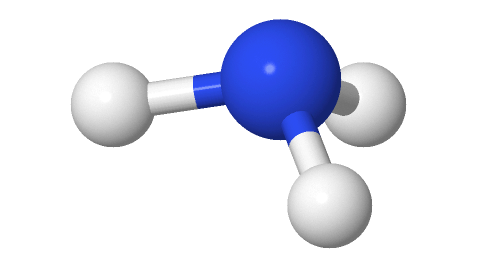} & NH$_3$  & $100$ & $100$ & $9.2\times10^{-5}$ ($k{=}30$) & $344$ & $3.0\times10^{-5}$ ($k{=}50$) & $414$ & $418$ \\
\bottomrule
\end{tabular}
\caption{Incomplete-subspace dipole-kick dynamics ($\kappa=0.05$, $T=1$): sector size $K$, symmetry-reachable support, maximum trajectory infidelity and pruned-circuit CNOT count at $k\approx0.3K$ and $k\approx0.5K$, and the full-sector CNOT count in the same counting convention (that of Fig.~\ref{fig:molecular-dynamics}). At $k=\mathrm{support}$ the restriction is exact: the BeH$_2$ entry at $k=18$ is its support, and the H$_2$O support ($k=52$) compiles to $250$ CNOTs, $60\%$ of the full sector at zero error.}
\label{tab:restricted}
\end{table*}

\begin{figure}[htbp]
\centering
\includegraphics[width=\textwidth]{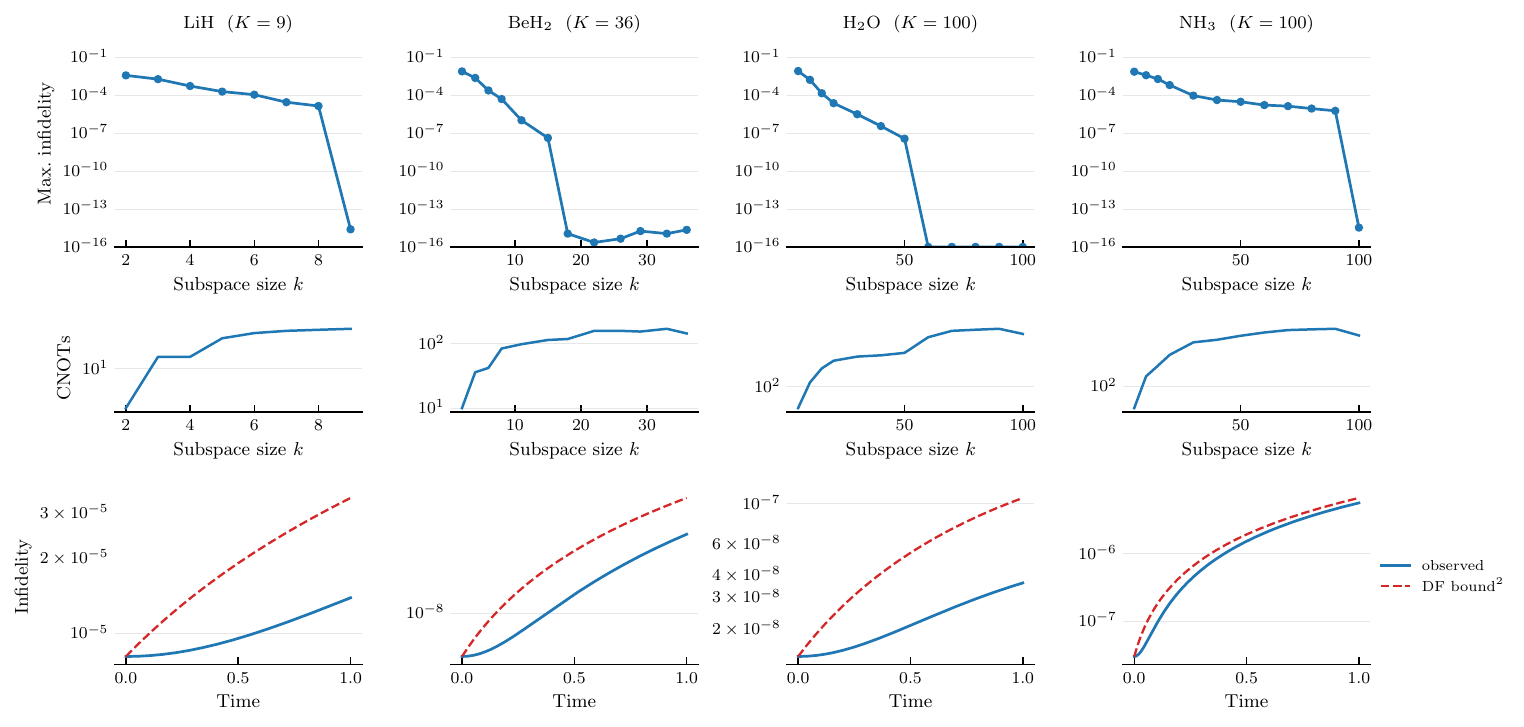}
\caption{Incomplete-subspace approximate dynamics, dependence on subspace size (companion to the reduced-subspace curves of Fig.~\ref{fig:molecular-dynamics}; $\kappa=0.05$, $T=1$; one column per molecule). \textbf{Top}: maximum trajectory infidelity against exact evolution versus subspace size $k$; the approximation improves monotonically and becomes exact once $k$ reaches the symmetry-reachable support. \textbf{Middle}: pruned complex-circuit CNOT count versus $k$ (same convention as Fig.~\ref{fig:molecular-dynamics}). \textbf{Bottom}: the squared Dirac--Frenkel certificate of Eq.~\eqref{eq:df-bound-full} (red dashed) majorises the observed infidelity (blue) along the trajectory, computed from the restricted run alone.}
\label{fig:restricted-subspace}
\end{figure}

\paragraph*{Dressed frame.} Composed with a fixed non-contractible dressing $U$, the same construction integrates the projected dynamics of $\widetilde H=U^\dagger HU$ on the $k$-dimensional model space $U\,\mathrm{span}(S)$, which is generically entangled and has no sparse classical description. The diagonal metric descends by Proposition~\ref{prop:invariance}, and the certificate Eq.~\eqref{eq:df-bound-full} stays measurable on the device ($H$ keeps its polynomial Pauli expansion, only its classical contraction through $U$ is lost). Certified approximate time evolution on a compact model space is then classically available only when the frame is, placing this composition on the same advantage footing as the block-decoupling of \sref{sec:decouple}. The bare-basis experiment here is the classically-checkable validation of exactly the subspace-truncation error that governs the dressed case.

\section{Symmetry-adapted VQE: extra resource data}\label{app:vqe-extra}

This appendix details the two-qubit gate-count progression summarised in \sref{sec:vqe} and plotted in Fig.~\ref{fig:vqe-bondcurves}(b)--(c). All entries use the same CAS active spaces and STO-3G orbitals at equilibrium geometry on a state-vector simulator, and every pruned-ansatz variant reaches the CASCI reference to better than $10^{-4}$\,Ha. Table~\ref{tab:vqe} reports the CNOT cost of the pruned ansatz as progressively more structure is exploited, the particle-and-spin sector $\{N,S_z\}$; additionally the ground-state point-group irrep (point group, the configuration plotted in Fig.~\ref{fig:vqe-bondcurves}); and the symmetry-adapted encoding (SAE) of Ref.~\cite{Picozzi2022}, which also lowers the qubit count ($n=6/6/8/10/10$ in Jordan--Wigner falls to $3/3/3/6/7$ for H$_3^+$/LiH/BeH$_2$/H$_2$O/NH$_3$), against UCCSD at both its full singles-and-doubles pool and its point-group-screened S$+$D variant.

\begin{table*}[t]
\centering
\renewcommand{\arraystretch}{1.2}
\begin{tabular}{clccccc}
\toprule
\rowcolor{headcolor}[\tabcolsep][\tabcolsep]  &  & \multicolumn{3}{c}{Pruned ansatz (this work)} & \multicolumn{2}{c}{UCCSD} \\
\cmidrule(lr){3-5}\cmidrule(lr){6-7}
\rowcolor{headcolor}[\tabcolsep][\tabcolsep]  & Molecule & $\{N,S_z\}$ & point group & symmetry-adapted & full & screened \\
\midrule
\molicon{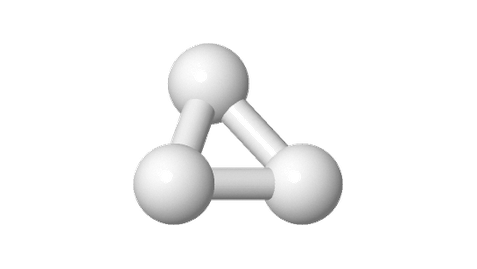} & H$_3^+$  & 32  & 24  & 6   & 280  & 136  \\
\molicon{LiH_R_1.50} & LiH  & 32  & 24  & 6   & 280  & 136  \\
\molicon{BeH2_eq} & BeH$_2$  & 84  & 48  & 6   & 1440 & 320  \\
\molicon{H2O_eq} & H$_2$O  & 224 & 116 & 62  & 3896 & 1168 \\
\molicon{NH3_eq} & NH$_3$  & 224 & 164 & 120 & 3896 & 1976 \\
\bottomrule
\end{tabular}
\caption{CNOT counts for the symmetry-adapted pruned ansatz versus UCCSD on five molecules in STO-3G at equilibrium geometry, state-vector simulation. The three pruned columns exploit progressively more structure: the particle-and-spin sector $\{N,S_z\}$; the same with the ground-state point-group irrep imposed (point group, the configuration plotted in Fig.~\ref{fig:vqe-bondcurves}); and the symmetry-adapted encoding (SAE) of Ref.~\cite{Picozzi2022}, which additionally reduces the qubit count. All three reach the CASCI reference to better than $10^{-4}$\,Ha. UCCSD is shown for the full S$+$D pool and its point-group-screened S$+$D subset; the deeper full-pool circuits are correspondingly harder to optimise and do not reach this accuracy within budget (energy error up to $\sim\!10^{-3}$\,Ha on BeH$_2$, H$_2$O and NH$_3$). In the SAE encoding the UCCSD baseline is itself shallower (LiH $32$, BeH$_2$ $16$, H$_2$O $556$, NH$_3$ $1356$ CNOTs), and the pruned-ansatz reduction persists (e.g.\ LiH $6$ versus $32$).}
\label{tab:vqe}
\end{table*}

\section{Spin adaptation via the sector Schur transform}\label{app:schur-csf}

This appendix records the construction, verification and cost accounting behind the spin-adapted VQE of \sref{sec:schur-csf}: the fixed unitary $U_{\mathrm{S}}$ whose columns, within a molecular $\{N, S_z\}$ sector, are configuration state functions (CSFs), simultaneous eigenfunctions of $N$, $S_z$ and total spin $S^2$.

\paragraph*{Construction.}
$U_{\mathrm{S}}$ is block-diagonal over occupation patterns. A determinant factorises into doubly occupied and empty orbitals (spectators for the spin coupling) and $t$ singly occupied, open-shell orbitals carrying spin-$\tfrac12$'s; within one pattern, the determinants of the sector form the fixed-$S_z$ weight block of the $t$-spin space, and the CSFs are the genealogical (sequential Clebsch--Gordan) coupled states of those $t$ spins in ascending orbital order~\cite{Paldus1974, Shavitt1977}, i.e.\ precisely the columns of the qubit Schur transform~\cite{Bacon2006} restricted to that weight block. One bookkeeping sign is essential in second quantisation: the naive spin-coupling amplitudes are exact when the determinant's creation operators are ordered orbital-by-orbital ($\alpha$ before $\beta$ within each orbital), where same-orbital spin flips carry no Jordan--Wigner string, whereas the qubit basis states of the blocked JW encoding order the modes $\alpha_0\ldots\alpha_{m-1}\beta_0\ldots\beta_{m-1}$; each CSF amplitude is therefore multiplied by the parity of the permutation relating the two orderings of the occupied modes, a per-determinant $\pm 1$.

The resulting CSFs are exactly the Gelfand--Tsetlin states of the quantum Paldus transform of Burkat and Fitzpatrick~\cite{Burkat2025}, whose canonical labels are per-orbital step vectors $\mathbf d\in\{0,1,2,3\}^d$ (empty, singly occupied spin-raising, singly occupied spin-lowering, doubly occupied); this labelling, and the gate-level realisation of the transform, are taken up below.

The column assignment fixes the leaf-indexing convention of the dressed ansatz, since the active leaves of the composed circuit $U_{\mathrm{S}}\ket{\psi(\bm\theta,\bm\omega)}$ are the preimages of the target CSFs. Two conventions are used in this appendix. In the block-wise assignment (the dense dressing of the main benchmark), within each pattern block the sector determinant strings in increasing (leaf-integer) order are matched to the CSFs ordered by total spin $S$ ascending, then by coupling path. The preimages of the lowest-spin CSFs are then the lowest determinant strings of each block, so the active-leaf set of a spin sector is a subset of the determinant leaves of its parent $\{N,S_z\}$ sector. In the canonical assignment (the gate-level Paldus realisation below), each CSF's preimage is its own Gelfand--Tsetlin step vector, written as two bits per orbital on the same register; the active leaves of a spin sector are then the valid step vectors of that $(N,S)$, a combinatorially structured (ballot-constrained) sparse set that the pruning compiler handles natively. Under both conventions closed-shell determinants are their own CSFs and their own preimages up to the Gelfand--Tsetlin phase, the reordering parity $\chi=\pm1$ of the determinant itself; a single-leaf sign is a global phase, so the Hartree--Fock corner and the singular-corner initialisation of Appendix~\ref{app:singular-init} apply unchanged. All CSF phases are fixed to the canonical Gelfand--Tsetlin convention, under which the dense unitary and the gate-level transform below agree column for column, signs included.

\paragraph*{Verification.}
For the three full $\{N,S_z\}$ sectors underlying the panel (CAS$(2e,3o)$, $k=9$; CAS$(4e,4o)$, $k=36$; CAS$(6e,5o)$, $k=100$; the point-group runs use singlet sub-blocks of these), the following were verified against exact blocked-JW matrix representations: orthogonality of $U_{\mathrm{S}}$ and simultaneous $N$/$S_z$/$S^2$ eigenstructure of every column to machine precision ($\le 10^{-15}$); spin-sector multiplicities matching the Weyl dimension formula ($6+3$, $20+15+1$, $50+45+5$ for $S=0,1,2$); exact decoupling of the singlet block of $U_{\mathrm{S}}^\dagger H U_{\mathrm{S}}$ (off-block norm $\sim 10^{-16}$) with lowest eigenvalue equal to the CASCI reference; and $\langle S^2\rangle \le 10^{-15}$ for the composed circuit at random tree parameters.

\paragraph*{Benchmark data.}
Table~\ref{tab:schur-csf} reports the panel of \sref{sec:schur-csf} through the same optimiser, stopping rule and eval currency as the convergence curves of \sref{sec:vqe} (the natural-gradient \texttt{minimize\_energy} loop, singular-corner initialised from the Hartree--Fock determinant), so the spin-adapted curve of Fig.~\ref{fig:vqe-bondcurves} is directly comparable to the point-group MinimalMetric curve beside it. Both members of each pair converge to the same energy floor; the spin-adapted runs reach it in up to $1.8\times$ fewer energy evaluations (the removed higher-spin CSFs cut the parameter count), and their $\langle S^2\rangle$ stays at the machine-zero floor throughout, whereas the point-group determinant sector passes through transient spin contamination up to $\sim 6\times10^{-6}$ and, on NH$_3$, retains $\sim 9\times10^{-7}$ at its stopping point (Fig.~\ref{fig:spin-contamination}). The spin-adaptation is exact by construction: because every retained CSF is an $S^2$ eigenstate, no parameter value can leave the singlet sector, which the determinant ansatz can only approach variationally.

\begin{figure}[h]
\centering
\includegraphics[width=0.7\textwidth]{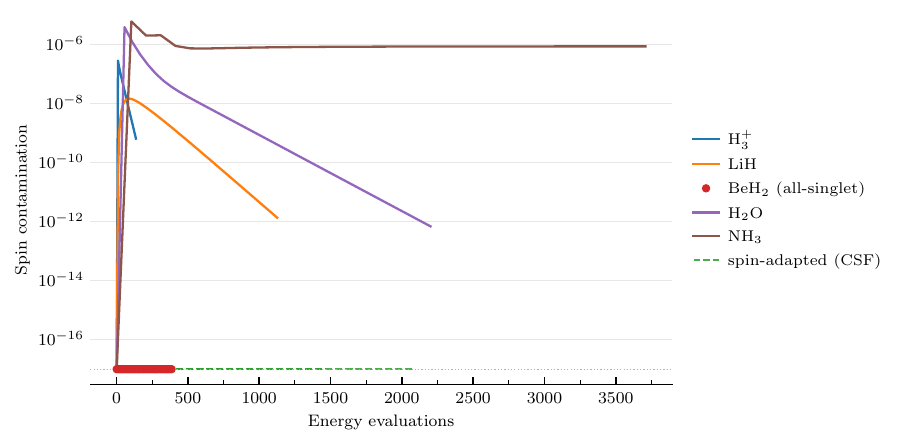}
\caption{Total-spin contamination $|\langle S^2\rangle|$ along the natural-gradient VQE trajectory (log $|\langle S^2\rangle|$ scale, linear energy-evaluation axis), point-group determinant sector versus the spin-adapted (CSF) block within the same irrep. The determinant-sector runs (solid, one colour per molecule) develop contamination up to $\sim10^{-6}$ before descending, and the NH$_3$ run retains it through convergence; the spin-adapted runs (green, dashed, all five molecules) sit at the machine-zero floor at every iterate. BeH$_2$'s $D_{2h}$ sector is already all-singlet and stays at the floor in both cases. Same runs as Fig.~\ref{fig:vqe-bondcurves}.}
\label{fig:spin-contamination}
\end{figure}

\begin{table*}[t]
\centering
\renewcommand{\arraystretch}{1.2}
\begin{tabular}{cllcccccc}
\toprule
\rowcolor{headcolor}[\tabcolsep][\tabcolsep]  & Molecule & Sector & leaves $k$ & parameters $P$ & core CNOTs & energy evaluations & $|E-E_{\mathrm{ref}}|$ & $\max\langle S^2\rangle$ \\
\midrule
\multirow{2}{*}{\molicon{H3plus_eq}} & \multirow{2}{*}{H$_3^+$} & point group  & $5$  & $4$  & $24$  & $144$  & $2\times10^{-10}$ & $3\times10^{-7}$ \\
 & & spin-adapted & $4$  & $3$  & $24$  & $105$  & $2\times10^{-11}$ & $<\!10^{-30}$ \\
\midrule
\multirow{2}{*}{\molicon{LiH_R_1.50}} & \multirow{2}{*}{LiH} & point group  & $5$  & $4$  & $24$  & $1134$ & $3\times10^{-10}$ & $1\times10^{-8}$ \\
 & & spin-adapted & $4$  & $3$  & $24$  & $938$  & $2\times10^{-10}$ & $<\!10^{-30}$ \\
\midrule
\multirow{2}{*}{\molicon{BeH2_eq}} & \multirow{2}{*}{BeH$_2$} & point group  & $6$  & $5$  & $48$  & $396$  & $2\times10^{-10}$ & $<\!10^{-30}$ \\
 & & spin-adapted & $6$  & $5$  & $48$  & $396$  & $2\times10^{-10}$ & $<\!10^{-30}$ \\
\midrule
\multirow{2}{*}{\molicon{H2O_eq}} & \multirow{2}{*}{H$_2$O} & point group  & $28$ & $27$ & $116$ & $2255$ & $2\times10^{-8}$  & $4\times10^{-6}$ \\
 & & spin-adapted & $18$ & $17$ & $126$ & $1505$ & $2\times10^{-8}$  & $1\times10^{-18}$ \\
\midrule
\multirow{2}{*}{\molicon{NH3_eq}} & \multirow{2}{*}{NH$_3$} & point group  & $52$ & $51$ & $164$ & $3811$ & $4\times10^{-6}$  & $6\times10^{-6}$ \\
 & & spin-adapted & $28$ & $27$ & $140$ & $2145$ & $4\times10^{-6}$  & $1\times10^{-19}$ \\
\bottomrule
\end{tabular}
\caption{Spin-adapted (CSF) versus determinant-sector natural-gradient VQE at equilibrium geometry, STO-3G, statevector simulation, identical pipeline and stopping rule for each pair. To match the point-group-screened MinimalMetric curves of \sref{sec:vqe}, both members of each pair are screened to the Hartree--Fock spatial irrep (JW encoding): the ``point group'' rows are the point-group determinant sector (the MM full-symmetry curve of Fig.~\ref{fig:vqe-bondcurves}), and ``spin-adapted'' the singlet CSF block within it. The spin refinement removes the higher-spin CSFs of the same irrep, halving the H$_2$O and NH$_3$ sectors ($28\!\to\!18$, $52\!\to\!28$) and reaching the same floor in up to $1.8\times$ fewer evaluations; BeH$_2$'s $D_{2h}$ sector is already all-singlet, so spin-adaptation is a no-op there. This overlap is expected: the ground-state irrep is singlet-enriched (closed-shell singlets are totally symmetric and have no triplet partner), so point-group screening already discards most triplets, and the residual spin reduction shrinks with the size of the group, recovering the full $\{N,S_z\}$ halving in the low-symmetry limit. ``Core CNOTs'' counts the pruned tree only; the fixed dressing $U_{\mathrm{S}}$ is applied exactly at the statevector level and charged separately (see the cost paragraph). $\max\langle S^2\rangle$ is the largest total-spin expectation recorded along the trajectory; CSF entries marked $<\!10^{-30}$ are structural machine zeros.}
\label{tab:schur-csf}
\end{table*}

\paragraph*{Gate-level realisation: the quantum Paldus transform.}
The quantum Paldus transform of Ref.~\cite{Burkat2025} has been implemented in full at the gate level. The transform is an isometry $U_P$ from the occupation-number basis on $2d$ interleaved JW qubits to the UGA basis $\ket{N}\ket{2S}\ket{2M}\otimes\ket{\mathbf d}$, with three $O(\log d)$-qubit registers and the data qubits left holding the step vector; the circuit is a cascade over orbitals of one Clebsch--Gordan block each, comprising register incrementers and a sequence of Givens rotations on the orbital's qubit pair, multiplexed over the $(2S,2M)$ register values, with the Clebsch--Gordan angles $\cos\theta_{S,M}=\sqrt{(S{+}M{+}\tfrac12)/(2S{+}1)}$. The implementation is verified along an end-to-end chain: against the explicit $d=1$ and $d=2$ basis tables of Ref.~\cite{Burkat2025} (their Figs.~9--10) exactly; against a reference isometry assembled from the independently verified coupling kernel above, column by column for $d\le3$ ($\sim\!10^{-13}$); and column-for-column, signs included, against the dense sector unitary $U_{\mathrm{S}}$ ($\sim\!5\times10^{-13}$). Exact (unoptimised) synthesis over $\{\mathrm{CX}, U\}$ gives $1628/3784/12294/21002$ CX on $12/14/19/21$ qubits for $d=2/3/4/5$, consistent with the $O(d^3\,\mathrm{polylog}\,d)$ scaling of Ref.~\cite{Burkat2025}, whose fault-tolerant data-lookup compilation reduces the count to $O(d^3)$ Toffolis ($5500$ at $d=50$).

Two structural facts make $U_P$ a clean dressing for the tree ansatz. First, used inside a fixed $(N,S,S_z)$ sector, the registers of the inverse transform start at the sector constants and provably return to $\ket{0}$ on any superposition of the sector's step vectors, so they disentangle exactly and the dressing invariance of Proposition~\ref{prop:invariance} applies verbatim, ancillas included. Second, the interface between the interleaved mode ordering of Ref.~\cite{Burkat2025} (in which same-orbital spin flips carry no JW string) and the blocked ordering of the Hamiltonians is a wire relabelling together with the reordering parity $(-1)^{\sum_{p<q} n_{\beta_p} n_{\alpha_q}}$, which is exactly a ladder of $d(d-1)/2$ CZ gates. The statevector benchmark of Table~\ref{tab:schur-csf} executes the sector isometry this verified circuit implements (dressing the tree with the dense $U_{\mathrm S}$), the two agreeing to simulation roundoff; the gate-level transform is the hardware realisation of the same map. On the unscreened $\{N,S\}$ sectors its canonical step-vector leaves compile to $32/32/86/196/196$ CNOTs across the panel, no worse, and at the largest sectors better, than the block-wise leaf assignment ($196$ against $208$ CNOTs for H$_2$O and NH$_3$).

\paragraph*{Cost of the dressing.}
$U_{\mathrm{S}}$ is a parameter-free, state-independent unitary, compiled once per sector: a fixed overhead shared by every energy evaluation, gradient component and optimisation step alike, with the concrete counts quoted above and the asymptotically efficient fault-tolerant compilation of Ref.~\cite{Burkat2025} available at scale (the qubit Schur transform provides the first-quantised analogue at $\mathrm{poly}(n)$ cost~\cite{Bacon2006}). The register qubits can be borrowed from JW spectator qubits, so no ancilla beyond the $2m$-qubit register is needed. Separately, the numerical evidence indicates that an exact, strictly ancilla-free polynomial synthesis is obstructed (the coupling angle depends jointly on the running spin and projection, which no single encoding on the data qubits alone makes simultaneously local), while bounded-spin and approximate variants remain polynomial, a dichotomy developed elsewhere.

\section{Gate-free spin adaptation: a CSF tree and its determinant-tree dual}\label{app:spin-tree}

Particle number, spin projection and point-group symmetry act diagonally on the computational basis, so a state of definite symmetry is supported on a combinatorially specified set of determinants and is prepared exactly by taking those determinants as the active leaves of the tree (\sref{sec:pruning}). Total spin is the exception: its eigenstates, the configuration state functions (CSFs), are entangled superpositions of determinants, so a determinant sector is generally reducible and the pruned tree reaches a definite-spin state only variationally. This appendix gives two constructions that prepare a definite-spin state exactly, with the closed-form geometry intact and with no change-of-basis gate. The first, a tree built directly in the CSF basis, has an exactly diagonal metric and gives a spin eigenstate at every parameter value; the second keeps the plain determinant tree and ties a subset of its angles to pin the state in one spin sector. The two rest on a common backbone (the same distinguished determinants, the same free nodes) and produce the identical circuit; they are dual coordinate systems on it, the CSF tree in the CSF amplitudes and the tied tree in the determinant amplitudes. We take the CSF tree as the primary construction, for its exactly diagonal metric, and the tied determinant tree as its realisation on the determinant tree itself. The construction is given for an arbitrary target spin $S$, the singlet ($S=0$) being the case used in the electronic-structure benchmarks.

\paragraph*{Total spin as a linear constraint.} Fix a molecular symmetry sector with $N$ electrons, spin projection $S_z=S$ (the highest weight of the target multiplet) and, optionally, a point-group irrep, on $n=2m$ blocked Jordan--Wigner qubits, and let $\mathcal D$ be its determinant set (the active leaves of \sref{sec:pruning}), of size $k=|\mathcal D|$. A state $\ket\psi=\sum_{d\in\mathcal D}c_d\ket d$ of projection $S_z=S$ has total spin exactly $S$ if and only if it is a highest-weight vector, annihilated by the total raising operator,
\begin{equation}\label{eq:splus}
S_+\ket\psi=0, \qquad S_+=\sum_{p=1}^m a^\dagger_{p\alpha}a_{p\beta};
\end{equation}
the lower-weight members of the multiplet follow from the fixed operator $S_-$ and need no separate ansatz. Since $S_+$ raises $S_z$ by one, Eq.~\eqref{eq:splus} is a set of linear equations in the amplitudes, one for each determinant $d'$ of the $S_z=S{+}1$ sector $\mathcal D^{+}$,
\begin{equation}\label{eq:splus-components}
\sum_{d\in\mathcal D}\bra{d'}S_+\ket d\,c_d=0, \qquad d'\in\mathcal D^{+},
\end{equation}
each coupling the (Jordan--Wigner-signed) amplitudes of the determinants reached from $d'$ by lowering one singly occupied $\alpha$ spin. There are $r=|\mathcal D^{+}|$ of them, all independent, and the spin-adapted states form a subspace $\mathcal H_S\subset\operatorname{span}(\mathcal D)$ of dimension $f_S=k-r$, the branching (Weyl) dimension of the spin-$S$ block, equal to the number of spin-$S$ CSFs of the sector.

\paragraph*{Block structure.} The constraints \eqref{eq:splus-components} decouple over occupation patterns. A determinant factorises into doubly occupied and empty orbitals, spectators for the spin coupling, and $t$ singly occupied open-shell orbitals; $S_+$ couples two determinants only if they share the occupation pattern and differ by one $\beta\!\to\!\alpha$ flip, so each constraint lives in a single occupation-pattern block and $\mathcal H_S$ is the direct sum of the spin-$S$ blocks of the $t$-spin coupling problem. A block of $t$ open shells (with $t\ge 2S$) then holds $\binom{t}{t/2-S}$ determinants and $\binom{t}{t/2-S-1}$ constraints, leaving $f_S^{(t)}=\binom{t}{t/2-S}-\binom{t}{t/2-S-1}$ spin-$S$ CSFs, and the sector totals are the sums over blocks.

\paragraph*{Ballot leaders and the backbone.} A distinguished subset of the determinants organises both constructions.

\begin{definition}[Ballot leaders]\label{def:ballot}
A determinant is a ballot leader if its open-shell spin word, read in ascending orbital order with $\alpha=+1$ and $\beta=-1$, has all partial sums nonnegative (closed-shell determinants are vacuously ballot leaders).
\end{definition}

Ballot words are the Yamanouchi words, the leading determinants of the genealogically coupled CSFs; a block of $t$ open shells has $f_S^{(t)}$ of them, so the sector has exactly $f_S$ ballot leaders, one for each spin-$S$ CSF. On the pruned determinant tree they span a subtree whose $f_S-1$ branch nodes (those with a ballot leader in each of their two subtrees) form the backbone; the remaining $r=k-f_S$ active nodes lie off it. The CSF tree parametrises the backbone as a tree over the ballot leaders; the determinant-tree ties keep the full tree, freeing the backbone nodes and tying the off-backbone ones. The split is manifestly spin-independent: only $\mathcal D$ and the ballot endpoint change with $S$.

\paragraph*{A tree in the CSF basis.} Let $U_{\mathrm S}$ be the fixed unitary that carries the sector's determinant basis to its CSF basis, a classically tabulated change of basis whose columns are the simultaneous $N$, $S_z$, $S^2$ eigenstates (the genealogically coupled CSFs), the $f_S$ spin-$S$ columns spanning $\mathcal H_S$. Build a binary tree whose $f_S$ leaves are those spin-$S$ CSFs, each placed at its ballot leader, so that the tree is the backbone. This CSF tree carries $f_S-1$ angles $\bm\theta$ and, by the path-product rule of \sref{sec:circuit}, defines a CSF amplitude vector $\mathbf a(\bm\theta)$; the determinant amplitudes are $\mathbf c=U_{\mathrm S}\,\mathbf a(\bm\theta)$, prepared on a plain determinant tree over their support (an intermediate Slater-determinant tree). Because $U_{\mathrm S}$ is a fixed unitary, the Fubini--Study metric of the mapped state is exactly the diagonal metric of Theorem~\ref{thm:diag-metric} in $\bm\theta$, unchanged (Proposition~\ref{prop:invariance}): no tie and no correction. The change of basis is used only classically, to turn $\mathbf a$ into the amplitudes the tree prepares, and never appears as a gate. Every $\bm\theta$ therefore gives a pure spin-$S$ eigenstate, so $\langle S^2\rangle=S(S+1)$ holds identically, and the natural-gradient and shift-rule machinery of \sref{sec:kahler} and \sref{sec:shift} carry over verbatim, with the diagonal inverse metric of Eq.~\eqref{eq:Minv} and the four-term shift rule (the energy is a degree-two trigonometric polynomial in each $\theta$). The pipeline is $\bm\theta\to\mathbf a\to\mathbf c=U_{\mathrm S}\mathbf a\to$ the determinant-tree circuit that prepares the state.

\paragraph*{Ties on the determinant tree.} Alternatively, keep the full pruned tree of \sref{sec:circuit} over $\mathcal D$, which carries $k-1$ active nodes whereas $\mathcal H_S$ has only $f_S-1$ free amplitude ratios; the backbone chooses which $r=k-f_S$ nodes to tie.

\begin{proposition}[Ballot selection rule]\label{prop:ballot-select}
On the pruned determinant tree over $\mathcal D$, keep the $f_S-1$ backbone nodes free and tie the remaining $r=k-f_S$ active nodes. Then for generic free angles the constraints \eqref{eq:splus-components} determine the tied angles uniquely, and the resulting state is a pure spin-$S$ eigenstate at every value of the free angles.
\end{proposition}

\noindent\textbf{Proof.} The backbone has $f_S-1$ nodes (the branch nodes of a binary subtree on the $f_S$ ballot leaders), leaving $(k-1)-(f_S-1)=k-f_S=r$ active nodes tied. The $r$ constraints \eqref{eq:splus-components} depend on the tied angles through a square Jacobian $\Phi_{\mathrm t}$, nonsingular at generic parameters, so the implicit-function theorem fixes the tied angles as smooth functions of the free ones; the amplitude vector then satisfies every constraint and lies in $\mathcal H_S$.\hfill$\square$

Nonsingularity of $\Phi_{\mathrm t}$ and $\langle S^2\rangle=S(S+1)$ have been checked to machine precision for open-shell counts up to $t=8$ and for every sector of the benchmark below. Figure~\ref{fig:spin-tree-example} shows the classification on a minimal validated example.

\begin{figure*}[!htbp]
    \centering
    \resizebox{0.74\linewidth}{!}{

\providecommand{\PWactC}{yellow!70!black}
\providecommand{\PWactF}{yellow!55}
\providecommand{\PWfixC}{red!65!black}
\providecommand{\PWfixF}{red!40}
\providecommand{\PWinaC}{gray!55}
\providecommand{\PWinaF}{gray!15}
\providecommand{\PWleafC}{green!55!black}
\providecommand{\PWleafF}{green!35}
\providecommand{\PWomegaC}{cyan!60!black}
\providecommand{\PWomegaF}{cyan!30}
\providecommand{\PWtieC}{violet!70!black}
\providecommand{\PWtieF}{violet!32}

\begin{tikzpicture}[
  grow=right,
  level distance=18mm,
  level 1/.style={sibling distance=44mm},
  level 2/.style={sibling distance=22mm},
  level 3/.style={sibling distance=11mm},
  level 4/.style={sibling distance=5.5mm},
  level 5/.style={sibling distance=5.5mm, level distance=16mm},
  int/.style={circle, draw, inner sep=0pt, font=\scriptsize, minimum size=5.2mm, align=center},
  ph/.style ={circle, draw=\PWleafC, fill=green!8, inner sep=0pt, font=\tiny, minimum size=4.2mm},
  phA/.style={circle, draw=\PWomegaC, fill=\PWomegaF, inner sep=0pt, font=\tiny, minimum size=4.2mm},
  phI/.style={circle, draw=\PWinaC, fill=\PWinaF, inner sep=0pt, font=\tiny, minimum size=4.2mm},
  lf/.style ={rectangle, draw, inner sep=2pt, font=\tiny},
  lfS/.style={rectangle, draw=\PWleafC, fill=\PWleafF, inner sep=2pt, font=\tiny},
  lfB/.style={rectangle, draw=\PWleafC, line width=0.9pt, fill=\PWleafF, inner sep=2pt, font=\tiny},
  edge from parent/.style={draw, line width=0.4pt},
  el/.style={midway, font=\tiny, sloped, above, inner sep=1pt},
]

\node[int, fill=\PWactF, draw=\PWactC] {$\theta_0$}
  child { node[int, fill=\PWfixF, draw=\PWfixC] {$\theta_1$}
    child { node[int, fill=\PWinaF, draw=\PWinaC] {$\theta_3$}
      child { node[int, fill=\PWinaF, draw=\PWinaC] {$\theta_7$}
        child { node[phI] {$\omega_0$} child { node[lf] {$0000$} } }
        child { node[phI] {$\omega_1$} child { node[lf] {$0001$} } }
      }
      child { node[int, fill=\PWinaF, draw=\PWinaC] {$\theta_8$}
        child { node[phI] {$\omega_2$} child { node[lf] {$0010$} } }
        child { node[phI] {$\omega_3$} child { node[lf] {$0011$} } }
      }
      edge from parent node[el]{$0$}
    }
    child { node[int, fill=\PWtieF, draw=\PWtieC] {$\theta_4$}
      child { node[int, fill=\PWfixF, draw=\PWfixC] {$\theta_9$}
        child { node[phI] {$\omega_4$} child { node[lf] {$0100$} } edge from parent node[el]{$0$} }
        child { node[lfB] {$0101$} edge from parent node[el]{$1$} }   
        edge from parent node[el]{$\cos\theta_4$}
      }
      child { node[int, fill=\PWfixF, draw=\PWfixC] {$\theta_{10}$}
        child { node[phA] {$\omega_6$} child { node[lfS] {$0110$} edge from parent node[el]{$e^{i\omega_6}$} } edge from parent node[el]{$1$} }
        child { node[phI] {$\omega_7$} child { node[lf] {$0111$} } edge from parent node[el]{$0$} }
        edge from parent node[el]{$\sin\theta_4$}
      }
      edge from parent node[el]{$1$}
    }
    edge from parent node[el]{$\cos\theta_0$}
  }
  child { node[int, fill=\PWfixF, draw=\PWfixC] {$\theta_2$}
    child { node[int, fill=\PWactF, draw=\PWactC] {$\theta_5$}
      child { node[int, fill=\PWfixF, draw=\PWfixC] {$\theta_{11}$}
        child { node[phI] {$\omega_8$} child { node[lf] {$1000$} } edge from parent node[el]{$0$} }
        child { node[phA] {$\omega_9$} child { node[lfB] {$1001$} edge from parent node[el]{$e^{i\omega_9}$} } edge from parent node[el]{$1$} }
        edge from parent node[el]{$\cos\theta_5$}
      }
      child { node[int, fill=\PWfixF, draw=\PWfixC] {$\theta_{12}$}
        child { node[phA] {$\omega_{10}$} child { node[lfB] {$1010$} edge from parent node[el]{$e^{i\omega_{10}}$} } edge from parent node[el]{$1$} }
        child { node[phI] {$\omega_{11}$} child { node[lf] {$1011$} } edge from parent node[el]{$0$} }
        edge from parent node[el]{$\sin\theta_5$}
      }
      edge from parent node[el]{$1$}
    }
    child { node[int, fill=\PWinaF, draw=\PWinaC] {$\theta_6$}
      child { node[int, fill=\PWinaF, draw=\PWinaC] {$\theta_{13}$}
        child { node[phI] {$\omega_{12}$} child { node[lf] {$1100$} } }
        child { node[phI] {$\omega_{13}$} child { node[lf] {$1101$} } }
      }
      child { node[int, fill=\PWinaF, draw=\PWinaC] {$\theta_{14}$}
        child { node[phI] {$\omega_{14}$} child { node[lf] {$1110$} } }
        child { node[phI] {$\omega_{15}$} child { node[lf] {$1111$} } }
      }
      edge from parent node[el]{$0$}
    }
    edge from parent node[el]{$\sin\theta_0$}
  };

\end{tikzpicture}}
    \caption{Spin-adapted classification tree, in the style of Fig.~\ref{fig:classified-tree}(a), on the smallest sector that ties a node: two spatial orbitals, two electrons, the singlet sector of the $\{N=2,\,S_z=0\}$ block on $n=4$ qubits (H$_2$ in STO-3G; the same single-tie, $r=1$ structure recurs in the H$_3^+$ and LiH sectors of Table~\ref{tab:spin-tree}). The four active leaves (green) are the sector determinants $\mathcal D=\{0101,0110,1001,1010\}$: the closed shells $\ket{\phi_0\bar\phi_0}$ ($0101$) and $\ket{\phi_1\bar\phi_1}$ ($1010$), and the open-shell pair $\ket{\phi_0\bar\phi_1}$ ($1001$) and $\ket{\bar\phi_0\phi_1}$ ($0110$), so $k=4$ with $f_S=3$ singlet CSFs and $r=k-f_S=1$ constraint. Internal nodes carry the classification of \sref{sec:classify}: \textsc{Active} (yellow, $\theta_i$ free), \textsc{Fixed} (red, $\theta_i$ pinned to $0$ or $\pi/2$, edge factors shown as $0/1$), \textsc{Inactive} (grey, gauge direction). Spin adaptation refines the active nodes by the ballot rule (Prop.~\ref{prop:ballot-select}): the three ballot leaders (heavier green border) span the backbone, whose $f_S-1=2$ branch nodes $\theta_0,\theta_5$ stay free (yellow), while the one off-backbone active node $\theta_4$ becomes \textsc{Tied} (violet). The single singlet constraint equates the two open-shell amplitudes, $c_{0110}=c_{1001}$, that is $\cos\theta_0\sin\theta_4=\sin\theta_0\cos\theta_5$, and so fixes the tied angle in closed form as a function of the free ones, $\theta_4=\arcsin(\tan\theta_0\cos\theta_5)$. The coupled determinants $0110$ and $1001$ are bitwise complements (Hamming distance $4$), so no qubit relabelling makes them siblings: the tie spreads across the three angles $\theta_0,\theta_4,\theta_5$ rather than pinning one node to a constant. Every value of the two free angles yields a pure singlet, $\langle S^2\rangle=0$ identically (verified to machine zero), and natural-gradient optimisation on the same determinant-tree circuit reaches the sector full-configuration-interaction energy (H$_2$/STO-3G, $\lvert E-E_{\mathrm{FCI}}\rvert=4\times10^{-16}$). Leaf phases $\omega_j$ (cyan, active; grey, inactive) behave exactly as in the bare tree, with the global phase gauge-fixed on the Hartree--Fock leaf $0101$ (direct edge).}
    \label{fig:spin-tree-example}
\end{figure*}
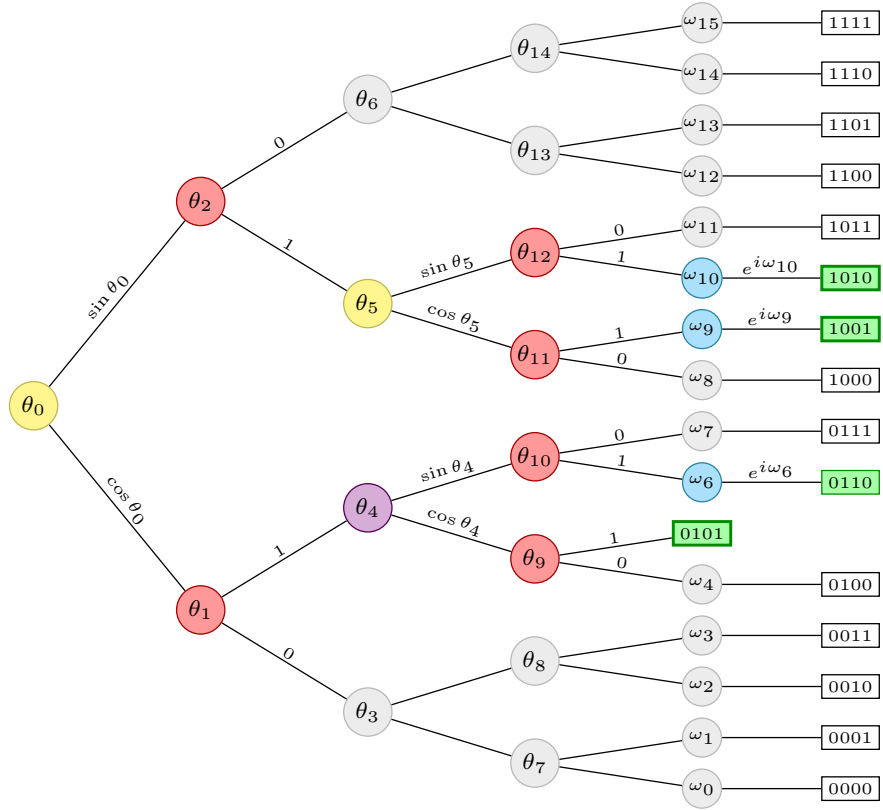

\paragraph*{Closed-form ties.} Fixing the tied angles from the free ones is closed form. Collecting the factor at a tied node $j$ in the path products makes its constraint $\sum_d\varphi_d c_d=0$ a single sinusoid in $\theta_j$,
\begin{equation}\label{eq:node-tie}
A\cos\theta_j+B\sin\theta_j=C,
\end{equation}
with $A,B$ the partial sums of the constraint over $j$'s two subtrees and $C$ the part carried by any of its determinants outside them, all functions of the remaining angles alone; one sinusoid set equal to a constant, it is solved for $\theta_j$ in closed form. When $j$ separates the support the constant vanishes and $\theta_j=\operatorname{atan2}(-A,B)$ is a pure rotation; a $t=2$ open pair whose two determinants are Jordan--Wigner far apart (every $r=1$ sector here) is separated only at a backbone node, so the ballot rule instead ties an off-backbone node reached by one of the pair, keeping $C\neq0$ (there $A=0$ and $\theta_j=\arcsin(C/B)$). A block with several coupled constraints is triangularised by eliminating a shared separating node between two of them through their resultant $A_1B_2-A_2B_1$, again of the form~\eqref{eq:node-tie} in the next node down, so the block's tied angles follow in sequence. The Clebsch--Gordan coefficients of the coupling enter only through the fixed weights $\varphi_d$.

\paragraph*{Induced metric.} Because every tied angle is a function of the free ones, the prepared state is a smooth map of the free parameters, and the metric it induces on them follows from the diagonal metric of Theorem~\ref{thm:diag-metric} by a low-rank correction. With $\bm\theta^{\mathrm f}$ and $\bm\theta^{\mathrm t}(\bm\theta^{\mathrm f})$ the free and tied angles, $J=\partial\bm\theta^{\mathrm t}/\partial\bm\theta^{\mathrm f}$ the tie Jacobian, and $D_{\mathrm f},D_{\mathrm t}$ the diagonal weights $w_i$ of the free and tied nodes, the chain rule gives $g_{\mathrm f}=D_{\mathrm f}+J^{\top}D_{\mathrm t}J$, a diagonal matrix plus a term of rank at most $r$, with Woodbury inverse
\begin{equation}\label{eq:induced-inv}
g_{\mathrm f}^{-1}=D_{\mathrm f}^{-1}-D_{\mathrm f}^{-1}J^{\top}\!\bigl(D_{\mathrm t}^{-1}+J\,D_{\mathrm f}^{-1}J^{\top}\bigr)^{-1}\!J\,D_{\mathrm f}^{-1},
\end{equation}
whose only inversion is the $r\times r$ matrix in parentheses; the tie Jacobian is $J=-\Phi_{\mathrm t}^{-1}\Phi_{\mathrm f}$, with $\Phi_{\mathrm f}$ the derivative of the constraints in the free angles. The natural-gradient and time-evolution updates of \sref{sec:kahler} therefore carry over verbatim, with the diagonal inverse metric of Eq.~\eqref{eq:Minv} replaced by Eq.~\eqref{eq:induced-inv} and the shift-rule energy gradient of \sref{sec:shift} projected onto the free angles by the same chain rule. Both additions are $r\times r$ classical linear algebra; the circuit and its shift-rule oracle are untouched.

\paragraph*{The two constructions are dual.} They share the backbone: the CSF tree's nodes are exactly the free nodes of the determinant-tree ties (the same heap positions), and the two produce the identical determinant-tree circuit, differing only in coordinates. The CSF tree varies the CSF amplitudes $\mathbf a$, whose metric is exactly diagonal; the tied tree varies the determinant amplitudes $\mathbf c$, whose metric is that diagonal plus the rank-$r$ tie correction \eqref{eq:induced-inv}. On the backbone the two are linked by the triangular Yamanouchi matrix $W=U_{\mathrm S}[\text{ballot leaders},\,\text{spin-}S\text{ columns}]$, the determinant amplitudes at the ballot leaders being $W\mathbf a$; the tied angles the CSF tree realises coincide with the closed-form ties to machine precision. The CSF tree carries the simpler geometry, a strictly diagonal metric and no ties, and reaches the spin eigenstate cleanly where the bare determinant tree stalls (NH$_3$ below), at the cost of tabulating $U_{\mathrm S}$; the tie construction needs no such tabulation and works in the determinant amplitudes directly.

\paragraph*{Cost and validation.} Both constructions prepare the pruned determinant tree of \sref{sec:pruning} for the sector $\mathcal D$: the two-qubit cost is that of the bare sector tree, no change-of-basis unitary is appended and no auxiliary register is needed, and spin adaptation is paid entirely in classical work, all polynomial in $k$: marking the ballot leaders and the backbone, and then, for the CSF tree, tabulating the fixed $U_{\mathrm S}$ and the diagonal metric, or, for the tied determinant tree, solving the closed-form ties \eqref{eq:node-tie} and forming the rank-$r$ metric correction \eqref{eq:induced-inv}. Table~\ref{tab:spin-tree} reports both on the benchmark molecules. Every state of either family lies in $\mathcal H_S$ by construction, so $\langle S^2\rangle=S(S+1)$ holds identically, and natural-gradient optimisation reaches the sector configuration-interaction reference using only the determinant-tree two-qubit gates. The construction is spin-general: for a triplet target ($S=1$) in the BeH$_2$ and H$_2$O active spaces, the CSF tree reaches the lowest spin-1 energy with $\langle S^2\rangle=2$ to within $10^{-15}$. It is also encoding-native: a symmetry-adapted encoding acts on a fixed sector as a sign-free $\mathrm{GF}(2)$ relabelling of the determinants, so mapping each Slater determinant of the intermediate tree to its encoded correspondent (the only sign being the blocked-to-interleaved reordering parity already used above) lands the identical spin-$S$ preparation in any symmetry-adapted encoding, at the encoded qubit count and with $\langle S^2\rangle=S(S+1)$ preserved -- the form used for the spin-conserving dressed two-block ansatz of Appendix~\ref{app:two-block}.

\begin{table*}[t]
\centering
\renewcommand{\arraystretch}{1.2}
\begin{tabular}{clcccccc}
\toprule
\rowcolor{headcolor}[\tabcolsep][\tabcolsep] &  & $k$ & $f_S$ & free/tied & core CNOTs & $|E-E_{\mathrm{ref}}|$ & $\langle S^2\rangle$ \\
\midrule
\molicon{H3plus_eq} & H$_3^+$ & $5$  & $4$  & $3$ / $1$   & $24$  & $8\times10^{-11}$ & $\le10^{-25}$ \\
\molicon{LiH_R_1.50} & LiH & $5$  & $4$  & $3$ / $1$   & $24$  & $1\times10^{-9}$  & $\le10^{-25}$ \\
\molicon{BeH2_eq} & BeH$_2$ & $6$  & $6$  & $5$ / $0$   & $48$  & $1\times10^{-10}$ & $0$ \\
\molicon{H2O_eq} & H$_2$O  & $28$ & $18$ & $17$ / $10$ & $116$ & $2\times10^{-8}$  & $\le10^{-18}$ \\
\molicon{NH3_eq} & NH$_3$  & $52$ & $28$ & $27$ / $24$ & $164$ & $5\times10^{-6}$  & $\le10^{-19}$ \\
\bottomrule
\end{tabular}
\caption{Gate-free spin-adapted VQE at equilibrium geometry, STO-3G, point-group-screened singlet sector, statevector simulation, for the CSF tree; the determinant-tree ties produce the identical circuit and match these figures to the displayed precision except on NH$_3$ (see below). $k$ is the sector determinant count and $f_S$ its spin-adapted (CSF) dimension; ``free/tied'' are the numbers of free and tied tree angles ($f_S-1$ and $k-f_S$), equal to the CSF tree's angles and to the backbone and off-backbone nodes of the tie construction; ``core CNOTs'' is the two-qubit count of the pruned determinant tree, the entire two-qubit cost of either construction. The energy error is against the sector configuration-interaction reference; $\langle S^2\rangle$ equals $S(S+1)=0$ by construction, and the entry is the largest deviation recorded along the optimisation. The CSF tree holds $\langle S^2\rangle$ at machine zero throughout, via its exactly diagonal metric; on NH$_3$ the determinant-tree ties instead reach $|E-E_{\mathrm{ref}}|=2\times10^{-7}$ but pass through a residual $\langle S^2\rangle\le10^{-7}$ near a chart corner that the CSF tree avoids. BeH$_2$'s $D_{2h}$ sector is already all-singlet, so no node is tied and the two constructions coincide.}
\label{tab:spin-tree}
\end{table*}

\section{Classical simulability of the bare ansatz, and its boundary}\label{app:simulability}

The bare ansatz of \sref{sec:circuit}, the tree state $\ket{\psi(\bm\theta,\bm\omega)}$ and the energy $\braket{\psi|H|\psi}$, optimised or propagated by the diagonal-metric updates of \sref{sec:kahler}, is a classical algorithm whenever the active-leaf support is small. This is made precise here, because it fixes what the benchmarks of \sref{sec:vqe} and \sref{sec:hubbard} do and do not claim, and it isolates exactly where a quantum device becomes indispensable: in the composition $U(\bm\phi)\ket{\psi(\bm\theta,\bm\omega)}$ with a non-classically-contractible dressing (Appendices~\ref{app:two-block}--\ref{app:fs-protocols}). The result is stated in a form deliberately broader than the tree ansatz, because the phenomenon is structural rather than an artefact of this construction: it applies verbatim to any variational family that shares the tree's two enabling features, an enumerable sparse support and classically computable amplitudes, including several constructions in the literature. The tree ansatz then enters as the instantiating example, and the closure properties of the class fix the necessary conditions any dressing must meet to restore genuinely quantum content.

\paragraph*{The class.}

Throughout, a frame is an orthonormal basis $\{V\ket{j}\}$ of the $n$-qubit Hilbert space specified by a classically describable unitary $V$: the identity (computational basis), a wire relabelling, or a classically tabulated basis change such as the sector Schur transform of Appendix~\ref{app:schur-csf}. All conditions below are stated in a fixed frame; they are written for $V=1$, and every statement is covariant under a change of frame, with $V$ absorbed into the observable.

\begin{definition}[Sector-sparse variational scheme]\label{def:sector-sparse}
A parametrised family of $n$-qubit states $\{\ket{\psi(\bm\lambda)}\}$ with $P$ real parameters, together with an observable $H$, is $(k,V)$-sector-sparse if, in the frame $V$:
\begin{itemize}
\item[(A1)] \textbf{Enumerable support.} Every reachable state is normalised, $\ket{\psi(\bm\lambda)}=\sum_{j\in S}c_j(\bm\lambda)\ket{j}$ with $\sum_{j\in S}|c_j|^2=1$, for a single explicitly enumerated label set $S$, $|S|=k$ (stored sorted, so that membership costs $\mathcal O(n+\log k)$); if instead the chart delivers unnormalised $c$, every quantity below acquires the classically computable factor $c^\dagger c$ and its derivatives, and the conclusions are unchanged.
\item[(A2)] \textbf{Amplitude and derivative oracles.} Each amplitude $c_j(\bm\lambda)$ and each derivative $\partial_\mu c_j(\bm\lambda)$ is computable in $\mathrm{poly}(n)$ time. (The derivative oracle is not implied by the amplitude oracle for an arbitrary parametrisation, but holds automatically whenever each parameter enters through a gate with an efficiently computable generator; in the tree chart both are $\mathcal O(n)$ path products by Eq.~\eqref{eq:c2p-ry}.)
\item[(A3)] \textbf{Frame-sparse observable.} $H=\sum_{l=1}^{M}h_l\,P_l$, where each term maps frame-basis states to frame-basis states up to a phase, $P_l\ket{j}=\eta_{l,j}\ket{\pi_l(j)}$ with $\pi_l$ and $\eta_{l,j}$ computable in $\mathrm{poly}(n)$. Pauli strings are the canonical case (for the molecular and Hubbard Hamiltonians in a fixed encoding, $M=\mathrm{poly}(n)$); every statement below extends verbatim to row-computable $s$-sparse terms at cost $\times\, s$.
\end{itemize}
\end{definition}

Write $c=(c_j)_{j\in S}\in\mathbb C^k$ and $\partial_\mu c$ for the coefficient vector and its derivatives, both classically available in $\mathcal O(k\,\mathrm{poly}(n))$ by (A2), and $P_S$ for the orthogonal projector onto $\mathrm{span}(S)$. A single lemma carries all of the $n$-qubit content; everything that follows is exact linear algebra in dimension $k$.

\begin{lemma}[Compression]\label{lem:compression}
Under (A1) and (A3), the compressed observable $H_S:=P_SHP_S$, as a $k\times k$ Hermitian matrix indexed by $S$, is computable exactly in $\mathcal O\big(Mk\,(n+\log k)\big)$ time, with at most $\min(k^2,Mk)$ nonzero entries; explicitly,
\begin{equation}\label{eq:energy-classical}
\braket{\psi|H|\psi}\;=\;c^\dagger H_S\,c,
\qquad
\big(H_S\big)_{j\,j'}\;=\;\sum_{l=1}^{M}h_l\,\eta_{l,j'}\,\big[\pi_l(j')=j\big].
\end{equation}
The same holds for any other frame-sparse observable, and for $H^2$ at cost $\mathcal O\big(M^2k\,(n+\log k)\big)$.
\end{lemma}

\textbf{Proof.} For each $l\in[M]$ and $j'\in S$, compute $\pi_l(j')$ and $\eta_{l,j'}$, test $\pi_l(j')\in S$ by binary search, and, if present, add $h_l\,\eta_{l,j'}$ to the entry $(H_S)_{\pi_l(j'),\,j'}$. Each pair $(l,j')$ touches exactly one entry, so the loop has $Mk$ iterations; Hermiticity is inherited from $H$, and the first identity in Eq.~\eqref{eq:energy-classical} is $\ket\psi=P_S\ket\psi$. For $H^2$, expand into the $M^2$ frame-sparse monomials $P_lP_{l'}$. $\square$

\begin{theorem}[Exact dequantization of the metric-aware loop]\label{thm:dequantization}
For any $(k,V)$-sector-sparse scheme, the following are computable classically, exactly (to machine precision, with no sampling and no additive-error estimation), in time polynomial in $n$, $k$, $M$ and $P$:
(i) the energy $E(\bm\lambda)=\braket{\psi|H|\psi}$;
(ii) the Euclidean gradient $dE$;
(iii) the chart pullback metric $\bar g_{\mu\nu}=\mathrm{Re}\braket{\partial_\mu\psi|\partial_\nu\psi}$ of Theorem~\ref{thm:diag-metric}, the full quantum geometric tensor $G_{\mu\nu}=\braket{\partial_\mu\psi|\partial_\nu\psi}-\braket{\partial_\mu\psi|\psi}\!\braket{\psi|\partial_\nu\psi}$, and hence the projective Fubini--Study metric $\mathrm{Re}\,G$ and the Berry curvature $-2\,\mathrm{Im}\,G$;
(iv) one step of quantum natural gradient, variational imaginary-time, or variational real-time (McLachlan) evolution;
(v) exact sampling of frame-basis measurements of $\ket{\psi(\bm\lambda)}$;
(vi) exact sampling from the sector-Haar ensemble on $\mathrm{span}(S)$, and closed-form design moments of frame-sparse observables over it;
(vii) the Dirac--Frenkel leakage certificate $\|(1-P_S)H\psi\|$ of Eq.~\eqref{eq:df-bound-full}, including its per-leaf decomposition.
\end{theorem}

\textbf{Proof.} Compress once by Lemma~\ref{lem:compression}. \textbf{(i)} Eq.~\eqref{eq:energy-classical}: one sparse matrix--vector product, $\mathcal O(\min(k^2,Mk))$. \textbf{(ii)} $\partial_\mu E=2\,\mathrm{Re}\big[(\partial_\mu c)^\dagger H_S\,c\big]$: one inner product per parameter, $\mathcal O(Pk)$ after the single matrix--vector product. (No shift rules are needed classically; the shift-rule identities of \sref{sec:shift} are statements about this same trigonometric polynomial, so the two evaluations agree identically.) \textbf{(iii)} In coefficients the chart pullback is $\bar g_{\mu\nu}=\mathrm{Re}\big[(\partial_\mu c)^\dagger(\partial_\nu c)\big]$ and the projective tensor adds the rank-one term, $G_{\mu\nu}=(\partial_\mu c)^\dagger(\partial_\nu c)-\big[(\partial_\mu c)^\dagger c\big]\big[c^\dagger(\partial_\nu c)\big]$: an $\mathcal O(k)$ contraction per entry, $\mathcal O(P^2k)$ in total for either. Theorem~\ref{thm:diag-metric} is the special case in which the tree chart renders the chart pullback $\bar g$ diagonal (the two metrics differing only by the rank-one global-phase correction on the leaf-phase block, cf.\ \sref{sec:metric}), with the $\mathcal O(P^2 k)$ assembly collapsing to the $\mathcal O(k)$ path products. \textbf{(iv)} Every ingredient of Eqs.~\eqref{eq:vite-flow} and \eqref{eq:tdvp-flow}, or equivalently of the McLachlan normal equations, whose matrix is the chart pullback $\bar g$ (equivalently $\mathrm{Re}\,G$ after the global-phase projection) and whose right-hand side has components $\mathrm{Im}\big[(\partial_\mu c)^\dagger H_S\,c\big]$ up to the sign convention of Ref.~\cite{Yuan2019}, is supplied by (i)--(iii); the linear solve is $\mathcal O(P^3)$ classical work, with the global-phase gauge handled by the same projection as in Appendix~\ref{app:tdvp-rhs} and a pseudoinverse where gauge directions render $\bar g$ singular, and collapses to $\mathcal O(P)$ when $\bar g$ is diagonal. \textbf{(v)} The outcome distribution is supported on $S$ with probabilities $|c_j|^2$: enumerate and sample exactly in $\mathcal O(k)$. \textbf{(vi)} Draw $z\sim\mathcal{CN}(0,1)^{k}$ and set $c=z/\lVert z\rVert$: by unitary invariance of the isotropic complex Gaussian this is the Haar ensemble on the sector projective space $\mathbb{C}\mathbf P^{k-1}$, at $\mathcal O(k)$ per sample; Proposition~\ref{prop:fs-haar} is the chart-level refinement identifying this law with the Fubini--Study volume of the tree chart. Ensemble moments follow in closed form from the symmetric-subspace identity $\mathbb E\big[(\ket\psi\!\bra\psi)^{\otimes t}\big]=\Pi^{(t)}_{\mathrm{sym}}\big/\binom{k+t-1}{t}$~\cite{Harrow2013, Dankert2009}: in particular $\mathbb E\braket{\psi|A|\psi}=\operatorname{Tr}(A_S)/k$ and $\mathbb E\big[\braket{\psi|A|\psi}\braket{\psi|B|\psi}\big]=\big[\operatorname{Tr}(A_S)\operatorname{Tr}(B_S)+\operatorname{Tr}(A_SB_S)\big]/[k(k+1)]$, with the compressed $A_S,B_S$ from Lemma~\ref{lem:compression}. \textbf{(vii)} $H\ket\psi$ is supported on $\bigcup_l\pi_l(S)$, at most $Mk$ labels whose amplitudes are explicitly computable (accumulate $h_l\,\eta_{l,j'}\,c_{j'}$ at label $\pi_l(j')$ in a dictionary), so $\|(1-P_S)H\psi\|^2$ is the exact sum of $|\cdot|^2$ over the labels outside $S$, in $\mathcal O\big(Mk(n+\log Mk)\big)$, per-leaf terms included. $\blacksquare$

\medskip

The tree ansatz satisfies Definition~\ref{def:sector-sparse} with $V$ the identity (or, for the spin-adapted loop of \sref{sec:schur-csf}, the sector Schur transform), amplitude and derivative oracles given by the $\mathcal O(n)$ path products of Eq.~\eqref{eq:c2p-ry}, and $P=\mathcal O(k)$ parameters. Theorem~\ref{thm:dequantization} therefore recovers, and slightly sharpens, the statement used in the main text:

\begin{proposition}\label{prop:simulability}
Ground-state search and real- or imaginary-time variational evolution of the bare ansatz $\ket{\psi(\bm\theta,\bm\omega)}$ supported on $k$ active leaves, under a Hamiltonian of $M$ Pauli terms, are classically simulable in $\mathcal O(Mk^2)$ time per optimiser/integrator step (one gradient sweep of $\mathcal O(k)$ energy evaluations, each $\mathcal O(Mk)$); the compressed evaluation of Theorem~\ref{thm:dequantization} lowers this to $\mathcal O(Mk+k^2)$ per step and extends it to measurement sampling, sector-Haar ensemble averages, and the leakage certificate.
\end{proposition}

Consequently the symmetry-adapted VQE of \sref{sec:vqe} and the variational dynamics of \sref{sec:mol-dynamics} and \sref{sec:hubbard} are, as run on these small sectors, classical computations. Their content is therefore not a quantum--classical separation but (a) that the quantum realisation of the same state is a circuit of strictly lower two-qubit depth than the UCCSD/HVA/pVQD alternatives, and (b) that the exactly diagonal geometry makes the optimisation and the time integration well-conditioned by construction. These are statements about the primitive; the separation question is delimited by the closure and boundary statements below.

\paragraph*{Closure under contractible dressings.}

The class of Definition~\ref{def:sector-sparse} is moreover closed under composition with a broad family of fixed dressings, so that entanglement, or even non-Clifford content, in a dressing does not by itself restore quantum hardness. Call a fixed unitary $U$ contractible (relative to the scheme) if each conjugated term $U^\dagger P_lU$ admits a classically computable re-expansion into $\mathrm{poly}(n)$ frame-sparse monomials, or if the matrix elements $\braket{j|U^\dagger HU|j'}$ on $S\times S$ are otherwise classically computable in polynomial time.

\begin{proposition}[Closure]\label{prop:closure}
If $\big(\{\ket{\psi(\bm\lambda)}\},H\big)$ is $(k,V)$-sector-sparse and $U$ is contractible, the dressed loop, Theorem~\ref{thm:dequantization} applied to $\ket\Psi=U\ket{\psi(\bm\lambda)}$ with observable $H$, equivalently to $\ket\psi$ with $\widetilde H=U^\dagger HU$, remains exactly classically computable in polynomial time. Contractible dressings include:
(a) any polynomial-size Clifford circuit, for which $U^\dagger P_lU$ is a single Pauli string computable by stabiliser tableau in $\mathcal O(n^2)$~\cite{Gottesman1997, Aaronson2004}, so $M'=M$;
(b) any matchgate (fermionic-Gaussian) circuit, provided each term of $H$ has Majorana degree at most a constant $d$: writing $P_l$ as a Majorana monomial and $U^\dagger\gamma_\mu U=\sum_\nu R_{\mu\nu}\gamma_\nu$ with $R\in\mathrm O(2n)$ classically computable~\cite{Valiant2002, TerhalDiVincenzo2002}, the conjugate expands into at most $(2n)^d$ Majorana monomials, each a single Pauli string under Jordan--Wigner, so $M'\le M(2n)^d$ (molecular and Hubbard Hamiltonians have $d\le4$);
(c) any constant-depth geometrically local circuit acting on terms of constant weight: $U^\dagger P_lU$ is supported on the light cone of $\mathrm{supp}(P_l)$, of constant size, and is computed by exact dense conjugation on that region;
(d) any circuit with a polynomial-bond-dimension tensor-network representation (matrix-product operator of bond $\chi$): each matrix element $\braket{j|U^\dagger HU|j'}$ is a product-state sandwich of an MPO of bond $\mathcal O(M\chi^2)$, computable in $\mathrm{poly}(n,\chi,M)$, and there are $k^2$ of them~\cite{Schollwock2011};
(e) any change to another classically describable frame, absorbed into $V$; the Schur/CSF dressing of \sref{sec:schur-csf}, whose compressed Hamiltonian is the classical unitary-group approach~\cite{Paldus1974, Shavitt1977}, is of this type.
\end{proposition}

\textbf{Proof.} In cases (a)--(c) the dressed Hamiltonian $\widetilde H$ satisfies (A3) with the stated term count $M'$, and Theorem~\ref{thm:dequantization} applies verbatim. In case (b) the bounded-degree hypothesis is essential: a generic weight-$w$ Pauli string carries a Jordan--Wigner $Z$-string of Majorana degree $\Theta(n)$, for which the expansion is exponential. In cases (d)--(e) one computes the $k\times k$ compressed matrix $\widetilde H_S$ directly, entrywise, and re-enters the proof of Theorem~\ref{thm:dequantization} after the compression step. A complementary view of (d): a $k$-sparse state is a sum of $k$ product terms, hence has Schmidt rank at most $k$ across every cut, so the bare manifold lies inside the bond-dimension-$k$ matrix-product-state manifold and the dressed loop is a bond-dimension-$k\chi$ tensor-network computation; the real-time integrator of \sref{sec:kahler} is, classically, a tree-structured TDVP on that manifold. $\square$

\begin{corollary}[Necessary conditions for quantum advantage]\label{cor:advantage}
\emergencystretch=2em
No quantum-advantage claim can rest on a sector-sparse core alone, nor on a sector-sparse core composed with a contractible dressing. An architecture of the form $U(\bm\phi)\ket{\psi(\bm\theta,\bm\omega)}$ with an efficiently preparable sparse core is a candidate for advantage only if, simultaneously, (i) no classically describable frame renders the dressed pair $\big(U\ket\psi,\,H\big)$ sector-sparse with computable amplitudes, and (ii) the dressing escapes every contraction of Proposition~\ref{prop:closure}: depth growing with $n$, and non-Clifford, non-Gaussian content admitting no polynomial tensor-network description. In particular, entanglement alone, shallow non-Clifford layers, and Gaussian (matchgate) dressings of local observables do not suffice.
\end{corollary}

\paragraph*{Scope, and relation to prior simulability results.}

Theorem~\ref{thm:dequantization} covers, beyond the tree ansatz, any construction meeting Definition~\ref{def:sector-sparse} with polynomial $k$: sparse-state-preparation circuits used as variational ans\"atze~\cite{Mottonen2005, Li2025, Luo2025}; fixed-Hamming-weight ans\"atze at constant weight $w$ (Dicke-state cores and number-conserving mixers on sectors of size $\binom nw=\mathcal O(n^w)$), for which the fixed-sector special case is known~\cite{Monbroussou2025}; qubit-efficient encodings of selected-CI spaces; and the Schur-dressed spin-adapted loop of \sref{sec:schur-csf}. It does not cover ans\"atze with generically dense sector support, UCCSD, HVA, or hardware-efficient circuits on sectors that grow superpolynomially: the theorem draws a boundary, not a blanket dequantization of variational algorithms. The proof pattern also places the result within an existing body of work, which it sharpens rather than discovers. Assumptions (A1)--(A2) define a subclass of the computationally tractable states of Van den Nest, for which expectation values of sparse observables are classically estimable to additive error~\cite{VandenNest2011, Schwarz2013}; enumerable support upgrades estimation to exact strong simulation, and extends it from single expectation values to the entire metric-aware loop, its geometry (iii)--(iv), its ensembles (vi), and its certificates (vii). The fixed $k$-dimensional invariant subspace likewise realises the polynomial-dimension case of Lie-algebraic and subspace simulation~\cite{Goh2023, Monbroussou2025}, and the conjunction of perfect conditioning (Theorem~\ref{thm:diag-metric}) with exact dequantization is an explicit instance of the observation that the structure enabling trainability also enables classical surrogates~\cite{Cerezo2025sim}.

\paragraph*{Where efficiency is lost.}

Theorem~\ref{thm:dequantization} rests on two finiteness assumptions, and the construction is designed so that a quantum device becomes necessary exactly when either fails.

\textbf{(a) Dense support.} If the target is not sparse in any accessible frame ($k=\Theta(2^n)$ active leaves), then the amplitude list, the compressed evaluation Eq.~\eqref{eq:energy-classical} and the parameter count all become exponential, and the pruned circuit saturates the M\"ott\"onen worst case $N_{\mathrm{cx}}=2(2^n-1)$ of Eq.~\eqref{eq:cnot-bound}. Sparsity in some efficiently describable frame is what keeps the bare loop classical; without it the ansatz is simply a dense (and now exponential-cost) state preparation.

\textbf{(b) A non-contractible dressing.} Compose the bare state with an entangling unitary, $\ket{\Psi}=U(\bm\phi)\ket{\psi(\bm\theta,\bm\omega)}$ as in Eq.~\eqref{eq:dressed-main} of the main text, and the relevant expectation becomes $\braket{\psi|\,\widetilde H(\bm\phi)\,|\psi}$ with the dressed Hamiltonian $\widetilde H(\bm\phi)=U^\dagger(\bm\phi)HU(\bm\phi)$. Two things change. First, the dressed state $U\ket\psi$ generically spreads the $k$-sparse support over all $2^n$ basis states, so the sparse evaluation Eq.~\eqref{eq:energy-classical} no longer applies. Second, conjugating $H$ through $U(\bm\phi)$ expands $\widetilde H$ into a number of Pauli strings that grows exponentially with the gate count of $U(\bm\phi)$ in general, with no telescoping. By Proposition~\ref{prop:closure}, entanglement alone is not the criterion: Clifford, Gaussian, constant-depth-local and polynomial-bond dressings all remain contractible, so the dressing must defeat all of these structures at once (Corollary~\ref{cor:advantage}). When $U(\bm\phi)$ is such a circuit (non-Clifford and non-Gaussian, of depth growing with $n$, admitting no efficient stabiliser or low-bond-dimension tensor-network description), neither $U\ket\psi$ nor $\widetilde H$ has a known polynomial classical representation. Estimating the per-configuration energy $\braket{\psi|\widetilde H|\psi}$, or, in the sampling applications of Appendix~\ref{app:fs-protocols}, the per-sample response $\braket{\psi(p)|U_\star^\dagger O U_\star|\psi(p)}$, is then believed classically intractable, while it remains a single expectation-value measurement on hardware that prepares Eq.~\eqref{eq:dressed-main} directly.

The bare ansatz is thus best read as an enabling primitive: an efficient, exactly controllable, well-conditioned preparation of states (and, through Appendix~\ref{app:fs-design}, of Fubini--Study sample ensembles). Its composition with a hard $U(\bm\phi)$ is where a quantum processor performs a computation that the classical primitive alone cannot: the two-block dressing of \sref{sec:two-block}, the Schrieffer--Wolff dressing of \sref{sec:decouple}, and the process-benchmarking and infinite-temperature protocols of \sref{sec:fs-apps}.

\section{Exact sector-Haar sampling from the diagonal metric}\label{app:fs-design}

Several of the applications in \sref{sec:fs-apps} require not a single prepared state but an unbiased ensemble of states covering a symmetry sector. The tree ansatz supplies this directly, and the diagonal metric makes the draw exact and closed-form: sampling its parameters from the Fubini--Study volume element reproduces the unitary-invariant (Haar) ensemble on the sector through a rejection-free, per-node construction.

\paragraph*{Statement.}

Fix the active-leaf sector $S$ with $k=|S|$, and the binary tree spanning $S$ with its $k-1$ active internal nodes. Use the complex chart in the canonical gauge $\omega_0=0$, so the free coordinates are the $k-1$ branch angles $\bm\theta$ and the $k-1$ relative leaf phases $\omega_1,\dots,\omega_{k-1}$, together a chart on the sector projective space $\mathbb{C}\mathbf P^{k-1}$. Each active node $a$ splits the active leaves beneath it into a cosine-child subtree with $n^a_L$ of them and a sine-child subtree with $n^a_R$.

\begin{proposition}\label{prop:fs-haar}
The Fubini--Study volume measure $d\mu_{\mathrm{FS}}\propto\sqrt{\det g^{\mathrm{FS}}}\,d\bm\theta\,d\bm\omega$ on this chart coincides with the independent draws
\begin{equation}\label{eq:fs-measure}
t_a\sim\mathrm{Beta}(n^a_R,n^a_L),\quad \theta_a=\arcsin\sqrt{t_a},\qquad \omega_1,\dots,\omega_{k-1}\sim\mathrm{Unif}[0,2\pi),
\end{equation}
and the states $\ket{\psi(\bm\theta,\bm\omega)}$ it produces are distributed exactly as the sector-Haar (uniform Fubini--Study) ensemble on $\mathbb{C}\mathbf P^{k-1}$, i.e.\ $\ket\psi\stackrel{d}{=}V\ket{\psi_0}$ with $V$ Haar on $\mathrm U(k)$ restricted to $S$. Each sample costs $\mathcal O(k)$ classical operations, with no determinant evaluation and no rejection.
\end{proposition}

\paragraph*{Proof.}

\textbf{The target law.} A uniform pure state on $\mathbb{C}\mathbf P^{k-1}$ is, in any fixed orthonormal basis, a normalised isotropic complex Gaussian vector. Its squared moduli $(|c_j|^2)_{j\in S}$ are therefore flat Dirichlet, $(|c_j|^2)\sim\mathrm{Dir}(1,\dots,1)$ on the probability simplex, and the $k-1$ relative phases are independent, uniform, and independent of the moduli. Sector-Haar sampling is thus exactly: draw $(|c_j|^2)\sim\mathrm{Dir}(1,\dots,1)$ and attach uniform relative phases.

\textbf{Dirichlet aggregation reproduces the tree.} By Eq.~\eqref{eq:c2p-ry} the squared leaf amplitude $p_j=\prod_{m\in\mathrm{anc}(j)}f_m(\theta_m)^2$ is built by routing unit mass down the tree, each node $a$ sending a fraction $t_a=\sin^2\theta_a$ into its sine child and $1-t_a$ into its cosine child. The Dirichlet grouping (aggregation) property states that if $(|c_j|^2)\sim\mathrm{Dir}(1,\dots,1)$ then, for any partition of the leaves into a cosine group of $n^a_L$ and a sine group of $n^a_R$, the total sine-group mass is $\mathrm{Beta}(n^a_R,n^a_L)$, and, conditioned on that split, each group is again flat Dirichlet and independent of the other. Applied recursively down the tree, this makes the per-node fractions $t_a$ independent $\mathrm{Beta}(n^a_R,n^a_L)$; conversely, drawing them so and routing the mass reconstructs $\mathrm{Dir}(1,\dots,1)$ exactly. With $\theta_a=\arcsin\sqrt{t_a}$ and uniform relative phases this is precisely Eq.~\eqref{eq:fs-measure}, which therefore samples the sector-Haar ensemble.

\textbf{It is the metric volume.} In the canonical gauge, set $q=(p_1,\ldots,p_{k-1})^\top$. The phase block of $g^{\mathrm{FS}}$ is $\operatorname{diag}(q)-qq^\top$ (Appendix~\ref{app:ng-omega-equiv}). On the open chart, the determinant lemma gives
\[
\det\bigl[\operatorname{diag}(q)-qq^\top\bigr]
=\left(1-\sum_{j=1}^{k-1}p_j\right)\prod_{j=1}^{k-1}p_j
=\prod_{j=0}^{k-1}p_j.
\]
Together with the amplitude block $\operatorname{diag}(w_i)$, this yields $\det g^{\mathrm{FS}}=\prod_i w_i\prod_{j=0}^{k-1}p_j$. Collecting the dependence on a single angle $\theta_a$, the factor $\cos^2\theta_a$ occurs once for each node lying in its cosine subtree and $\sin^2\theta_a$ once for each node in its sine subtree; since a full binary subtree with $m$ leaves has $2m-1$ nodes, the resulting exponents are $2n^a_L-1$ and $2n^a_R-1$, so
\begin{equation}\label{eq:detg-factorise}
\sqrt{\det g^{\mathrm{FS}}}\;=\;\prod_a (\cos\theta_a)^{2n^a_L-1}(\sin\theta_a)^{2n^a_R-1}.
\end{equation}
This separates over the angles and is independent of the phases. The marginal $\propto(\cos\theta_a)^{2n^a_L-1}(\sin\theta_a)^{2n^a_R-1}d\theta_a$ is exactly the density of $\theta_a=\arcsin\sqrt{t_a}$ with $t_a\sim\mathrm{Beta}(n^a_R,n^a_L)$ (change variables $t_a=\sin^2\theta_a$), and the phases, absent from $\det g^{\mathrm{FS}}$, are uniform. Hence $\sqrt{\det g^{\mathrm{FS}}}\,d\bm\theta\,d\bm\omega$ equals Eq.~\eqref{eq:fs-measure}, completing the proof. $\square$

\textbf{Independent check.} The tree sampler is cross-validated against the Gaussian oracle $z\sim\mathcal{CN}(0,1)^k,\ \ket\psi=z/\lVert z\rVert$ scattered on $S$, exact sector-Haar by rotational invariance of the isotropic complex Gaussian, through first-/second-moment and Kolmogorov--Smirnov tests in the reference implementation.

\paragraph*{Design corollaries.}

Being the Haar measure on sector pure states, $\mu_{\mathrm{FS}}$ is an exact state $t$-design for every $t$. Two moments are used in the applications.

\textbf{One-design (trace estimation).} With $P_S$ the projector onto $\mathrm{span}\,S$,
\begin{equation}\label{eq:fs-1design}
\mathbb E_{\psi\sim\mathrm{FS}(S)}\big[\ket\psi\!\bra\psi\big]=\frac{P_S}{k},
\end{equation}
so that $\mathbb E_{\psi}\big[\braket{\psi|W|\psi}\big]=\operatorname{Tr}(P_S W)/k$ for any observable $W$: the FS average of any expectation value is an unbiased estimator of the sector trace. This is the basis of the infinite-temperature estimator of Appendix~\ref{app:fs-protocols} ($\mathbb E_\psi\ket\psi\!\bra\psi$ is the maximally mixed sector state), and typicality suppresses the single-sample fluctuation to $\mathcal O(1/k)$.

\textbf{Two-design (average fidelity).} The second moment is
\begin{equation}\label{eq:fs-2design}
\mathbb E_{\psi}\big[\ket\psi\!\bra\psi^{\otimes2}\big]=\frac{P_S\otimes P_S\,(\mathbb I+\mathrm{SWAP})}{k(k+1)} ,
\end{equation}
which gives, for any operator $M$ acting on the sector, the closed form
\begin{equation}\label{eq:Favg}
\mathbb E_{\psi}\big[\,|\braket{\psi|M|\psi}|^2\,\big]=\frac{|\operatorname{Tr}M|^2+\operatorname{Tr}(M^\dagger M)}{k(k+1)} .
\end{equation}
Equation~\eqref{eq:Favg} with $M=P_S\,U_\star^\dagger U_{\mathrm{exact}}\,P_S$ is the exact sector-averaged process fidelity that the echo benchmark of Appendix~\ref{app:fs-protocols} measures.

The distinction from the standard practice of averaging over computational-basis states matters: the basis-state ensemble is only a $1$-design (it reproduces Eq.~\eqref{eq:fs-1design} but not Eq.~\eqref{eq:fs-2design}), so it is unbiased for linear trace quantities but biased for the quadratic, state-dependent quantities (coherent process error, anisotropic noise) that the second moment controls. The FS sampler is the cheapest exact $2$-design on the sector, and the bias it removes is exactly the off-diagonal coherence that basis states miss.

\section{Measurement protocols for the Fubini--Study sampling applications}\label{app:fs-protocols}

This appendix specifies the two device protocols summarised in \sref{sec:fs-apps}: the echo-fidelity process benchmark (Fig.~\ref{fig:haar-echo}) and the infinite-temperature dynamical-correlator estimator (Fig.~\ref{fig:thermal-fs}). Both consume FS samples from Appendix~\ref{app:fs-design} and pass them through a classically hard unitary, so the sampling is the enabling primitive and the hard evolution is what makes the per-sample quantity quantum.

\paragraph*{Echo fidelity (process benchmarking).}

To benchmark a process $U_\star$ against an ideal target $U_{\mathrm{exact}}$ over a sector $S$, draw $\ket{\psi(p)}=A(p)\ket{0^n}$ from $\mathrm{FS}(S)$, apply the device process, the ideal inverse, and the state-preparation inverse, and measure the return probability
\begin{equation}\label{eq:echo}
f(p)=\Pr(0^n)=\big|\braket{\psi(p)|U_{\mathrm{exact}}^\dagger U_\star|\psi(p)}\big|^2 .
\end{equation}
The FS average $\overline{f}$ estimates the sector-averaged process fidelity, whose exact value is given in closed form by the two-design identity Eq.~\eqref{eq:Favg} with $M=P_S U_{\mathrm{exact}}^\dagger U_\star P_S$; this serves as ground truth alongside the sampled estimate. In Fig.~\ref{fig:haar-echo} the process is $U_\star=(\text{second-order Trotter step at }dt)^{T/dt}$ and $U_{\mathrm{exact}}=e^{-iHT}$, so $\overline{1-f}$ is the state-dependent coherent Trotter error averaged uniformly over the sector. In the noiseless study the ideal inverse is applied exactly; on hardware $U_{\mathrm{exact}}^\dagger$ is instantiated as the time-reversed Trotter circuit at an $r$-fold finer step, the same palindromic second-order construction with the gate order reversed and the angles negated, whose own coherent error enters the echo at relative order $1/r^2$ and is negligible against the coarse-step error under study. Two comparison probes are shown. Averaging over computational sector-basis states is biased because the Trotter error is state-dependent and the basis ensemble is only a one-design (Appendix~\ref{app:fs-design}), so it misses the off-diagonal coherent error. The single (ground) state echo is unrepresentative for the opposite reason: a near-eigenstate accumulates only a global phase that the echo cancels, so it reports almost no error and is the least informative probe of the averaged process.

\paragraph*{Infinite-temperature dynamical correlators.}

On a sector the infinite-temperature state is the maximally mixed $P_S/k$, which the FS one-design Eq.~\eqref{eq:fs-1design} represents as a pure-state Monte-Carlo ensemble. The target observables are dynamical, since static infinite-temperature quantities are typically featureless:
\begin{equation}\label{eq:CAB}
C_{AB}(t)=\frac1K\operatorname{Tr}\!\big[P_S\,A(t)\,B\,P_S\big],\qquad A(t)=U^\dagger(t)\,A\,U(t).
\end{equation}
For each FS sample the per-state quantity $\braket{\psi(p)|A(t)B|\psi(p)}=\braket{\psi(p)|U^\dagger(t)\,A\,U(t)\,B|\psi(p)}$ is measured by an ancilla Hadamard test: prepare the ancilla in $\ket{+}$, apply a controlled $B$, evolve the system by $U(t)$, and jointly measure $A$ on the system register together with the ancilla in $X$ (real part) and $Y$ (imaginary part). On the post-evolution state $\tfrac{1}{\sqrt2}(\ket{0}U\ket\psi+\ket{1}UB\ket\psi)$ the joint observables give $\braket{A\otimes X_a}=\operatorname{Re}\braket{\psi|U^\dagger A U B|\psi}$ and $\braket{A\otimes Y_a}=\operatorname{Im}\braket{\psi|U^\dagger A U B|\psi}$. The FS average over samples then reproduces Eq.~\eqref{eq:CAB} with single-sample variance suppressed by typicality as $\mathcal O(1/k)$. (For Haar samples a controlled $B$ with an ancilla is required because the samples do not commute with the eigenprojectors of $B$, so sequential projective measurement is invalid.) In Fig.~\ref{fig:thermal-fs}, $U(t)$ is the term-level second-order Trotter circuit of the Hubbard dynamics and $A,B$ are local $S^z$ (spin) or $\delta n$ (charge) operators; the conjugated operator $U^\dagger(t)AU(t)$ becomes increasingly nonlocal and classically expensive to represent, while the device merely evolves the sampled state and measures the ancilla. The exact-propagator correlator is overlaid as the target, and basis-state averaging, also unbiased but with $\mathcal O(1)$ per-sample variance, is shown for comparison.

\section{The FS-sampled Schrieffer--Wolff gradient surrogate}\label{app:sw-surrogate}

The two-block dressing of \sref{sec:two-block} and the exact decoupling of \sref{sec:decouple} are trained by the same problem-agnostic objective: an energy-only surrogate for the Schrieffer--Wolff~\cite{Bravyi2011} block-decoupling condition, evaluated by ordinary parameter-shift gradients on FS samples of the model space. The objective, its commutator interpretation, and the condition under which its global minimum is the exact decoupling, are set out below.

\paragraph*{Objective and gradient.}

Let $P$ be the model space ($k_P=|P|$), $Q$ its complement within the sector, and $U(\bm\phi)=\prod_a e^{-i\phi_a A_a}$ a dressing built from symmetry-adapted $P\!\leftrightarrow\!Q$ excitation generators $A_a=i(T_a-T_a^\dagger)$. Each $A_a$ is a partial-isometry difference with $A_a^3=A_a$, so $e^{-i\phi_a A_a}$ is exact and the four-term rule of Eq.~\eqref{eq:shift-rule} applies unchanged. Exact block decoupling is the Schrieffer--Wolff condition
\begin{equation}\label{eq:sw-condition}
Q\,\widetilde H(\bm\phi)\,P=0,\qquad \widetilde H(\bm\phi)=U^\dagger(\bm\phi)HU(\bm\phi),
\end{equation}
whose held-out witness is the leakage $\mathcal L(\bm\phi)=\|Q\widetilde H P\|_F^2/k_P$. Rather than measure $\mathcal L$ (which needs $Q$, the spectrum, and a block encoding of $H$), train the FS-sampled gradient surrogate
\begin{equation}\label{eq:cgrad}
C_{\mathrm{grad}}(\bm\phi)=\mathbb E_{\chi\sim\mathrm{FS}(P)}\Big[\textstyle\sum_a\big|\,\partial_{\phi_a}\braket{\chi|\widetilde H(\bm\phi)|\chi}\big|^2\Big],
\end{equation}
in which $\ket\chi$ is drawn from the FS measure on the model space (Appendix~\ref{app:fs-design}) and each $\partial_{\phi_a}\braket{\chi|\widetilde H|\chi}$ is a parameter-shift energy gradient of the circuit $U(\bm\phi)\,T(\bm\theta_\chi)\ket{0}$. The optimiser never constructs $Q$, $\widetilde H$, or $\mathcal L$.

\paragraph*{Why it approximates Schrieffer--Wolff.}

Each surrogate gradient is a commutator response,
\begin{equation}\label{eq:commutator-response}
g_a(\chi)=\partial_{\phi_a}\braket{\chi|\widetilde H|\chi}=\braket{\chi|[\widetilde H,B_a]|\chi},\qquad B_a=U^\dagger\partial_{\phi_a}U ,
\end{equation}
i.e.\ the first-order energy response along the tangent direction $B_a$ of the dressing. $C_{\mathrm{grad}}$ is an average of the non-negative quantities $|g_a(\chi)|^2$ over an ensemble of full support on $P$; since each $[\widetilde H,B_a]$ is Hermitian, $\braket{\chi|[\widetilde H,B_a]|\chi}=0$ for all $\chi\in P$ is equivalent by polarization to $P\,[\widetilde H,B_a]\,P=0$. Hence $C_{\mathrm{grad}}=0$ if and only if the dressed energy is stationary throughout $P$ under every pool direction, $P\,[\widetilde H,B_a]\,P=0$ for every $a$. If the generators $\{A_a\}$ span the full off-block tangent space, all $P\!\leftrightarrow\!Q$ rotations, then this stationarity is equivalent to the Schrieffer--Wolff condition Eq.~\eqref{eq:sw-condition}, and the global minimum of $C_{\mathrm{grad}}$ is exact decoupling. Otherwise $C_{\mathrm{grad}}$ is a projected Schrieffer--Wolff residual: it drives the leakage down only within the span of the pool, and $\mathcal L$ plateaus at the residual outside it (the inhomogeneous-Hubbard regime).

\paragraph*{Tangent-completeness.}

The pool is tangent-complete exactly when no $P$ determinant and no $Q$ determinant differ by more than a double excitation, for then the screened generalised singles-and-doubles exhaust the off-block tangent and
\begin{equation}\label{eq:tangent-complete}
n_\phi=|P|\,|Q| .
\end{equation}
This holds whenever the active space spans at most three spatial orbitals (so $n_\phi=6$ for CAS$(2e,3o)$ and $n_\phi=20$ for CAS$(4e,3o)$), which is why the LiH decoupling of \sref{sec:decouple} reaches machine-zero leakage at every geometry, whereas the half-filled Hubbard pool floors below the full off-block and its leakage plateaus at the higher-order residual.

\section{Two-block dressed-Hamiltonian construction}\label{app:two-block}

This appendix records the construction benchmarked in \sref{sec:two-block}: the factorised ansatz, the external-excitation pool and its parameter-matched screening, and the two-stage optimisation. It instantiates the dressed-state composition Eq.~\eqref{eq:dressed-main} with a state-preparation block that is itself the symmetry-pruned tree.

\paragraph*{The factorised ansatz.}

In the full symmetry-adapted encoding the trial state factorises as
\begin{equation}\label{eq:twoblock-state}
\ket{\Psi(\bm\theta,\bm\phi)}=U_{\mathrm{ext}}(\bm\phi)\,T(\bm\theta)\,\ket{0},
\end{equation}
where $T(\bm\theta)$ is the symmetry-pruned tree of \sref{sec:pruning} with active leaves equal to a complete-active-space (CAS) sector $P$, preparing the exact CAS ground state at fixed cost and with the diagonal metric of Theorem~\ref{thm:diag-metric} intact on its parameters; and $U_{\mathrm{ext}}(\bm\phi)=\prod_k e^{-i\phi_k A_k}$ is an ordered product of symmetry-allowed generalised single- and double-excitation rotations coupling $P$ to its external complement $Q$. Because the $P$-block state is multireference, the generators range over all spin-orbital index combinations (not only Hartree--Fock occupied$\to$virtual), which matters at stretched geometries. Each $A_k=i(T_k-T_k^\dagger)$ satisfies $A_k^3=A_k$, so the exponential is exact and real-orthogonal and the real tree chart suffices. The ordered product is the first-order Trotterisation of the Schrieffer--Wolff dressing $e^{S}$ with purely off-block $S$ (the \texttt{--symmetric} variant is the second-order symmetric product at twice the dressing CNOT cost), so \sref{sec:two-block} and \sref{sec:decouple} are two readings of the same dressed Hamiltonian $\widetilde H=U_{\mathrm{ext}}^\dagger H U_{\mathrm{ext}}$: minimising the dressed energy versus driving its off-block to zero.

\paragraph*{Pool screening and parameter matching.}

The raw pool is screened by dropping null, sector-leaking and purely in-$P$ generators, then pruned by ordered Gram--Schmidt on the off-block parts $Q A_k P$ (singles before doubles, lexicographic), so the reported $\phi$-count equals the off-block tangent dimension. For a fair, equal-parameter comparison with UCCSD the surviving pool is ranked by the device-measurable dressing gradient at $\bm\phi=0$,
\begin{equation}\label{eq:adapt-rank}
\big|\partial_{\phi_k}E\big|_{\bm\phi=0}=2\,\big|\braket{\psi_P|H A_k|\psi_P}\big| ,
\end{equation}
a single round of ADAPT-style selection~\cite{Grimsley2019} charged at shift-rule cost, and only the top $M_{\mathrm{UCCSD}}-n_{\mathrm{tree}}$ generators are retained, so the TwoBlock total parameter count equals that of the UCCSD reference in the identical encoding. Optimisation proceeds in two stages: natural-gradient convergence of the CAS-only tree ($\bm\theta$), then joint refinement (natural gradient on $\bm\theta$, gradient descent on $\bm\phi$ against the dressed Hamiltonian). The undressed baseline (\textsc{CAS-Tree}) is stage one alone.

\paragraph*{Exactly spin-adapted two-block.}

The spin-adapted variant (\textsc{TwoBlock-MM-spin}) composes the gate-free spin-adapted state preparation of Appendix~\ref{app:spin-tree} with a spin-conserving dressing. The state-preparation block is the CSF tree over the singlet block of the CAS sector $P$, built natively in the symmetry-adapted encoding: the intermediate-Jordan--Wigner CSF tree is formed as in Appendix~\ref{app:spin-tree} and each of its Slater determinants is relabelled to its encoded computational-basis correspondent, the only sign being the blocked-to-interleaved reordering parity (the encoding's Clifford being a sign-free $\mathrm{GF}(2)$ map on the sector). The prepared state is a pure singlet at every $\bm\theta$ ($\langle S^2\rangle=0$), the CAS tree carries $f_S-1$ free angles ($17$ for the water CAS$(6,5)$ singlet block, against $k-1=27$ for the determinant sector), and the exactly diagonal metric of Theorem~\ref{thm:diag-metric} is intact by Proposition~\ref{prop:invariance}.

The dressing generators are the spin-conserving analogues of the generalised singles and doubles, built from the spin-summed one-body operators $E_{pq}=a^\dagger_{p\uparrow}a_{q\uparrow}+a^\dagger_{p\downarrow}a_{q\downarrow}$: the singlet singles $E_{pq}-E_{qp}$ and the generalised singlet doubles $E_{pq}E_{rs}-\mathrm{h.c.}$, each commuting with $S^2$ (verified $\|[A_a,S^2]\|\le10^{-9}$ for all $392$ water generators), so $U(\bm\phi)$ preserves total spin and $U(\bm\phi)\ket\psi$ stays an exact singlet throughout. Unlike the spin-orbital generators of the determinant construction, a singlet double is a genuine two-body operator whose Jordan--Wigner image is a sum of non-commuting Pauli terms and therefore does not satisfy $A^3=A$; in the simulation each rotation $e^{-i\phi_a A_a}$ is applied exactly by an eigendecomposition on the $k$-dimensional sector, so the dressed $\langle S^2\rangle$ is machine zero ($\le5\times10^{-18}$). At the gate level a singlet double is realised by a single first-order Trotter step: although that step is not exact as an operator (its Pauli terms do not all commute -- $128$ of the $\binom{32}{2}$ pairs anticommute for a generic double), the error at the small dressing angles reached by the optimiser ($|\bm\phi|_\infty\approx10^{-2}$) is negligible, the reps-one circuit reproducing the exact dressed state to $\Delta E\approx1.6\times10^{-10}$\,Ha, residual $\langle S^2\rangle\approx4\times10^{-10}$ and sector leakage $\approx10^{-15}$; the quoted two-qubit counts thus correspond to a faithfully spin-pure circuit with no Trotter-repetition penalty.

Because a singlet double costs roughly twice the CNOTs of a spin-orbital double, the spin-adapted pool is screened by a cost-aware ADAPT criterion: generators are ranked by $|\partial_{\phi_a}E|/\mathrm{CNOT}(A_a)$ and admitted until the cumulative dressing CNOT count reaches that of the determinant two-block (evaluated from the same $S_z$-only pool at the CAS ground), so the comparison is drawn at matched two-qubit cost rather than matched parameter count. This ranking fills the budget with the cheap, gate-exact singlet singles and the cost-efficient doubles; a plain-gradient ranking instead spends the budget on a few expensive high-gradient doubles and underperforms. At CNOT count equal to or below the determinant two-block ($1140/1048/776/1024$ against $1196/1108/896/1048$ at $R=1.5/1.75/2.0/2.25\,\text{\AA}$) the spin-adapted ansatz reaches chemical accuracy at every geometry ($1.6$ and $1.5\times10^{-4}$\,Ha near equilibrium, improving to $8.4$ and $2.4\times10^{-5}$\,Ha at $R=2.0$ and $2.25\,\text{\AA}$), more accurately than the determinant two-block in the strongly-correlated regime and with $\langle S^2\rangle$ at machine zero throughout.

\section{Absence of barren plateaus on a polynomial sector}\label{app:no-bp}

Two properties established above fix the trainability of the metric-aware loop with no further assumption: the pruned chart is \emph{complete} on $\mathrm{span}(S)$ (\sref{sec:pruning}), and, drawn from the Fubini--Study measure, its state is an \emph{exact $2$-design} on the sector (Appendix~\ref{app:fs-design}, Eq.~\eqref{eq:fs-2design}). We show they imply that the natural-gradient signal driving the optimiser is the quantum variance of the sector-projected observable, and is therefore only polynomially small whenever the sector dimension $k=|S|$ is polynomial in $n$.

Throughout $O=O^\dagger$ is the objective operator ($O=H$ for VQE), $P_S$ the projector onto $\mathrm{span}(S)$, and $O_S:=P_SOP_S$ its $k\times k$ sector block (the $H_S$ of \sref{sec:mol-dynamics} when $O=H$). Write $\mathrm{Var}_\psi(A):=\braket{\psi|A^2|\psi}-\braket{\psi|A|\psi}^2$ for the quantum variance in $\ket\psi$, and $\mathrm{Var}_{\mathrm{spec}}(A):=\tfrac1k\mathrm{Tr}(A^2)-\big(\tfrac1k\mathrm{Tr}A\big)^2$ for the variance of the eigenvalue spectrum of a $k\times k$ operator $A$. The energy is $E=\braket{\psi|O|\psi}$, and $g$ the diagonal metric of Theorem~\ref{thm:diag-metric}; we write $\|\nabla_g E\|_g^2:=dE^\top g^{-1}dE=\sum_\mu(\partial_\mu E)^2/g_{\mu\mu}$ for the squared natural-gradient norm, with $\nabla_g E=g^{-1}dE$ as in \sref{sec:kahler}.

\begin{lemma}[Gradient--variance identity]\label{lem:grad-var}
If the chart is complete on $\mathrm{span}(S)$, then at every non-singular chart point
\begin{equation}
\|\nabla_g E\|_g^2 \;=\; 4\,\mathrm{Var}_\psi(O_S).
\end{equation}
\end{lemma}

\noindent\textbf{Proof.} By completeness the reachable tangent directions at $\ket\psi$ are, away from the measure-zero chart singularities of \sref{sec:gauge-singular}, all $\ket\delta\in\mathrm{span}(S)$ with $\braket{\psi|\delta}=0$ (the chart's tangent space, cf.\ the projection identity of Appendix~\ref{app:restricted-subspace}), with metric $g(\delta,\delta)=\braket{\delta|\delta}$ (the chart metric of \sref{sec:metric}, i.e.\ $\tfrac14$ of the quantum Fisher information; the global-phase direction carries no gradient of a physical objective, \sref{sec:gauge-singular}). Since $dE[\delta]=2\,\mathrm{Re}\braket{\delta|O|\psi}$,
\begin{equation}
\|\nabla_g E\|_g^2=\sup_{\|\delta\|=1}dE[\delta]^2=4\,\big\|(O_S-\langle O_S\rangle)\ket\psi\big\|^2=4\,\mathrm{Var}_\psi(O_S),
\end{equation}
the supremum attained at $\ket\delta\propto(O_S-\langle O_S\rangle)\ket\psi$, the component of $O\ket\psi$ orthogonal to $\ket\psi$ inside $\mathrm{span}(S)$. \hfill$\square$

\begin{lemma}[Sector average]\label{lem:sector-avg}
Under Fubini--Study initialisation, for which $\ket\psi$ is Haar-distributed on the unit sphere of $\mathrm{span}(S)$,
\begin{equation}
\mathbb E\,\|\nabla_g E\|_g^2=\frac{4k}{k+1}\,\mathrm{Var}_{\mathrm{spec}}(O_S),
\qquad
\mathrm{Var}[E]=\frac{\mathrm{Var}_{\mathrm{spec}}(O_S)}{k+1}.
\end{equation}
\end{lemma}

\noindent\textbf{Proof.} For $\ket\psi$ Haar on the unit sphere of $\mathrm{span}(S)$, the symmetric-subspace identity~\cite{Harrow2013, Dankert2009} gives $\mathbb E\braket{\psi|O_S|\psi}=\tfrac1k\mathrm{Tr}(O_S)$, $\mathbb E\braket{\psi|O_S^2|\psi}=\tfrac1k\mathrm{Tr}(O_S^2)$ and $\mathbb E\braket{\psi|O_S|\psi}^2=\big[(\mathrm{Tr}O_S)^2+\mathrm{Tr}(O_S^2)\big]/[k(k+1)]$, the sector-Haar moments of Theorem~\ref{thm:dequantization}(vi). Hence $\mathbb E\,\mathrm{Var}_\psi(O_S)=\tfrac1k\mathrm{Tr}(O_S^2)-\big[(\mathrm{Tr}O_S)^2+\mathrm{Tr}(O_S^2)\big]/[k(k+1)]=\tfrac{k}{k+1}\mathrm{Var}_{\mathrm{spec}}(O_S)$; Lemma~\ref{lem:grad-var} then gives the first identity and the same moments the second. \hfill$\square$

\begin{proposition}[No barren plateau on a polynomial sector]\label{prop:no-bp}
\emergencystretch=2em
If $k=|S|=\mathrm{poly}(n)$ and $\mathrm{Var}_{\mathrm{spec}}(O_S)=\Omega(1/\mathrm{poly}(n))$, then by Lemma~\ref{lem:sector-avg} $\mathbb E\,\|\nabla_g E\|_g^2=\tfrac{4k}{k+1}\mathrm{Var}_{\mathrm{spec}}(O_S)$ and $\mathrm{Var}[E]=\mathrm{Var}_{\mathrm{spec}}(O_S)/(k+1)$; both are bounded below by an inverse polynomial in $n$, not exponentially small, so the bare ansatz exhibits no barren plateau. The exponential suppression of a fully expressive circuit is recovered from the same identities on replacing the sector by the whole space, $k\to2^n$.
\end{proposition}

\noindent\textbf{Estimation cost.} The optimiser forms $\nabla_g E$ from the parameter-shift derivatives and the metric weights $(w_i,p_j)$, both measured directly as probabilities. Under Fubini--Study initialisation these weights are inverse-polynomial in $k$ rather than exponentially small (each active node carries at least two of the $k$ leaves, and the leaf probabilities are flat Dirichlet with mean $1/k$), so the natural gradient is estimable at polynomial sampling cost.

\noindent\textbf{Dressed core.} For the two-block ansatz $\ket{\Psi}=U(\bm\phi)\ket\psi$ (\sref{sec:two-block}), differentiating in the core parameters at fixed $\bm\phi$ replaces $O$ by the dressed Hamiltonian $\widetilde H=U^\dagger H U$ (Proposition~\ref{prop:invariance}), so Lemma~\ref{lem:grad-var} gives $\|\nabla_{\bm\theta}E\|_g^2=4\,\mathrm{Var}_\psi(\widetilde H_S)$ with $\widetilde H_S=P_S U^\dagger H U P_S$ and $\|\widetilde H_S\|\le\|H\|$; Proposition~\ref{prop:no-bp} therefore holds verbatim for the state-preparation block at every dressing. The dressing parameters $\bm\phi$ are not covered: $U(\bm\phi)$ generically maps out of the sector and, in the advantage regime (Corollary~\ref{cor:advantage}), is deep enough to be classically hard and hence liable to its own barren plateau. The split accordingly removes the plateau from state preparation and confines any residual one to the dressing, which is kept shallow and structured (ADAPT-screened, warm-started at the exact sector state).

\noindent The result is stated for the noiseless landscape and requires the effective spectrum $\mathrm{Var}_{\mathrm{spec}}(O_S)$ not to collapse, which holds generically and, in the applications here, is the spread of the physical effective Hamiltonian on the sector.

\SupplementaryStandaloneEnd

\end{document}